\documentclass[12pt]{article}

\usepackage{amsmath,amssymb,fullpage,graphicx}
\usepackage[round]{natbib}
\usepackage{color}
\usepackage{epstopdf}
\usepackage{wasysym}
\usepackage{wrapfig}
\usepackage{caption}
\usepackage[charter]{mathdesign}
\usepackage[scaled=.95]{inconsolata}
\usepackage[font=small,labelfont=bf,labelsep=period]{caption}
\usepackage{titlesec}

\titleformat*{\section}{\normalsize\bfseries\sffamily}
\titleformat*{\subsection}{\normalsize\bfseries\sffamily}

\newtheorem{theorem}{Theorem}
\newtheorem{lemma}[theorem]{Lemma}

\newenvironment{proof}{{\bf Proof.}}{$\Box$}

\newcommand{\R}{\mbox{$\mathbb{R}$}}

\renewcommand{\Re}{{{\rm I\!R}}}

\newenvironment{enum}{
\begin{enumerate}
  \setlength{\itemsep}{1pt}
  \setlength{\parskip}{0pt}
  \setlength{\parsep}{0pt}
}{\end{enumerate}}

\let\hat\widehat

\begin{document}

\begin{center}
{\bf\Large FINDING SINGULAR FEATURES}\\
\vspace{.5cm}
{\bf Christopher Genovese, Marco Perone-Pacifico,\\
Isabella Verdinelli and Larry Wasserman}\\
Carnegie Mellon University and University of Rome
\end{center}

\begin{center}
April 22 2016
\end{center}

\begin{quote}
{\em We present a method for finding 
high density, low-dimensional structures in noisy point clouds.
These structures are sets with zero 
Lebesgue measure with respect to the $D$-dimensional
ambient space and belong to a $d<D$ dimensional space.
We call them ``singular features.''
Hunting for singular features corresponds 
to finding unexpected or unknown structures hidden
in point clouds belonging to $\R^D$.
Our method outputs well defined
sets of dimensions $d<D$.
Unlike spectral clustering, the method 
works well in the presence of noise. 
We show how to find singular features by
first finding ridges in the estimated density, followed
by a filtering step based on the
eigenvalues of the Hessian of the density.}

\end{quote}


\section{Introduction}

Let $X_1,\ldots, X_n\in \mathbb{R}^D$
be a sample from a distribution $P$
with density $p$. 
We are interested in the question of finding
salient features hidden in the $\R^D$ space.
We are looking for sets of arbitrary shape, with
dimension $d<D$. The approach we take in this paper is
based on an idea that we call {\em singular feature finding.}
The idea is to find sets with high density but 
zero Lebesgue measure (with respect to $\mathbb{R}^D$), that can be of dimension
$d\in \{0,1,\ldots, D-1\}$.

\bigskip
Figure \ref{fig::single}
shows a simple two dimensional dataset.
The top left plot shows the data.
The top middle plot shows the result of standard
single linkage clustering. Although this reveals 
some structure, the output is quite messy and hard
to interpret. The top right plot shows an upper level set 
of the density estimate, $\{\hat p > t\}$. While 
this reveals the modes and ring, the output 
is two dimensional and the modes and ring are not well 
localized. The bottom row shows our singular feature finding
method. The bottom left shows the singular features of 
dimension $d=0$.
The bottom right 
shows the singular feature of dimension $d=1$.

\newpage

\begin{center}
\vspace{-1cm}
\begin{tabular}{ccc}
\includegraphics[scale=.30]{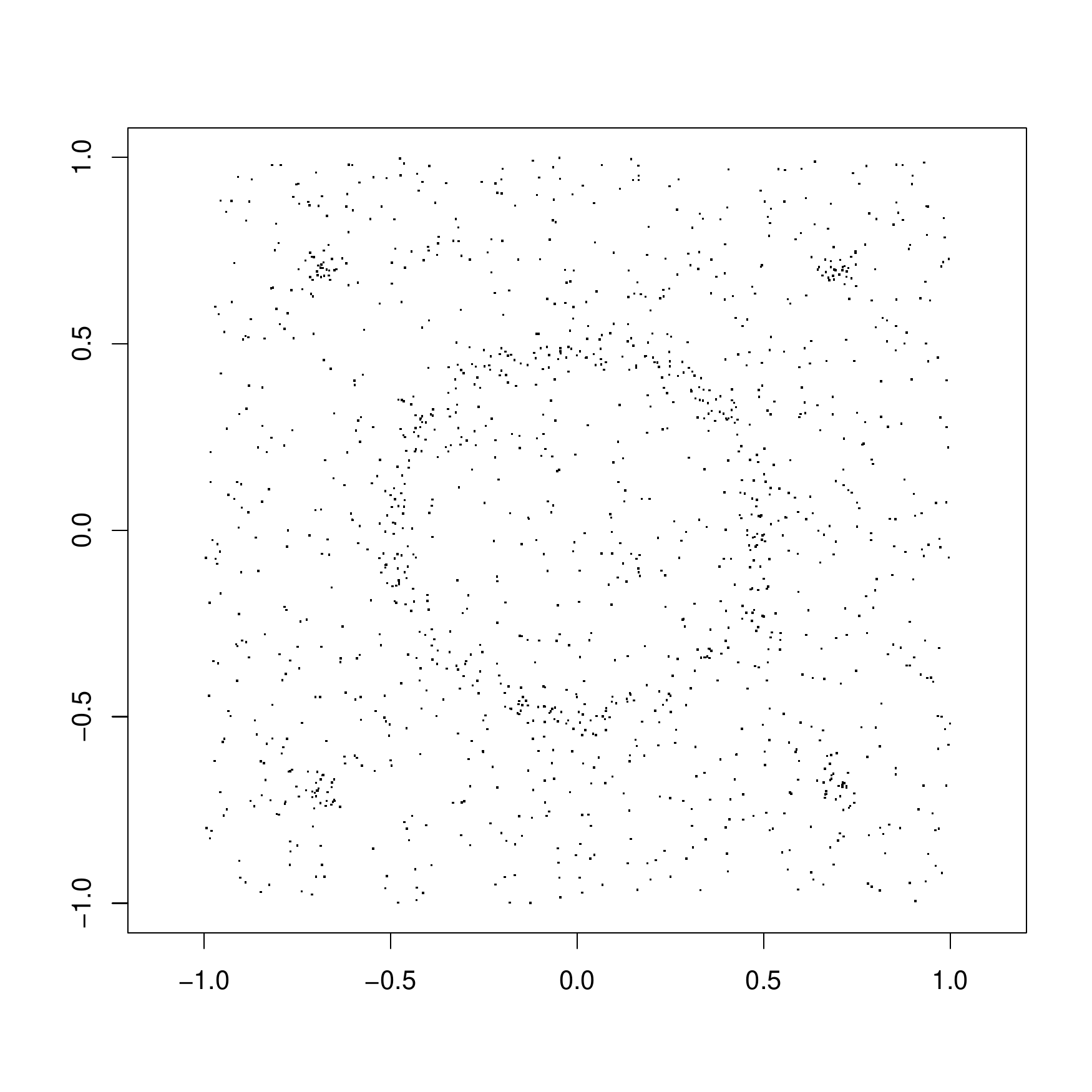} & 
\includegraphics[scale=.30]{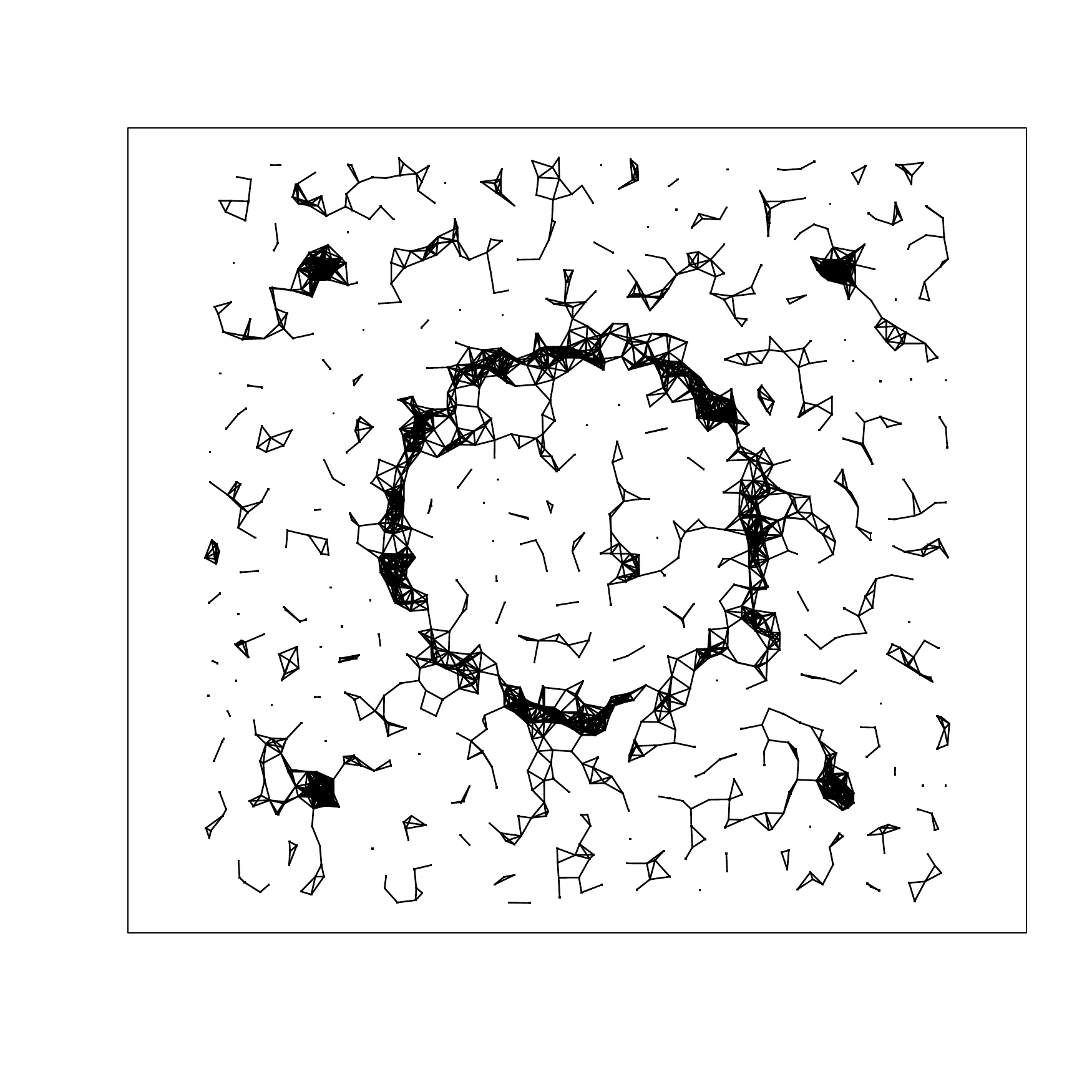} &
\includegraphics[scale=.3]{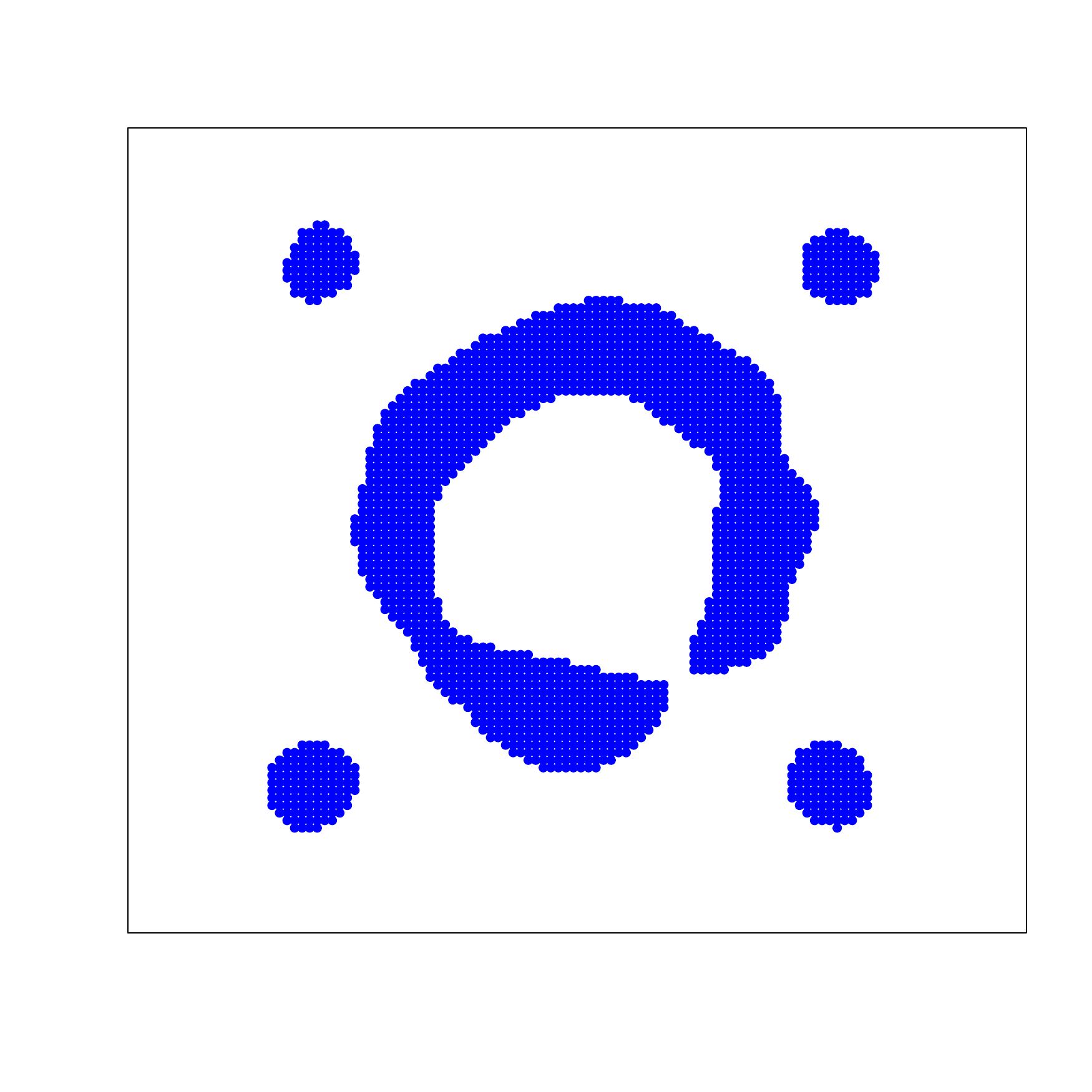}
\end{tabular}
\end{center}

\begin{center}
\vspace{-1.5cm}
\begin{tabular}{cc}
\includegraphics[scale=.30]{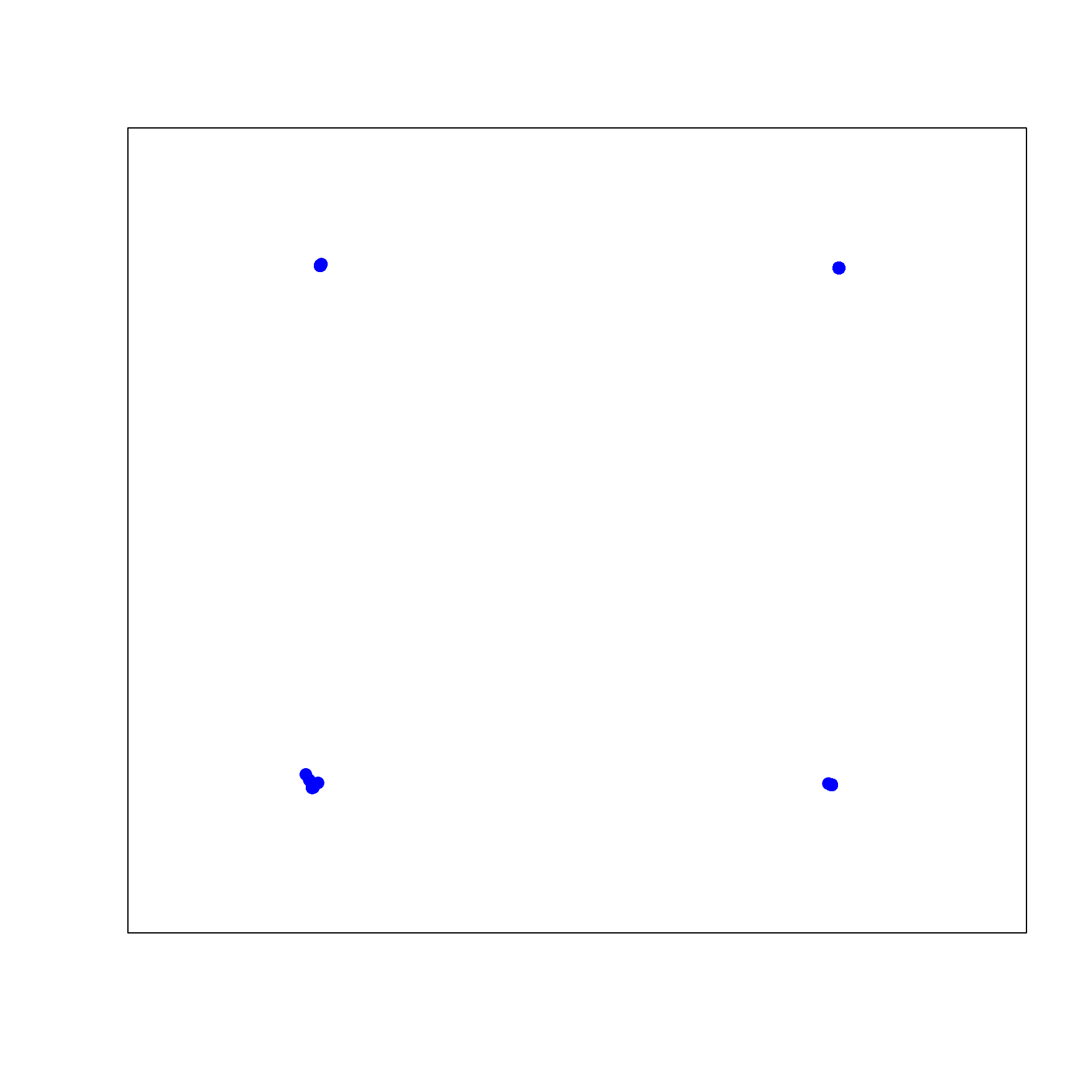} & 
\includegraphics[scale=.30]{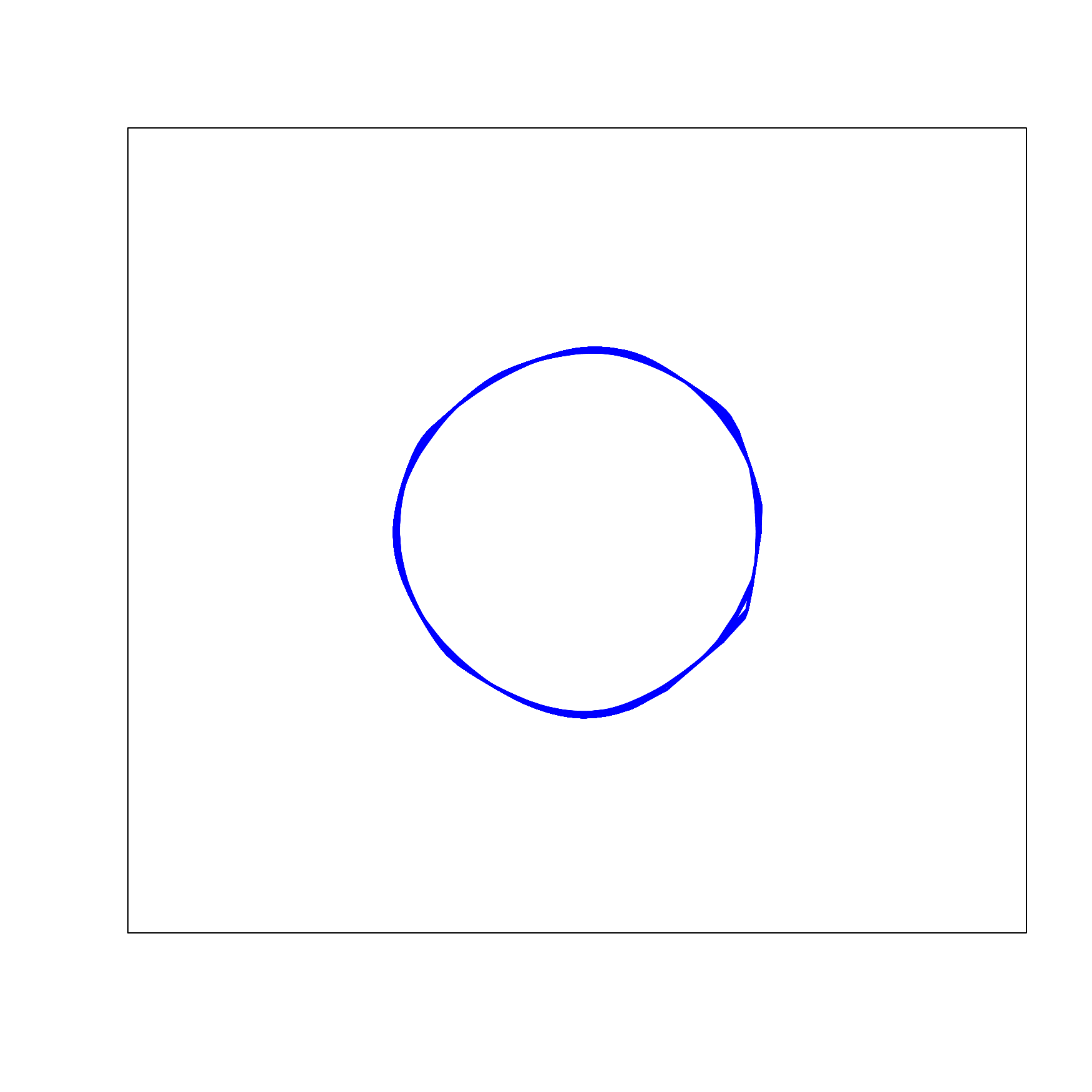}
\end{tabular}
\captionof{figure}{\em Top left: raw data. 
Top middle: single linkage clustering.
Top right: level set of density estimate
Bottom: our method. Left: Dimension $d=0$ features.
Right: Dimension $d=1$ feature.}
\label{fig::single}
\end{center}

The left plot in Figure
\ref{fig::main-example} shows a three-dimensional
dataset.
The right plot shows
high density structures hidden in the three dimensional point cloud.
The structures are zero-dimensional 
(four modes marked  as red triangles),
one-dimensional (the blue ring at the

\medskip

\begin{center}
\begin{tabular}{cc}
\includegraphics[scale=.5]{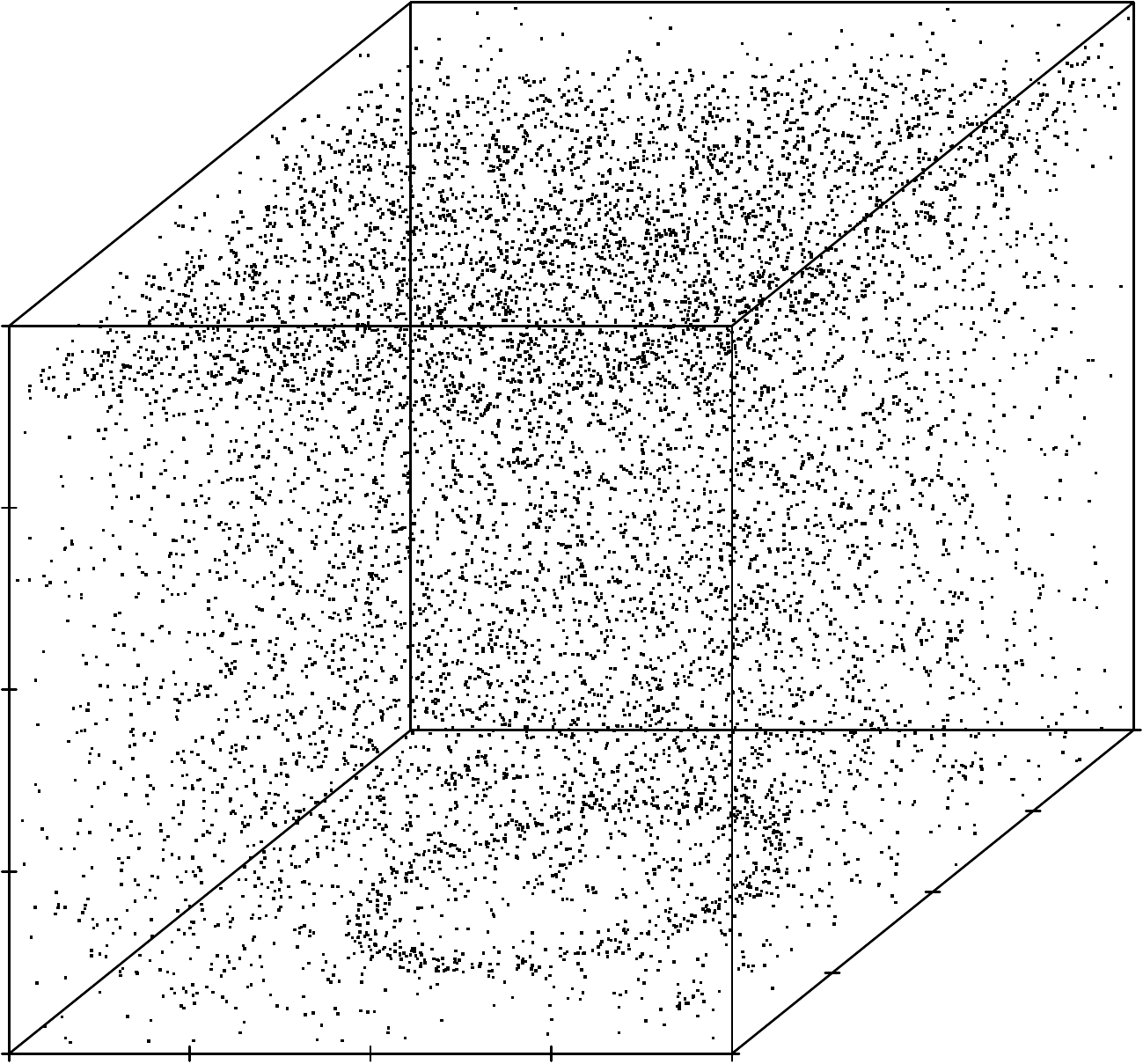} 
\hspace{1em} & 
\includegraphics[scale=.5]{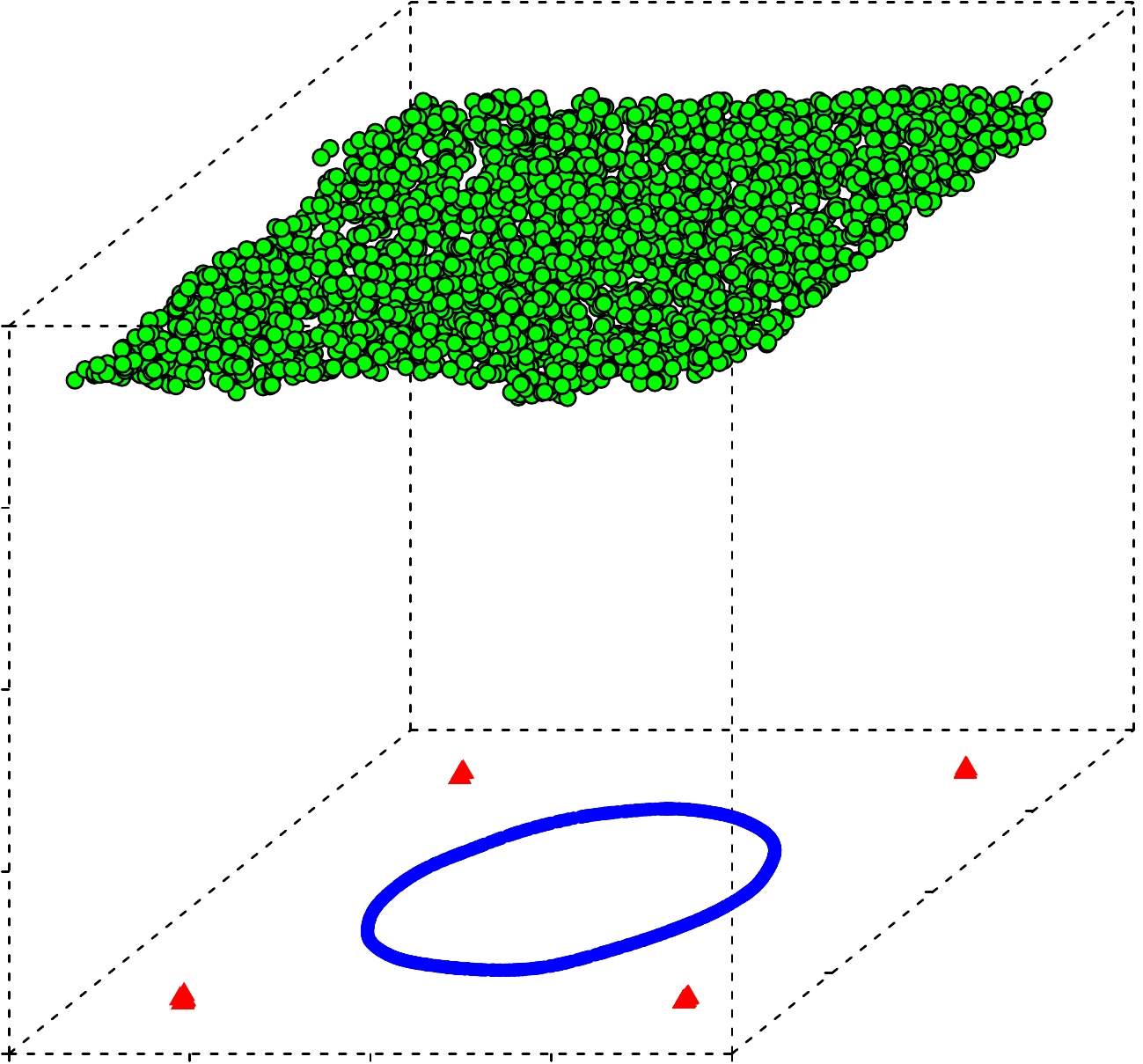} 
\end{tabular}
\captionof{figure}{\em Left: raw data. 
Right: our algorithm extracts four 0-dimensional  
structures (modes), one 1-dimensional structure 
(a ring) and one 2-dimensional structure (plate)
from the point cloud.}
\label{fig::main-example}
\end{center}

bottom) and two-dimensional (the wall at the top).
These structures have high density but,
as subsets of $\mathbb{R}^3$,
they have zero Lebesgue measure.

In this paper we present a method for finding
singular features.
The method has the following steps:
\begin{enum}
\item Estimate the density.
\item Find the estimated {\em ridges} $\hat R_d$ of the density 
(or the log density) of dimension $d=0,1,\ldots, D-1$.
\item Filter out weak ridges: remove points from $x\in \hat R_d$ 
based on the eigenstructure of the Hessian of the density.
The surviving points are denoted by $\hat R_d^\dagger$.
\item Apply single linkage clustering to $\hat R_d^\dagger$
and, optionally, discard small connected components.
\end{enum}

The four steps of our method are summarized in Figure
\ref{fig::flowchart}.
Note that even though $\hat p$ can be complex and unsmooth,
the resulting singular feature is very smooth and simple.

\begin{center}
\vspace{-.9cm}
\begin{tabular}{ccc}
\includegraphics[scale=.3]{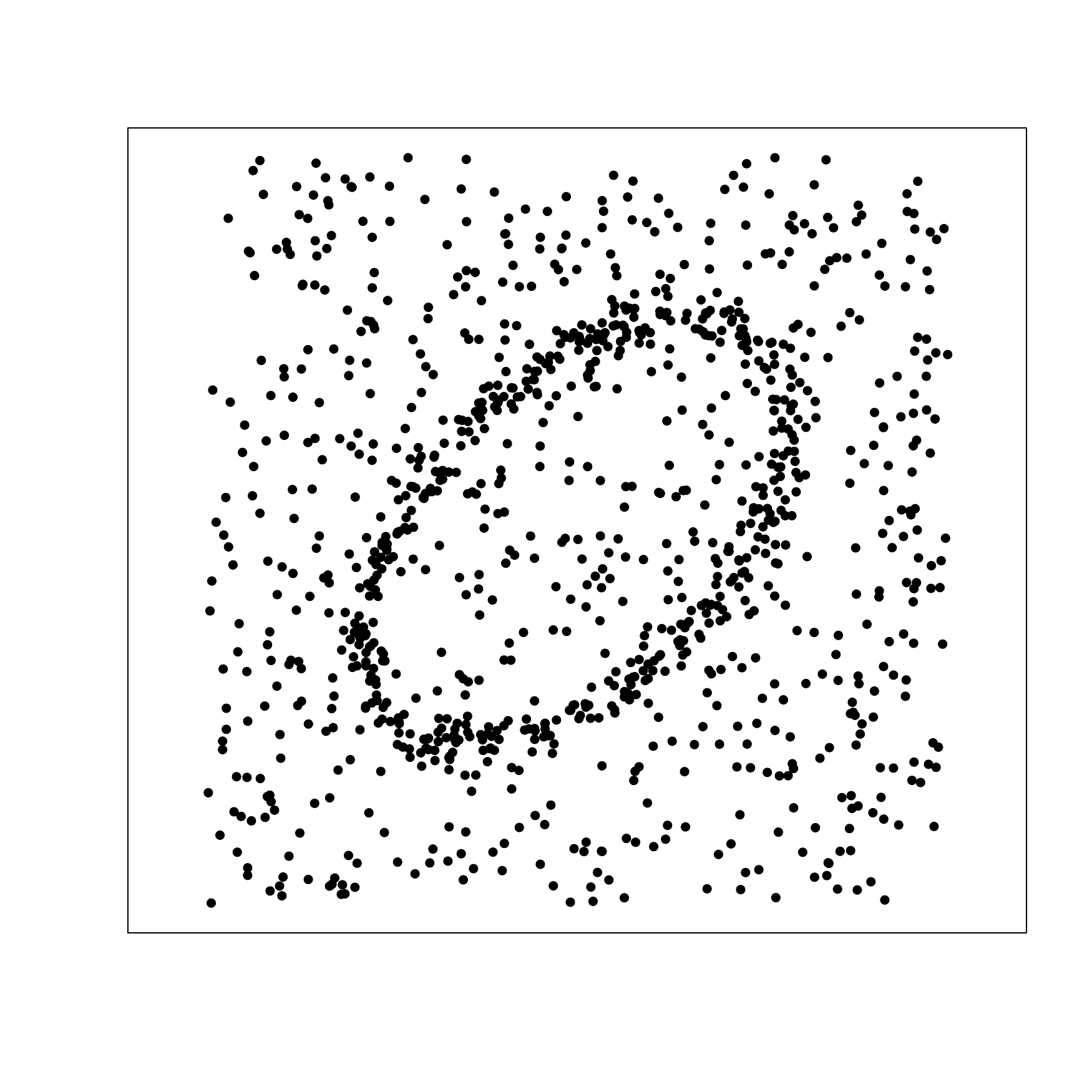} & 
\includegraphics[scale=.3]{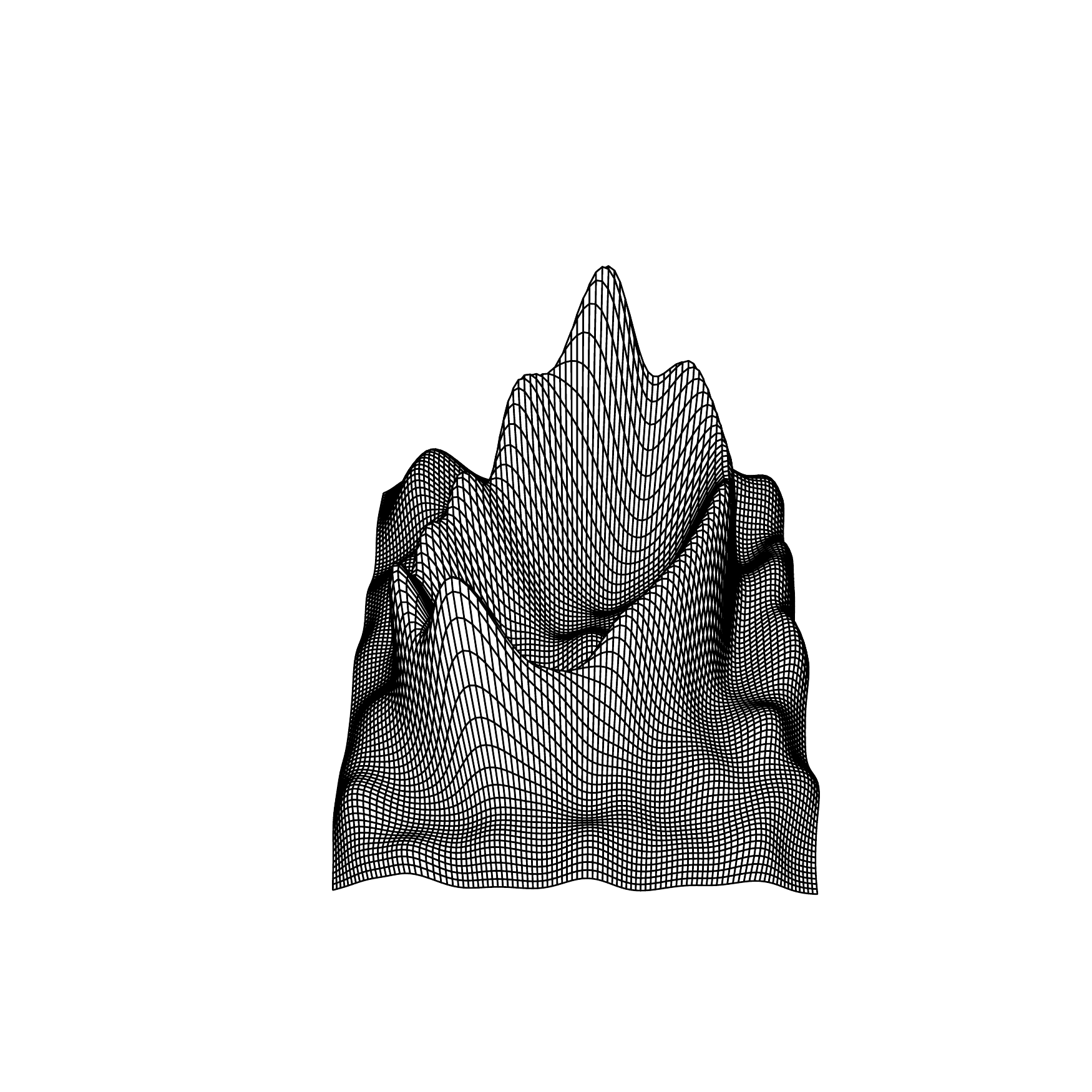} & 
\includegraphics[scale=.3]{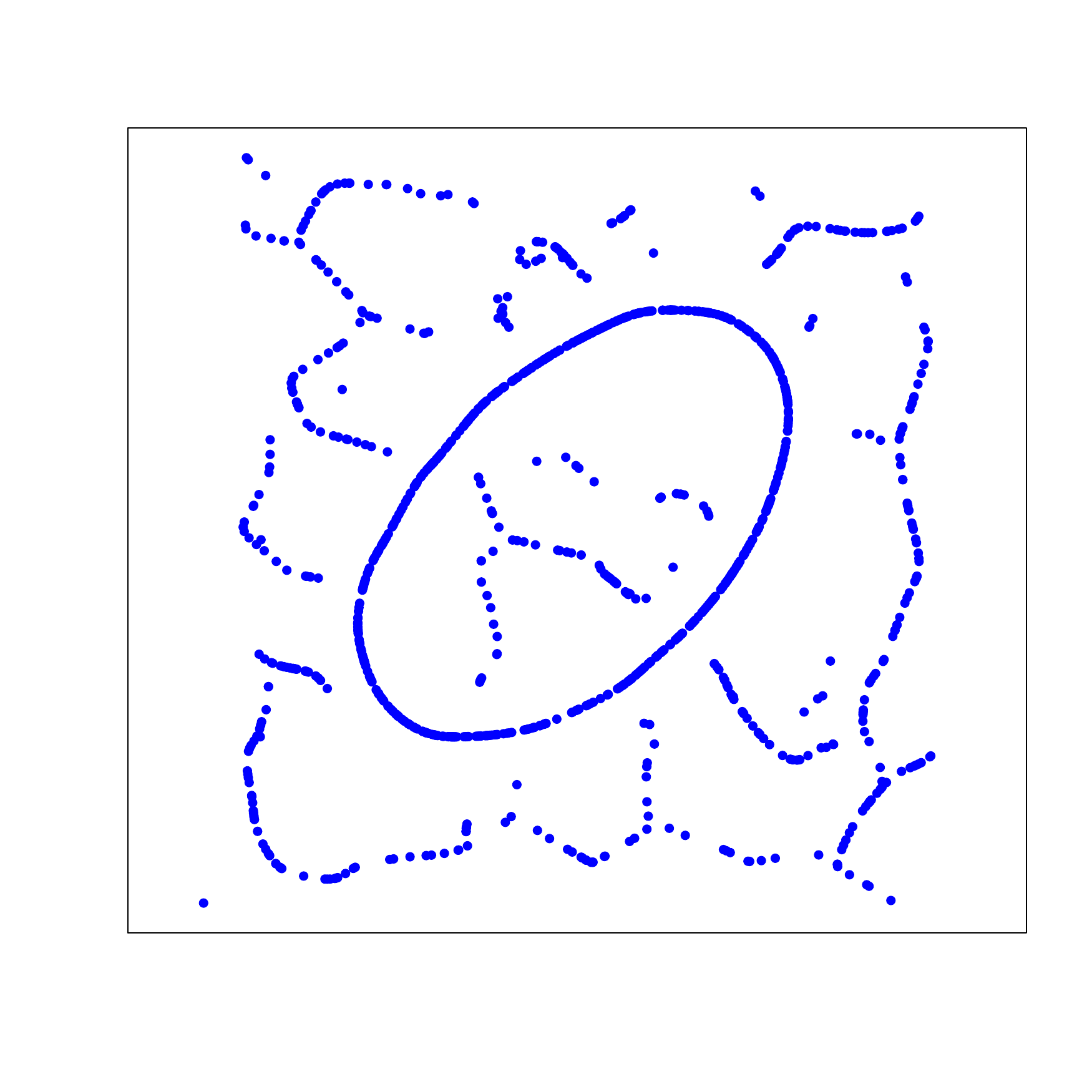}\\
\end{tabular}
\end{center}

\vspace{-.8in}
\begin{center}
$$
\hspace{2em}{\rm Data} \hspace{5em}  \Longrightarrow \hspace{5em} 
              \hat p   \hspace{5em}  \Longrightarrow \hspace{6em}  
              \hat R_d
$$
\end{center}

\vspace{-.5in}
\begin{center}
\begin{tabular}{ccc}
\includegraphics[scale=.3]{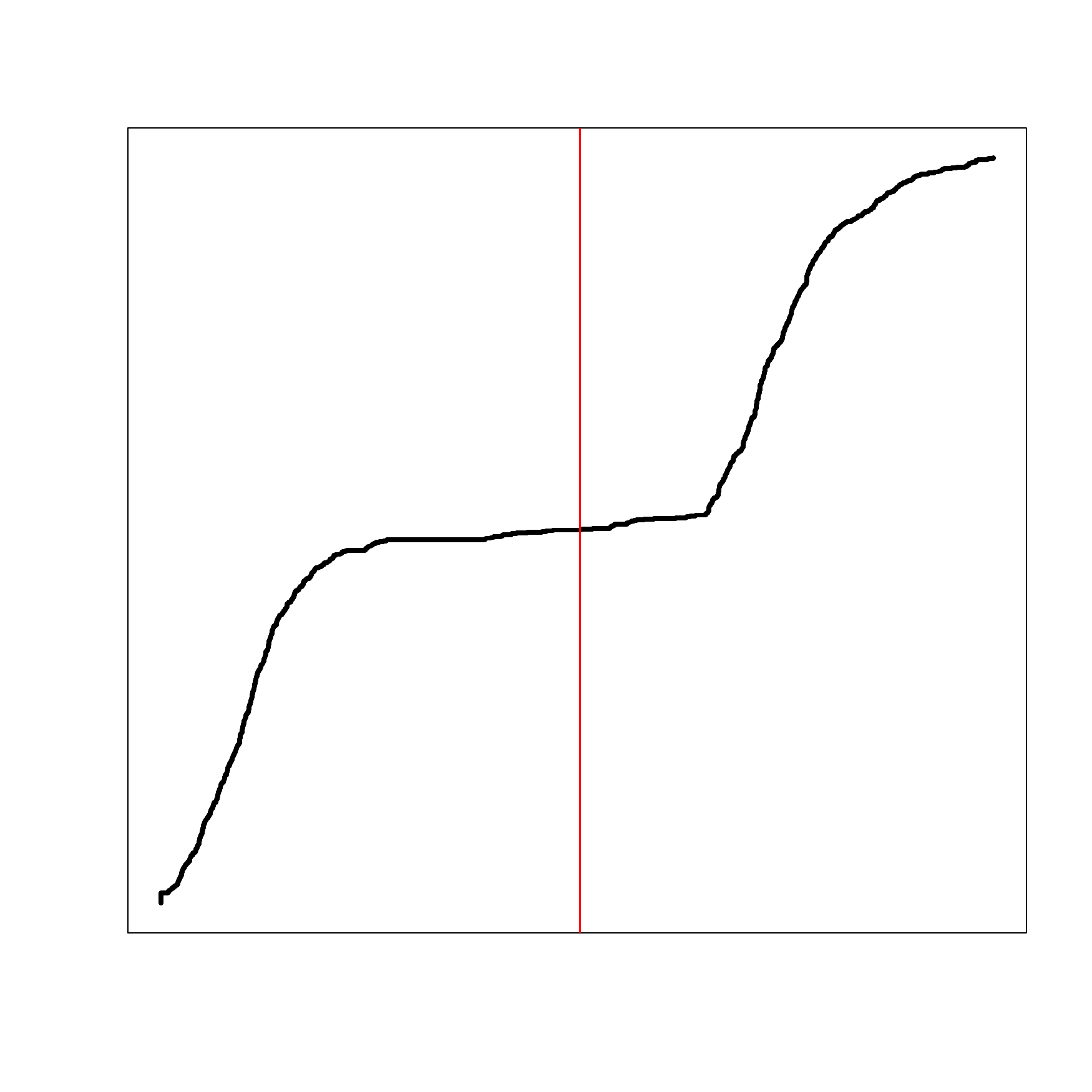} & 
\includegraphics[scale=.3]{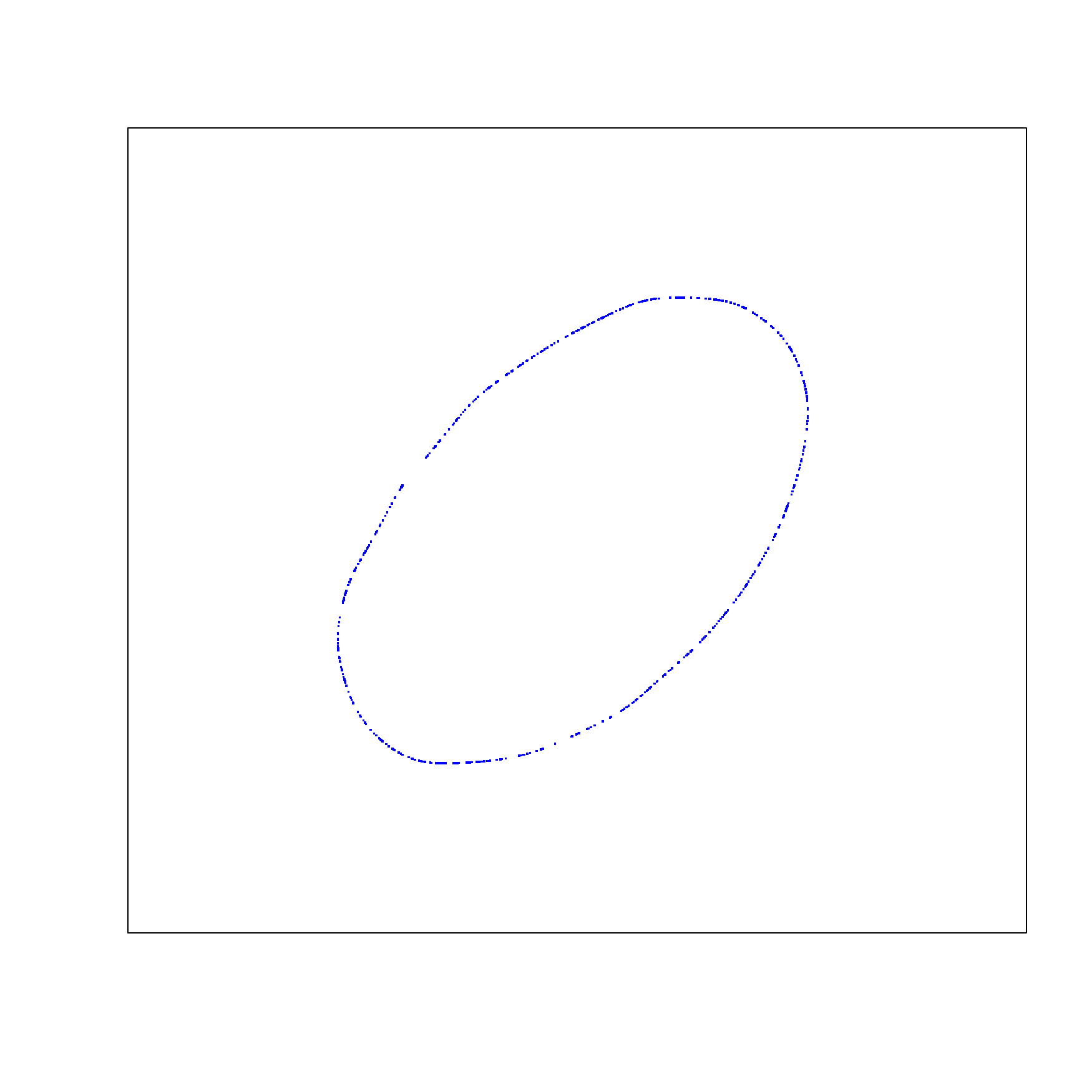} &
\includegraphics[scale=.3]{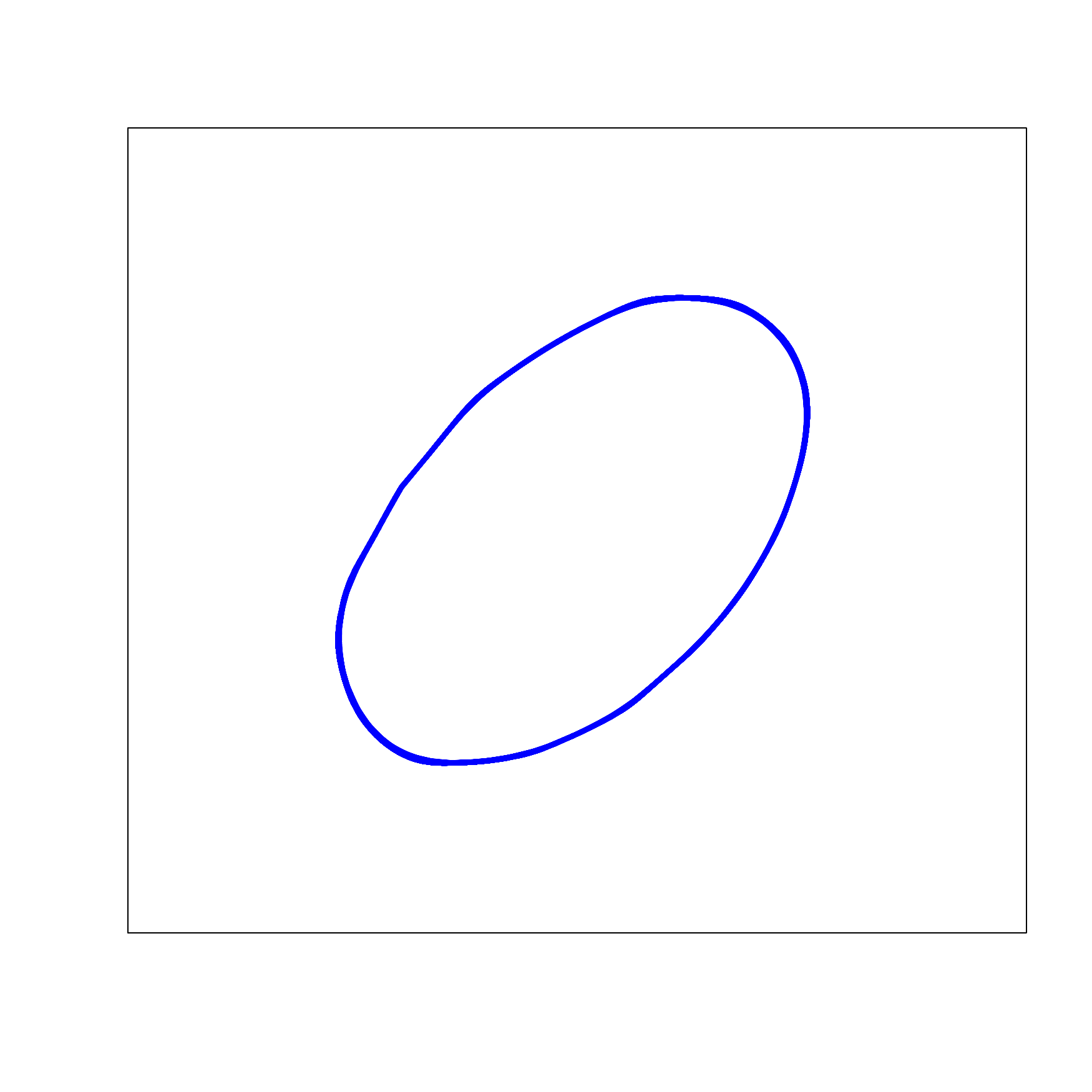}\\
\end{tabular}
\end{center}

\vspace{-.8in}
\begin{center}
$$
{\rm cdf\ of\ Signature}
\hspace{3em} \Longrightarrow \hspace{6em} 
\hat R_d^\dagger \hspace{6em} \Longrightarrow \hspace{4em} {\rm Feature}
$$
\captionof{figure}{\em Steps of the algorithm.
First we estimate the density then extract the ridge.
Next we compute the signatures over the ridge.
Based on the empirical cdf of the signatures 
(see section 4.5)
we choose a threshold.
Finally we apply single linkage clustering to the surviving points.
Even though $\hat p$ can be complex and unsmooth,
the resulting singular feature is very smooth and simple.}
\label{fig::flowchart}
\end{center}

The first step uses kernel density estimation.
The second step
uses the subspace constrained mean shift (SCMS) algorithm due to
\cite{Principal}; see also \cite{chen2015asymptotic} and \cite{us::2014.ANN}.
This algorithm finds density ridges $\hat{R}_d$ of dimension $d$
for $d=0,1,\ldots, D-1$.
The output of this step is a set (or collection of sets)
of dimension $d$.
In the third step,
we use the eigenvalues of
the estimated Hessian of the density to eliminate
weak ridges.
In the last step, 
we find the connected components using single linkage clustering
and we eliminate the small components.
What remains are large, high-density, lower dimensional structures.

\smallskip
{\em Related Work.}
Mode finding is a special case of
singular feature finding: modes of the density are
zero dimensional singular features.
There is a large literature on mode finding.
It is impossible to list all the references here.
Some useful recent references include
\citep{klemela2009smoothing, li2007nonparametric, dumbgen2008multiscale, 
chacon, us2015modes}.
Ridge theory has a long history in image processing;
a standard reference is
\cite{eberly1996ridges}.
In the statistical context,
ridges have been studied in
\cite{Principal},
\cite{us::2014.ANN}, \cite{chen2015asymptotic},
\cite{qiao2014theoretical}.
The idea of classifying structure based on the eigenvalues
of the Hessian of the density 
has appeared in several places such as statistics, astronomy
and image processing.
A good astronomy reference
is \cite{NEXUS}.
Perhaps the work closest to ours in this regard is
\cite{godtliebsen2002significance}.
They estimate the density and then classify points
into different types
(valleys, ridges, etc) according to 
the eigenvalues of the Hessian.
A critical difference with our approach
is that our method
separates structure by dimension.

\smallskip
A related idea is
manifold learning.
The literature on this topic is enormous.
Some early key references are
\cite{tenenbaum2000global, roweis2000nonlinear}.
These methods are aimed at identifying structure when the data
fall exactly on a sub-manifold
which is quite different than the setting in this paper.

There is a large literature on spectral clustering
which transforms the data using a kernel then performs
some form of clustering.
However, spectral methods do not work well
in the presence of noise and they do not output structures of
a given dimension $d$.
We show an example in Section~5.

We would also like to mention some clustering methods
aimed specifically at finding non-convex clusters.
These include
\cite{amiri, karypis1999chameleon, zhong2003unified}.
These methods find clusters in two stages, which allows them to
find more general clusters than one can find with simpler methods such as $k$-means.
For example, 
\cite{zhong2003unified} use $k$-means clustering 
(which tends to find spherical clusters)
followed by post-processing that combines them into more general clusters.
These methods appear to be effective at finding non-convex clusters.
Their goal, however, is quite different than ours.

Another approach to finding structure is
persistent homology;
see
\cite{edelsbrunner2002topological,
edelsbrunner2010computational,
frosini1992discrete,
turner2012fr,
bubenik2007statistical,
carlsson2009topology,
carlsson2009theory,
chazal2008towards,
chazal2011persistence,
cohenSteiner2005stability,
edelsbrunner2008persistent,
ghrist2008barcodes}. 
Roughly speaking, 
persistent homology looks for voids in the data.
Singular feature finding, instead,
seeks low dimensional, high density regions.
In some cases, voids could be surrounded by high density ridges
so it may be possible to connect ridge-based methods
to persistent homology but we do not pursue that connection
in this paper.

\smallskip
{\em Outline.}
In Section \ref{sec::background}
we give some mathematical background on modes and ridges.
In Section \ref{sec::singular}
we formally define singular features.
In Section \ref{section::estimation} we show how to estimate the singular features.
We consider some examples in Section \ref{sec::examples}.
In Section \ref{sec::theory}
we study the asymptotic properties of the method.
We conclude in Section \ref{sec::discussion}
with some discussion.

\vspace{-.5cm}
\section{Background: Modes and Ridges}
\label{sec::background}

%\vspace{-.5cm}
We begin by reviewing mode finding
(\cite{klemela2009smoothing, li2007nonparametric, dumbgen2008multiscale, 
chacon, Enhanced, us2015modes}).
Let $p$ be a density on $\mathbb{R}^D$
with gradient $g$ and Hessian $H$.
We assume that $p$ is a Morse function
meaning that the Hessian is non-degenerate at all critical points.
Until Section
\ref{section::estimation},
we assume that $p$ is known.

A point $m$ is a mode if
$||g(x)||=0$ and the eigenvalues of
$H(x)$ are all negative.
The {\em flow}
starting at any point $x$,
is a path $\pi_x:\mathbb{R}\to \mathbb{R}^D$ such that
$\pi_x(0)=x$ and
\begin{equation}\label{eq::flow}
\pi'_x(t)=\frac{d}{dt} \pi(t) = \nabla p(\pi_x(t)).
\end{equation}
Then $\pi_x$ is the steepest ascent curve starting at $x$.
The destination of the path is defined by
$$
{\rm dest}(x) = \lim_{t\to\infty} \pi_x(t).
$$
It can be shown that 
${\rm dest}(x)$ is a mode,
(except for $x$ in a set of measure 0, whose paths lead to saddlepoints;
\cite{irwin1980smooth,chacon,us2015modes}).
Thus the modes can be thought of as the limits of gradient ascent paths.
This observation leads to an algorithm for finding the modes
of a kernel density estimator
called the {\em mean shift algorithm}
(\cite{Fukunaga, Comaniciu}).
Mode finding is similar to
level set estimation
(\cite{polonik1995measuring, cadre2006kernel, 
walther1997granulometric})
but, in general, the two problems are not the same.

Now we turn to ridges.
Let 
\begin{equation}
\lambda_1(x)\geq \lambda_2(x) \geq \cdots \geq \lambda_D(x)
\end{equation}
denote the eigenvalues of $H(x)$
and let $\Lambda(x)$ be the diagonal matrix whose
diagonal elements are the eigenvalues.
Write the spectral decomposition of $H(x)$ as
$H(x) = U(x) \Lambda(x) U(x)^T.$
Fix $0 \leq d < D$
and let
$V(x)$ be the last $D-d$ columns of $U(x)$
(that is, the columns corresponding to the $D-d$ smallest 
eigenvalues).
If we write $U(x) = [V_\diamond(x):V(x)]$ then we can write
$H(x) = [V_\diamond(x):V(x)]\Lambda(x) [V_\diamond(x):V(x)]^T$.
Let $L(x) = V(x) V(x)^T$ be the projector
on the linear space defined by 
the columns of $V(x)$.
That is, $V(x) V(x)^T$ projects onto the local normal space.
Similarly,
$V_\diamond(x) V_\diamond(x)^T$ projects onto the local tangent space.
Define the {\em projected gradient} 
\begin{equation}
G(x) = L(x) g(x).
\end{equation}
The ridge set $R_d$ of dimension $d$ is defined by
\begin{equation}
R_d = \Biggl\{ x:\ ||G(x)|| = 0,\ \ \ \lambda_{D-d} < 0\Biggr\}.
\end{equation}
The ridge set is geometrically like the ridge of a mountain: each point in $R_d$
is a local maximum along a slice in the normal direction.
Modes (0-dimension ridges) have $D$ negative eigenvalues.
One dimensional ridges are filaments and have $D-1$ negative eigenvalues.
Two dimensional ridges are walls and have $D-2$ negative eigenvalues.
And so on.
By definition,
\begin{equation}
R_0 \subset R_1 \subset \cdots \subset R_{D-1}.
\end{equation}

Like modes, ridges
are also the destinations of the paths
but we replace the gradient field by the projected gradient vector field.
Specifically, define $\pi_x$ by
$\pi_x(0) = x$ and
\begin{equation}\label{eq::diffeq}
\pi'_x(t) = G(\pi(t)).
\end{equation}
The points in $R_d$
are the destinations of these paths.

%%\medskip
%%More formally,
%%if the vector field $G(x)$ is Lipschitz
%%then by Theorem 3.39 of \cite{irwin1980smooth},
%%$G$ defines a global flow as follows.
%%The flow is a family of functions
%%$\phi(x,t)$ such that
%%$\phi(x,0)=x$ and $\phi'(x,0) = G(x)$ and
%%$\phi(s,\phi(t,x))=\phi(s+t,x)$.
%%The flow lines, or integral curves, partition the space
%%and at each $x$ where $G(x)$ is non-null, 
%%there is a unique integral curve passing through $x$.
%%Unlike the paths defined by the 
%%gradient $g$ which move towards modes,
%%the paths define by $G$ move towards ridges.
%%
%%

\section{Singular Features}
\label{sec::singular}

Ridge finding is a good start
for finding structure.
For example, it has been used successfully to
map out the structure of matter in the Universe
\citep{cosmicweb}.
But ridges that are flat or small are not interesting.
We want to find the portions of the
$d$-dimensional ridges that are
sharp and large.
Also, as we explain later, ridges experience dimensional leakage:
zero dimensional structure shows up in one dimension ridges and vice versa.
In other words, ridge finding does not cleanly separate structure
into objects of different dimensions.
The next section shows how to find the sharp portion of the ridge.

\subsection{Sharpness}

We quantify the sharpness of a ridge component by the eigenstructure
of the Hessian $H$.
The idea is to construct functions of the eigenvalues --- called eigensignatures ---
that summarize the local shape of the density.
In statistics, 
\cite{godtliebsen2002significance},
used the eigenstructure of $H$ to classify points as belonging to different types of structure.
Examples from the image processing literature include
\cite{descoteaux2004multi} and
\cite{frangi1998multiscale}.
In the astronomy literature, examples include
the NEXUS signature due to \cite{NEXUS} and the method in \cite{snedden2014new}.
Our use of the eigenstructure is somewhat different than these papers.

Let $\lambda_1(x) \geq \cdots \geq \lambda_D(x)$
be the eigenvalues of the Hessian $H(x)$ of the $\log p(x)$.
Let us also define $\lambda_0(x) \equiv 0$.
The intuition of all eigensignature methods is based
on the following heuristic.
In the {\em idealized case},
a $d$-dimensional structure
would have the $d$ smallest eigenvalues
very negative and approximately equal.
The remaining eigenvalues would be close to 0.
For example, when $D=3$
we have
the following idealized cases:
\begin{align*}
{\rm mode} \ \     & \lambda_3 \approx \lambda_2 \approx \lambda_1 < 0\\
{\rm filament} \ \ & \lambda_3 \approx \lambda_2 < 0 \ {\rm and}\ \lambda_1 \approx 0\\
{\rm wall} \ \     & \lambda_3 < 0\   {\rm and}\ \lambda_2 \approx \lambda_1 \approx 0.
\end{align*}
Let us expand on this idealized case.
Imagine a unimodal density with a spherical, sharply defined mode.
In that case, the eigevalues will all be negative and approximately equal, that is,
$\lambda_3 \approx \lambda_2 \approx \lambda_1 < 0$.
Now suppose that the mode is locally highly elliptical so that the density around the mode
looks like a filament.
In this case there will be two very negative eigenvalues
but the largest eigenvalue will be close to 0.
Thus,
$\lambda_3 \approx \lambda_2 < 0 \ {\rm and}\ \lambda_1 \approx 0$.
In the last case, suppose the density is very sharply concentrated around the mode
in one dimension only and is otherwise very flat.
Here, the density resembles a wall
and we have
$\lambda_3 < 0$ and
$\lambda_2 \approx \lambda_1 \approx 0$.

Now, we need to define functions of the eigenvalues
that formalize these cases.
We want a function $S_0$ that is large when
$\lambda_3 \approx \lambda_2 \approx \lambda_1 < 0$.
Similarly, we want a function $S_1$ that is large when
$\lambda_3 \approx \lambda_2 < 0$ and $\lambda_1 \approx 0$.
And so on.
To formalize this,
we define the {\em eigensignatures}:
\begin{align*}
S_0 &= |\lambda_1|\ \frac{|\lambda_1|}{|\lambda_3|} I(\lambda_1 < 0)\\
S_1 &= |\lambda_2|\ \frac{|\lambda_2|}{|\lambda_3|} 
\left(1 - \frac{|\lambda_1 \wedge 0|}{|\lambda_3|}\right)I(\lambda_2 < 0)\\
S_2 &= |\lambda_3|\ \left( 1 - \frac{|\lambda_1 \wedge 0|}{|\lambda_3|}\right)
\left( 1 - \frac{|\lambda_2 \wedge 0|}{|\lambda_3|}\right)I(\lambda_3 < 0)\\
\vdots &= \vdots
\end{align*}
where $I(\cdot)$ is the indicator function. More generally, for
$d=0,\ldots, D-1$:
\begin{equation}
S_j = I(\lambda_{j+1} < 0)\  
         |\lambda_{j+1}|\ \left(\frac{|\lambda_{j+1}|}{|\lambda_d|} \right)
    \prod_{i=0}^j \left( 1 - \frac{|\lambda_i \wedge 0|}{|\lambda_d|}\right)
\end{equation}
where we recall that $\lambda_0 \equiv 0$.

{\bf Remark:}
Note that, for simplicity, we have written
$S_j$ instead of $S_j(x)$ and
$\lambda_j$ instead of $\lambda_j(x)$.
We will often do this in what follows.

To get more intuition,
consider the defintion of $S_0$.
Suppose that $\lambda_3 \approx \lambda_2 \approx \lambda_1 < 0$.
Then the first term $|\lambda_1|$ will be large and the second term
$|\lambda_1|/|\lambda_3|$ will be close to 1.
Thus, $S_0$ will be large.
Now suppose 
$\lambda_3 \approx \lambda_2 < 0$
and $\lambda_1 \approx 0$.
In this case, 
$|\lambda_1|/|\lambda_3|$ will be small, causing $S_0$ to decrease.
On the other hand, all the terms in $S_1$ will be large.
Hence, a filament-like structure will have $S_0$ small and $S_1$ large.
Similar remarks apply to the higher order structures.

\newpage
A large value of $S_0$ 
thus suggests that $x$ is a mode.
A large value of $S_1$ suggests that 
$x$ is a one-dimensional ridge point and so on.
In particular, given a threshold $T_d$, we call
$$%\begin{equation}
R_d^\dagger = \Bigl\{ x\in R_d:\ S_d(x) > T_d\Bigr\}
$$%\end{equation}
the sharp portion of the ridge.
In Section \ref{section::tuning} we explain how to choose the threshold $T_d$.

%\begin{samepage}
To explore the eigensignature further, 
Figure \ref{fig::idealizedsig}
shows some stylized cases.

\begin{center}
\begin{tabular}{lccc|ccc}
          & $\lambda_3$ & $\lambda_2$ & $\lambda_1$ & $S_0$ & $S_1$ & $S_2$ \\ \hline\hline
mode:     & $-C$        & $-C$        & $-C$        & $C$   & 0     & 0 \\  \hline
filament: & $-C$        & $-C$        & 0           & 0     & $C$   & 0\\    \hline
wall:     & $-C$        & $0$         & 0           & 0     & 0     & $C$\\
\end{tabular}
\captionof{figure}{\em Examples of signatures. The first three
entries of each row show a configuration
of the eigenvalues of the Hessian.
$C>0$ is any positive constant.
The last three entries show the corresponding values of the signature
$S = (S_0,S_1,S_2)$.}
\label{fig::idealizedsig}
\end{center}

Some more varied cases are shown Figure \ref{fig::sig1}.
The left plot of Figure \ref{fig::sig1} shows that when there
\vspace{-.3in}
\begin{center}
\begin{tabular}{cc}
%\hspace{-2em}\includegraphics[width=3.5in,height=6in]{Eigen1} 
\hspace{-2em}\includegraphics[width=3.5in,height=5in]{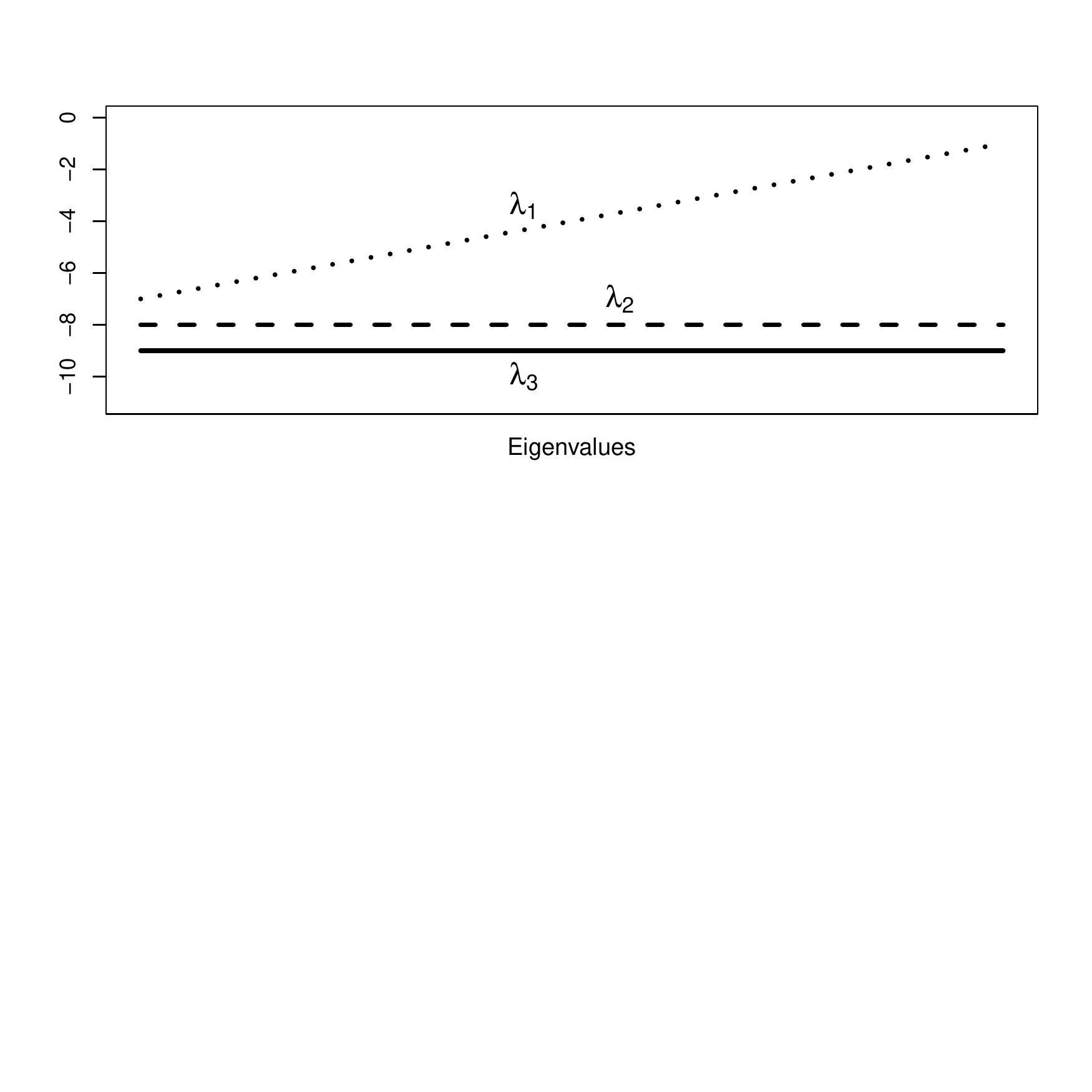} 
\hspace{-2em}& \hspace{-2em}
\includegraphics[width=3.5in,height=5in]{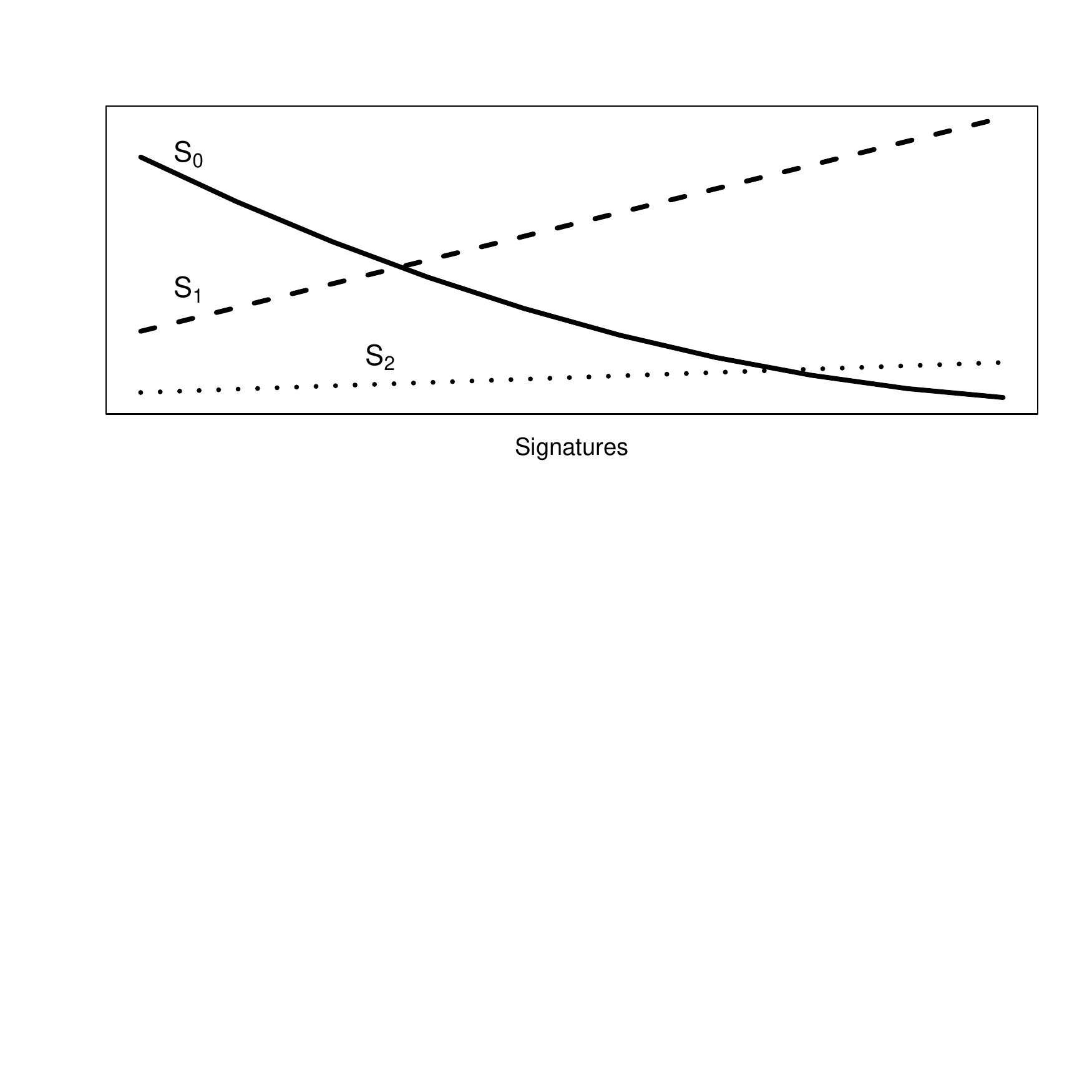} 
\end{tabular}
\vspace{-3.in}
%\vspace{-3in}
\captionof{figure}{\em Left: The evolution from mode to filament.
Left: eigenvalues. Solid line: $\lambda_1$, broken line: $\lambda_2$,
dotted line: $\lambda_3$. 
Right: signatures. Solid line: $S_0$, 
broken line: $S_1$, dotted line: $S_2$.}
\label{fig::sig1}
\end{center}
\smallskip
are three, negative, tightly clustered eigenvalues,
$S_0$ is large. 
As we increase the largest eigenvalue
(which corresponds to moving up higher in the left plot)
we see that $S_0$ decreases and $S_1$ increases.
Thus, two negative, clustered eigenvalues and one larger eigenvalue
is declared to be more filament-like.

\section{Estimating Singular Features}
\label{section::estimation}

So far, our discussion has referred to the true density $p$.
To find singular features in point clouds,
we need to estimate all the relevant quantities.

\subsection{Estimating The Density}

To estimate the density $p$,
we use the standard kernel density estimator
\begin{equation}
\hat p_h(x) =
\frac{1}{n}\sum_{i=1}^n 
\frac{1}{h^D}K\left(\frac{||y-X_i||}{h}\right)
\end{equation}
where $K$ is a smooth, symmetric kernel and $h>0$
is the bandwidth.

We use the multivariate
form of Silverman's Normal reference rule to choose the bandwidth $h$.
For a matrix-valued bandwidth,
\cite{chacon2011asymptotics}
show that this is
$$
H = 
\left(\frac{4}{n (d+2)}\right)^{\frac{2}{4+d}}S
$$
where $S$ is an estimate of the covariance matrix.
For simplicity we use a scalar bandwidth
\begin{equation}\label{eq::silverman}
h = \left(\frac{4}{n (d+2)}\right)^{\frac{1}{4+d}}s
\end{equation}
where $s^2 = D^{-1}\sum_{j=1}^D s_j^2$, and $s_j$ is 
the standard deviation of the $j^{\rm th}$ coordinate.

{\bf Remark:}
This choice of bandwidth has the optimal rate for estimating $p$
under standard smoothness assumptions.
However, estimating the ridges involves estimating the gradient and Hessian
as well as the density.
In principle, we could use separate bandwidths for estimating
the gradient and the Hessian.
However, to keep the method simple, we suggest using
a single bandwidth even if it is not optimal.

%%{\bf Remark}
%%An argument can be made that it might be better to use a bandwidth
%%that does not even tend to 0 as 
%%$n\to\infty$.
%%\cite{us::2014.ANN} argue 
%%that, for the purposes of exploring structure in the 
%%data, there are advantages in using a fixed 
%%(non-deacreasing) bandwidth. 
%%In this case, we 
%%regard $\hat p_h(x)$ as an estimate of
%%$p_h(x) = \mathbb{E}(\hat p_h(x))$ rather than 
%%as an estimate of $p$. Note that $p_h$ is just a smooth 
%%version of $p$ and we expect $p_h$ to preserve 
%%the salient structure in $p$. This notion has been made 
%%formal by Theorem 7 in 
%%\cite{us::2014.ANN}. The theorem shows that 
%%there is a well-defined ridge in a neighborhood 
%%of a manifold. Further, the ridge is close to the 
%%manifold and it is nearly homotopic to it. 
%%The non-decreasing-bandwidth approach proposed in
%%\cite{us::2014.ANN} has the advantage that
%%$$
%%||\hat p_h - p_h||_\infty =
%%O_P\left(\sqrt{\frac{\log n}{n}}\right)
%%$$
%%independently of the dimension.
%%
%%
%%

\subsection{Estimating the Ridge}

Let $\hat H(x)$ be the Hessian of $\hat p(x)$ and let
\begin{equation}
\hat\lambda_1(x)\geq \hat\lambda_2(x) \geq \cdots \geq \hat\lambda_D(x)
\end{equation}
denote the eigenvalues of $\hat H(x)$.
Write the spectral decomposition of $\hat H(x)$ as
$$
\hat H(x) = \hat U(x) \hat \Lambda(x) \hat U(x)^T
$$
and let
$\hat V(x)$ be the last $D-d$ columns of $\hat U(x)$
(that is, the columns corresponding to the $D-d$ smallest 
eigenvalues).
Write $\hat U(x) = [\hat V_\diamond(x):\hat V(x)]$ and let 
$\hat L(x) = \hat V(x) \hat V(x)^T$ be the projector
on the linear space defined by 
the columns of $\hat V(x)$.
The plugin estimate of the projected gradient
is $\hat G(x) = \hat L(x) \hat g(x)$
where $\hat g$ is the gradient of $\hat p$.
The estimated ridge set $\hat R_d$ is
\begin{equation}
\hat R_d = \Biggl\{ x:\ ||\hat G(x)|| = 0,\ \ \ \hat \lambda_{D-d} < 0\Biggr\}.
\end{equation}
The properties of $\hat R_d$ are studied in
\cite{morse-smale}, \cite{chen2015asymptotic} and
\cite{us::2014.ANN}.
In particular, 
\cite{us::2014.ANN} showed that
the estimated ridges consistently estimate the true ridges.
\cite{chen2015asymptotic} established the limiting distribution of the estimated ridge
which leads to a method for constructing confidence sets.

In practice, we replace $\hat R_d$
with a numerical approximation
based on  the subspace constrained mean shift (SCMS) 
algorithm due to \cite{Principal}.
Before explaining the SCMS algorithm,
it is helpful to first review
the basic mean shift algorithm
(\cite{Fukunaga, Comaniciu}).
This is a method for finding the modes of a density
by approximating the steepest ascent paths.
The algorithm starts with a mesh of points and then 
moves the points along gradient ascent trajectories 
towards local maxima.

%%\vspace{-.7cm}
%\begin{center}
%\fbox{\parbox{6in}{
%\begin{center}
%{\sf SCMS Algorithm}
%\end{center}
%\vspace{-.7cm}
%\begin{enum}
%\item Input: $q(x) = \log \hat p(x)$ and a mesh of 
%points $M = \{m_1,\ldots, m_N\}$. \\
%\item Let $g(x)$ and $H(x)$ be the gradient and Hessian 
%of $q$.
%\item Let $V = [v_1 \cdots v_{D-d}]$ be the 
%$D\times (D-d)$ matrix whose columns are the\\
%eigenvectors of $H(x)$ corresponding to the 
%smallest $D-d$ eigenvalues.
%\item For each mesh point $m$ repeat:
%$$
%m \Longleftarrow 
%V V^T \left(\frac{\sum_i c_i X_i}{\sum_i c_i} - m\right)
%$$
%where
%$c_i = h^{-D}K(||m-X_i||/h)$.
%\item Stop when convergence is obtained.
%\end{enum}}}
%\captionof{figure}{\em The SCMS algorithm of \cite{Principal}.}
%\label{fig::algorithm}
%\end{center}

Recall that
$\hat{p}_h(x)$ denotes the kernel density estimator.
Let
${\cal M}=\{ m_1,\ldots, m_N\}$ be a collection of mesh points.
These are often taken to be the same as the data
but in general they need not be.

Let $m_j(1)=m_j$ and
for $t=1,2,3,\ldots$ we define the trajectory
$m_j(1),m_j(2),\ldots,$ by
\begin{equation}
m_j(t+1) = 
\frac{\sum_{i=1}^n X_i\, K\left( \frac{||m_j(t) - X_i||}{h}\right)}
{\sum_{i=1}^n K\left( \frac{||m_j(t) - X_i||}{h}\right)}.
\end{equation}
It can be shown that each trajectory
$\{m_j(t):\ t=1,2,3,\ldots,\}$ 
approximates the gradient ascent path and
converges to a mode of $\hat{p}_h$.
Conversely, if the mesh ${\cal M}$ is rich 
enough, then for each mode of $\hat{p}_h$, some 
trajectory will converge to that mode.
The mean shift algorithm is simply a numerical approximation to the
flow defined by (\ref{eq::flow}).
This was recently made rigorous in
\cite{Arias-Castro2013}.
The mean shift algorithm is simply a mode-finding algorithm.
In fact, it can be shown that the mean shift algorithm is just a type of gradient ascent.
Given any starting value $x$, the algorithm keeps moving the point in the direction of the gradient
towards a local maximum.

The SCMS algorithm mimics the mean shift algorithm
but it replaces the gradient with the projected
gradient at each step.
The algorithm can be applied to $\hat p$ or any
monotone function of $\hat p$, such as $\log \hat p$,
in Figure \ref{fig::algorithm}.
The SCMS algorithm can also be regarded as a gradient ascent algorithm.
Starting from a point $x$,
the algorithm moves the point in the direction of the gradient.
However, we want the point to move towards a ridge rather than a mode.
To accomplish this, SCMS takes the point and then projects into the space,
perpendicular to the ridge.
This forces the point to move towards the ridge and stop,
rather than moving up the ridge towards a mode.
The details are formalized in
\cite{Principal}.

%\vspace{-.7cm}
\begin{center}
\fbox{\parbox{6in}{
\begin{center}
{\sf SCMS Algorithm}
\end{center}
\vspace{-.7cm}
\begin{enum}
\item Input: $q(x) = \log \hat p(x)$ and a mesh of 
points $M = \{m_1,\ldots, m_N\}$. \\
\item Let $g(x)$ and $H(x)$ be the gradient and Hessian 
of $q$.
\item Let $V = [v_1 \cdots v_{D-d}]$ be the 
$D\times (D-d)$ matrix whose columns are the\\
eigenvectors of $H(x)$ corresponding to the 
smallest $D-d$ eigenvalues.
\item For each mesh point $m$ repeat:
$$
m \Longleftarrow m + 
V V^T \ \left(\frac{\sum_i c_i X_i}{\sum_i c_i} - m\right)
$$
where
$c_i = h^{-D}K(||m-X_i||/h)$.
\item Stop when convergence is obtained.
\end{enum}}}
\captionof{figure}{\em The SCMS algorithm of \cite{Principal}.}
\label{fig::algorithm}
\end{center}

Figure \ref{fig::plot2d} shows a simple dataset with the outputs
$S_0$ and $S_1$ of SCMS for $d=0$ and $1$. 
The left plot shows the data,
the middle plot shows $\hat R_0$ and the right plot shows
$\hat R_1$.
Here we see an effect that we call
{\em dimensional leakage}.
The set $\hat R_0$ contains some ridge points
(since ridge points can also be local modes).
Similarly, $\hat R_1$ contains the modes; this must happen since,
by definition,
$R_0 \subset R_1 \subset R_2 \cdots$. 
This is why filtering with the eigensignature is helpful.

\vspace{-.5cm}
\begin{center}
\begin{tabular}{ccc}
\includegraphics[scale=.28]{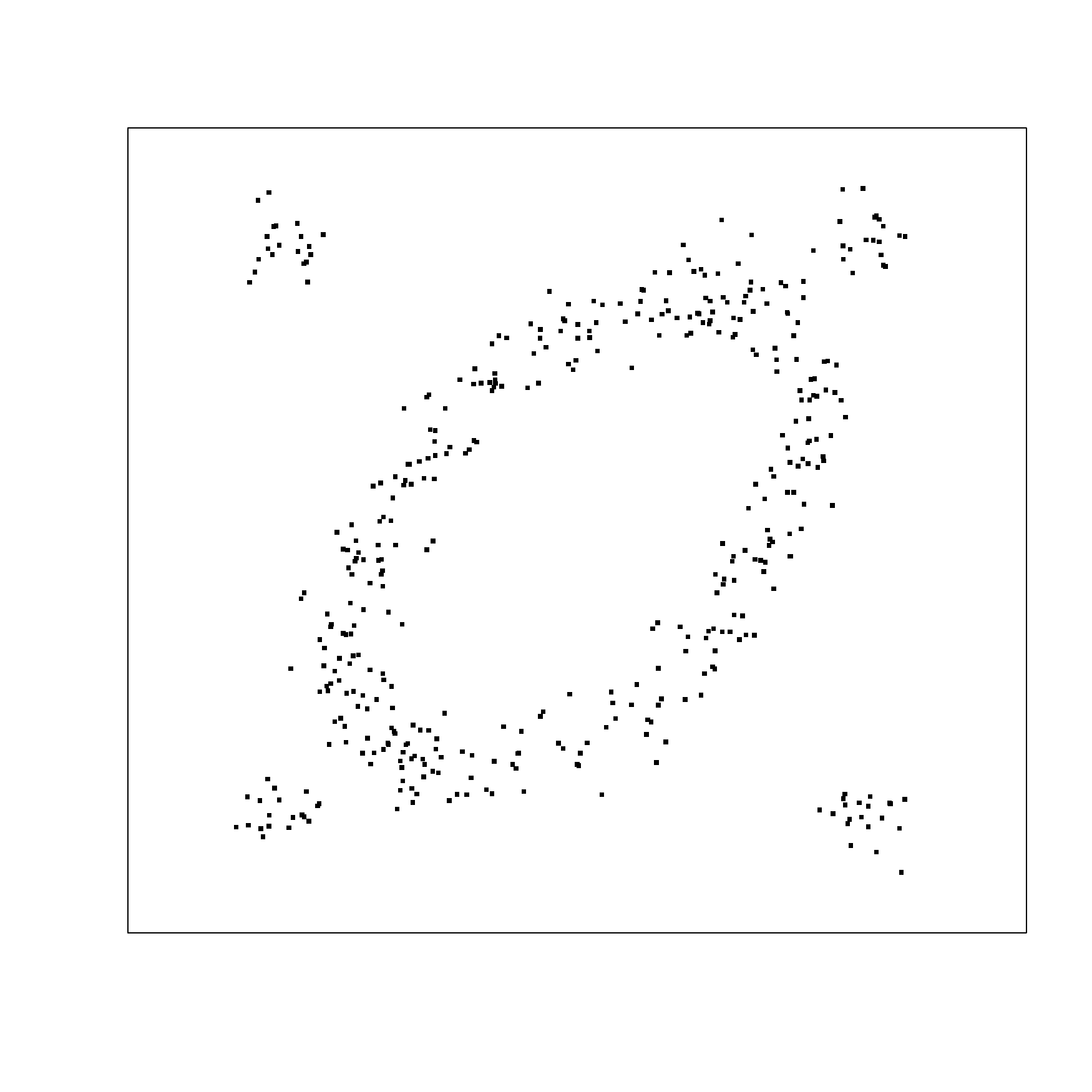} &
\includegraphics[scale=.28]{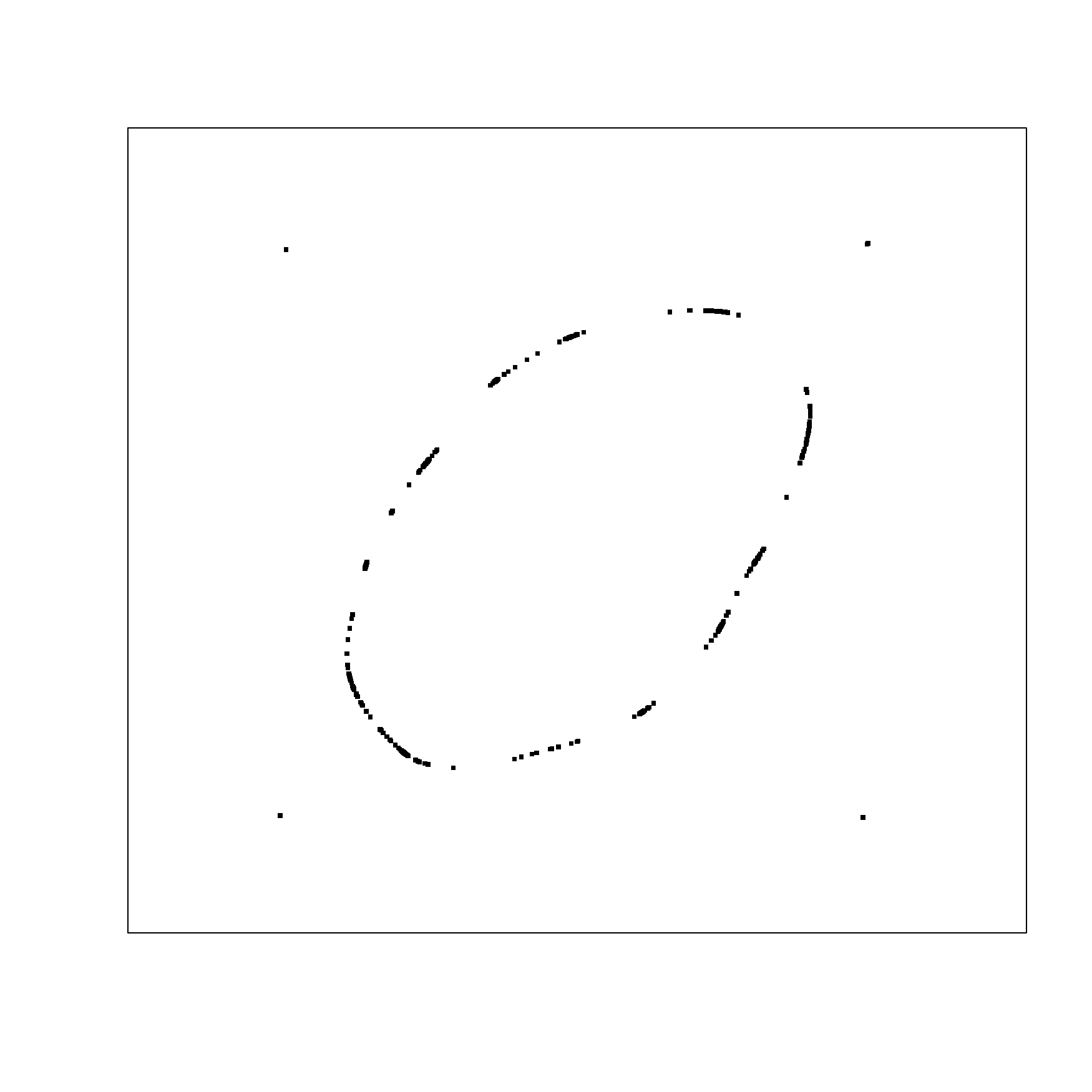} &
\includegraphics[scale=.28]{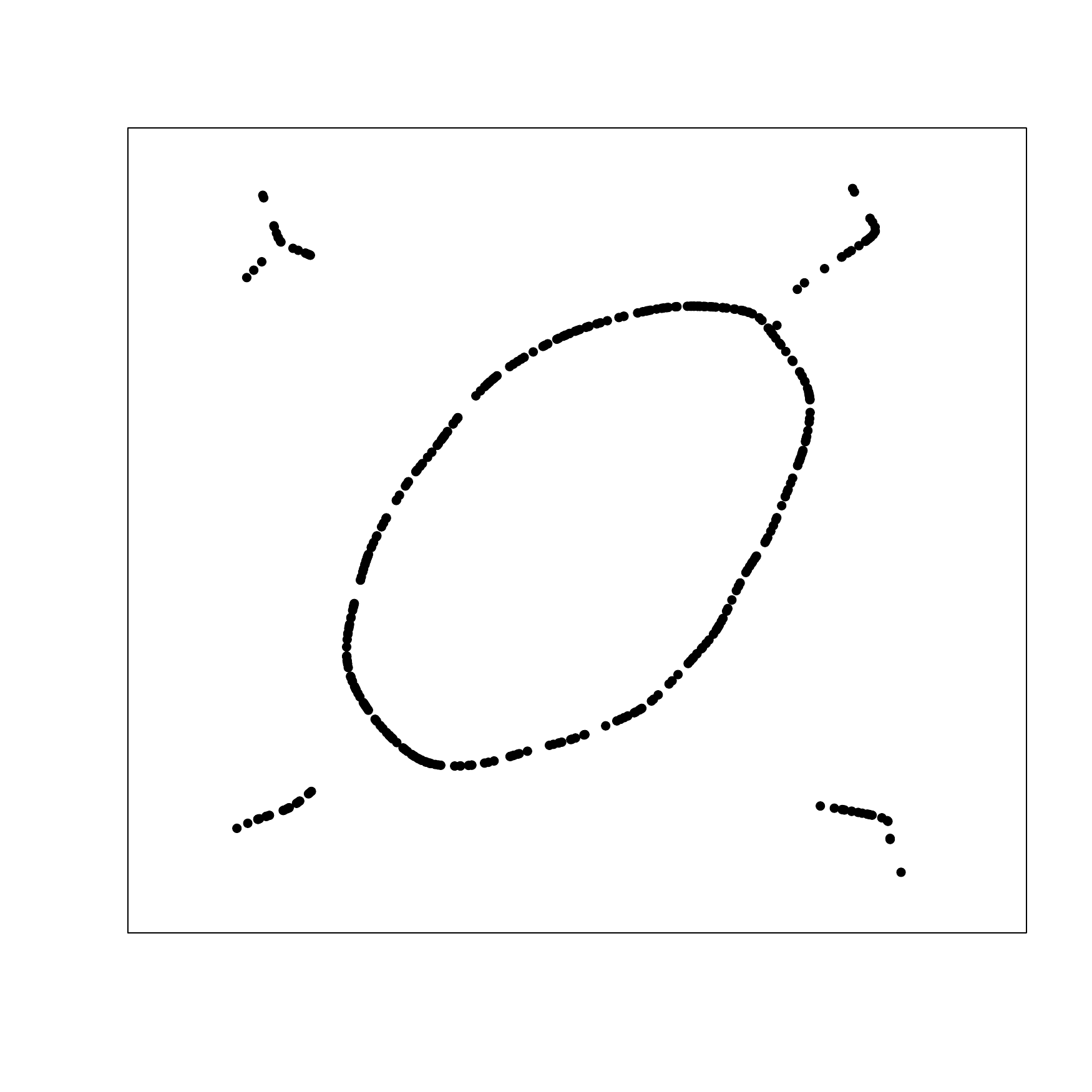}
\end{tabular}
\vspace{-.7cm}
\captionof{figure}{\em Dimensional leakage. Left: Data. Center: ${\rm SCMS}_0$ 
contains the true modes but also false modes from 
the oval. Right: ${\rm SCMS}_1$ contains the 
true oval but also false ridges from the modes.}
\label{fig::plot2d}
\end{center}

\vspace{-.1in}
In principle, we could use a very fine grid of starting values
and trace out $\hat R_d$ very accurately.
But this is time consuming and in practice
we have found that using the data points
as starting values leads to an accurate approximation of the ridge.
Thus, let $Y_i$ be the destination of SCMS after running the algorithm with
the data point $X_i$ as a starting point.
We use
$\hat R_d = \{Y_i:\ i=1,\ldots, n\}$
as our estimate of the ridge.

\subsection{Estimating the Signature}

We use the plug-in estimator for the 
signatures.
Thus, we take
$\hat H(x)$ to be the Hessian of $\log \hat p(x)$ and
we let $\hat\lambda_j(x)$ be the eigenvalues of $\hat H(x)$.
Then we
estimate the signature by
\begin{equation}
\hat S_j = I(\hat{\lambda}_{j+1} < 0)\  
         |\hat{\lambda}_{j+1}|\ \left(\frac{|\hat{\lambda}_{j+1}|}{|\hat{\lambda}_d|} \right)
          \prod_{i=1}^j \left( 1 - \frac{|\hat{\lambda}_i \wedge 0|}{|\hat{\lambda}_d|}\right)
\end{equation}
Having estimated the signature,
we estimate the sharp ridge by
\begin{equation}
\hat R_d^\dagger = \Biggl\{ Y_i\in \hat R_d:\ \hat S_d(Y_i)> T_d\Biggr\}.
\end{equation}

\subsection{Connected Components: Single Linkage Clustering 
(The Rips Graph)}
\label{section::singlelinkage}

\smallskip
The sets
$\hat R_d^\dagger$ can be used as the final output
of the procedure.
However, in some cases,
the user may want to summarize 
$\hat R_d^\dagger$ as a collection of connected components.

Let us write
$\hat R_d^\dagger = \{Y_1,\ldots, Y_m\}$ say.
The {\rm Rips graph}
is defined as follows.
Each node in the graph corresponds to a point
$Y_i \in \hat R_d^\dagger$. We put an edge between 
two nodes $Y_i$ and $Y_j$ if and only if
$||Y_i - Y_j|| \leq \epsilon$. The connected components 
of the Rips graph are precisely the clusters obtained 
\begin{center}
\begin{tabular}{cc}
\includegraphics[scale=.4]{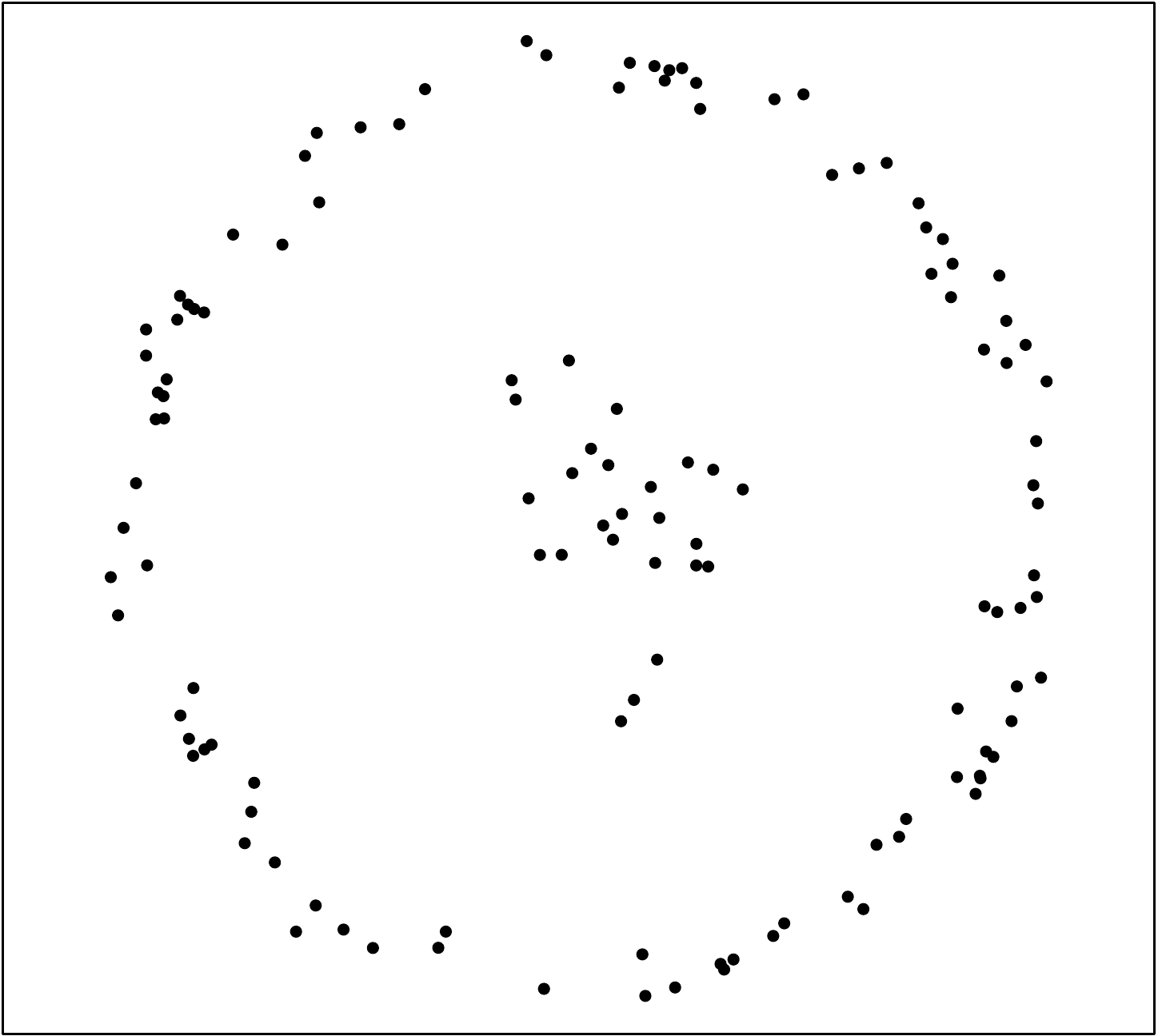} & 
\includegraphics[scale=.4]{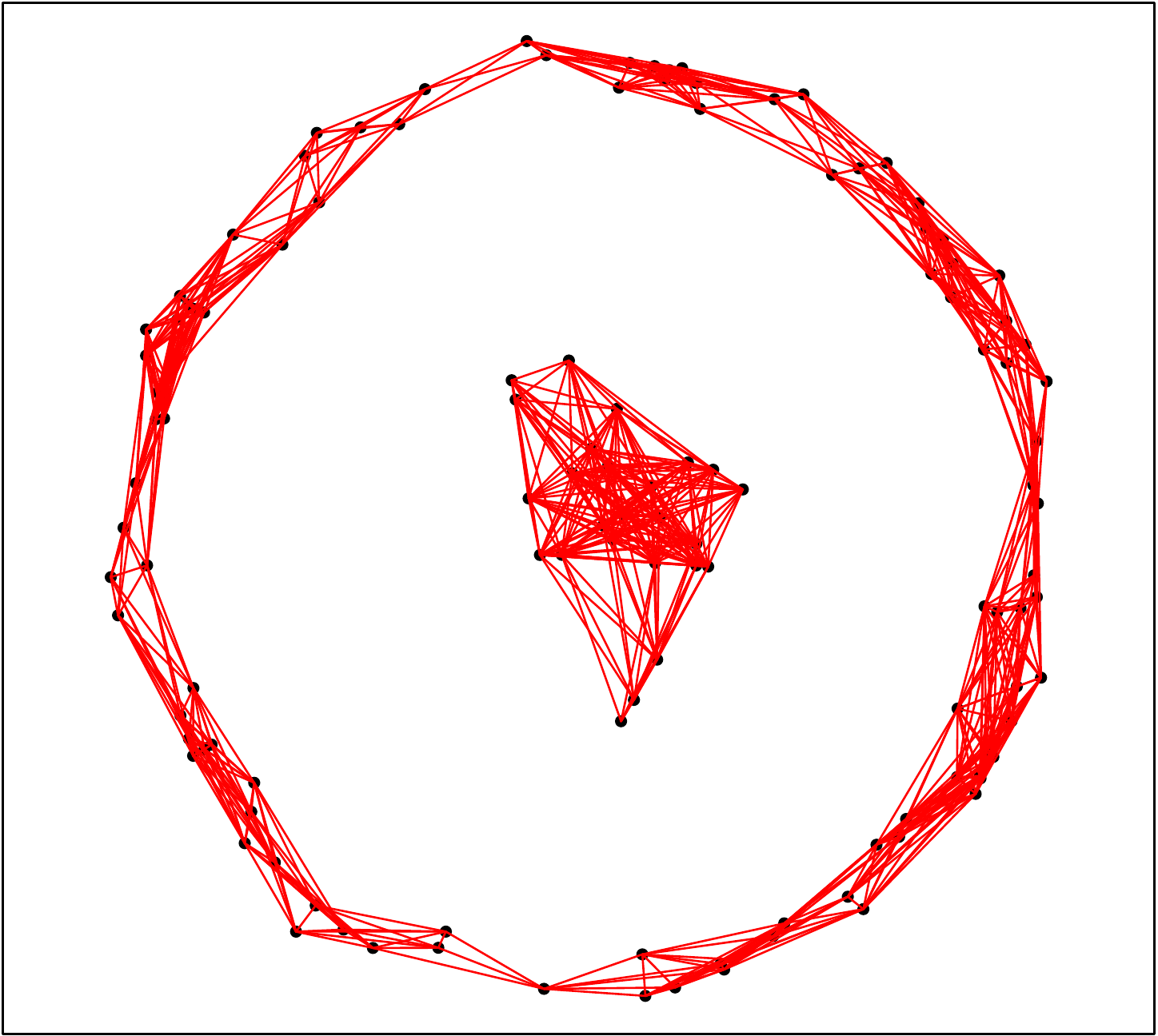}
\end{tabular}

\medskip
\captionof{figure}{Left: A two-dimensional data set. Right: the Rips 
graph. Every pair of points $X_i$ and $X_j$ such that 
$||X_i - X_j||\leq \epsilon$ is connected.}
\label{fig::rips}
\end{center}
by performing
single linkage clustering and then cutting the dendogram
at height $\epsilon$.
Let $\hat R_{d1}^\dagger, \ldots, \hat R_{d\ell}^\dagger$
denote the connected components of the graph
(i.e. the single linkage clusters).
Figure \ref{fig::rips} shows a Rips graph.

\medskip
Optionally, we may also want to discard small
components.
Thus we can discard a connected component
if $n_j \geq N_d$ where $n_j$ is the number of data points that
ended up in component $\hat R_{dj}^\dagger$
and $N_d$ is a user-chosen constant.
The tuning parameter $\epsilon$ is discussed 
in Section \ref{section::tuning}.

\subsection{Choosing the Tuning Parameters}
\label{section::tuning}

Here we provide suggestions for choosing all the tuning parameters.
Our suggestions are meant to be simple
and transparent.
We make no claim that they are optimal.
Singular feature finding is an exploratory
method and the user is encouraged to try different tuning parameters.
Our suggestions can be used as starting points.

\begin{enum}
\item {\em The Bandwidth $h$.}
Our default is given in equation (\ref{eq::silverman}).
(In some cases, we have found that using $h/2$ can improve the performance.)

\medskip
\item {\em The Signature Threshold $T_d$}.
A simple method that we have found very effective is to
plot a histogram of the signature values and
choose $T_d$ to be at a minimum.
Alternatively, we plot the empirical cdf of the signature values
and choose $T_d$ to be in a flat spot of the cdf.
Specifically, we compute the signature at each data point and we find the
empirical cdf of these numbers (as in Figure 3).
We illustrate this method in Section \ref{sec::examples}.
It would not be difficult to automate this method.
For example, fit a one dimensional kernel estimator to the signatures 
and choose the rightmost local minimum.
In practice, choosing it by visually inspecting the cdf works very well.

An alternative, and more formal approach, is to 
use a null distribution.
We assume that the data are a sample of size $n$ from a distribution
on a known, compact set ${\cal X}$.
(Otherwise, the data can be truncated to some compact set.)
Draw $U_1,\ldots, U_n$ from a uniform distribution $\mathbb{U}$ 
on ${\cal X}$.
We regard $\mathbb{U}$ as a null distribution.
We run the SCMS algorithm on these points
(using the same bandwidth $h$ that is used on the real data)
and we compute the estimated signatures.
Let $\tau_d$ be the maximum value of the signatures.
This process is repeated many times to get an estimate of the cdf
$$
F_d(t) = \mathbb{U}(\tau_d \leq t)
$$
where $\mathbb{U}$ indicates that we are using the uniform distribution.
This defines, in a sense, a null distribution for the maximum signature.
One can use an upper quantile of $F_d$
as the threshold $T_d$.
This Monte-Carlo approach is simple but time consuming.

\medskip
\item {\em The Rips Parameter $\epsilon$.}
In Section \ref{section::singlelinkage}
we used a parameter $\epsilon$ for the Rips graph.
As above, 
let $U_1,\ldots, U_n$ be $n$ draws from a uniform distribution on ${\cal X}$.
Let $U_i'$ be the nearest neighbor to $U_i$.
We suggest setting
$$
\epsilon = \mathbb{E} ||U_i - U_i'||.
$$
The parameter $N_d$ can be chosen subjectively or omitted entirely.
Alternatively,
let $N$ be the size of the largest connected component
of the Rips graph on the uniform data
using threshold $\epsilon$.
The null simulation gives the distribution
$G(t) = \mathbb{U}(N \leq t)$.
We can then choose an upper quantile of $G$.
An alternative is to start with a small $\epsilon$ and then increase $\epsilon$
until the number of connected components stabilizes.
\end{enum}

\vspace{-.5cm}
\section{Examples}
\label{sec::examples}

\subsection{Two-dimensional Example.}
\label{sub::Ex1}

Fig \ref{fig::2DMarcoA} shows an example of a two dimensional 
dataset hiding five singular features. Our method finds four
$0$-dimensional features (modes) and one $1$-dimensional feature
(filament). 

\vspace{-1cm}
\begin{center}
\begin{tabular}{ccc}
\includegraphics[scale=.3]{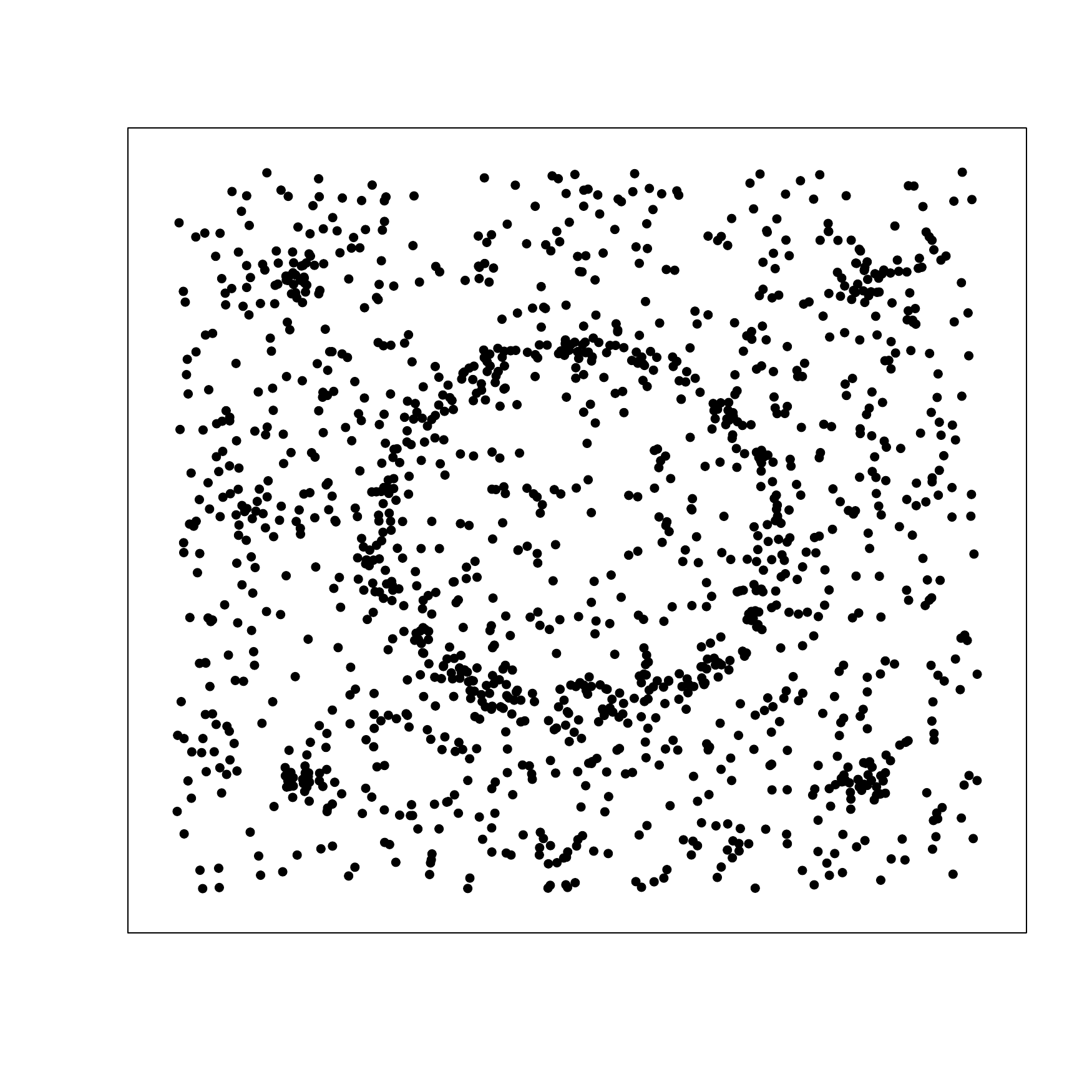} &
\includegraphics[scale=.3]{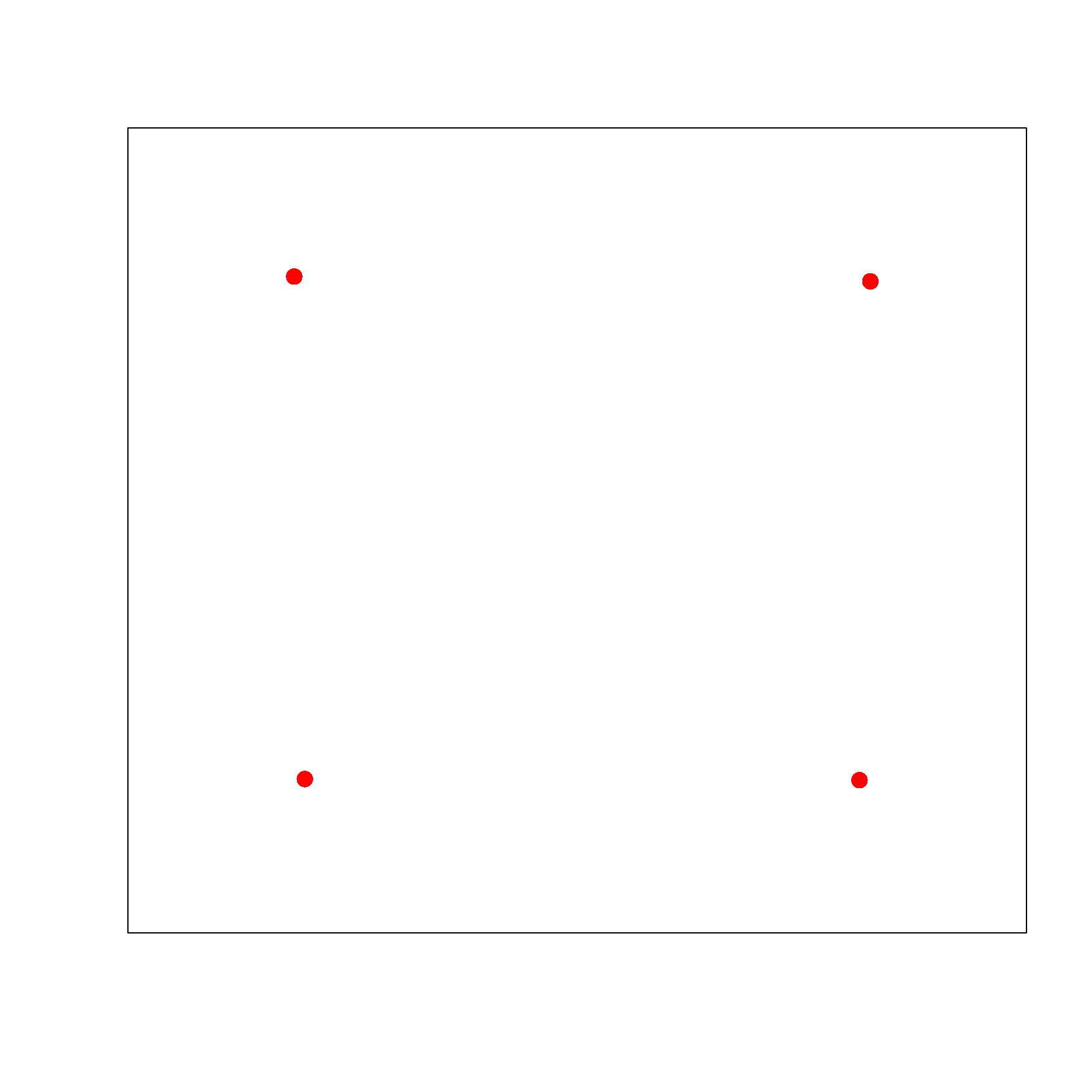}   &
\includegraphics[scale=.3]{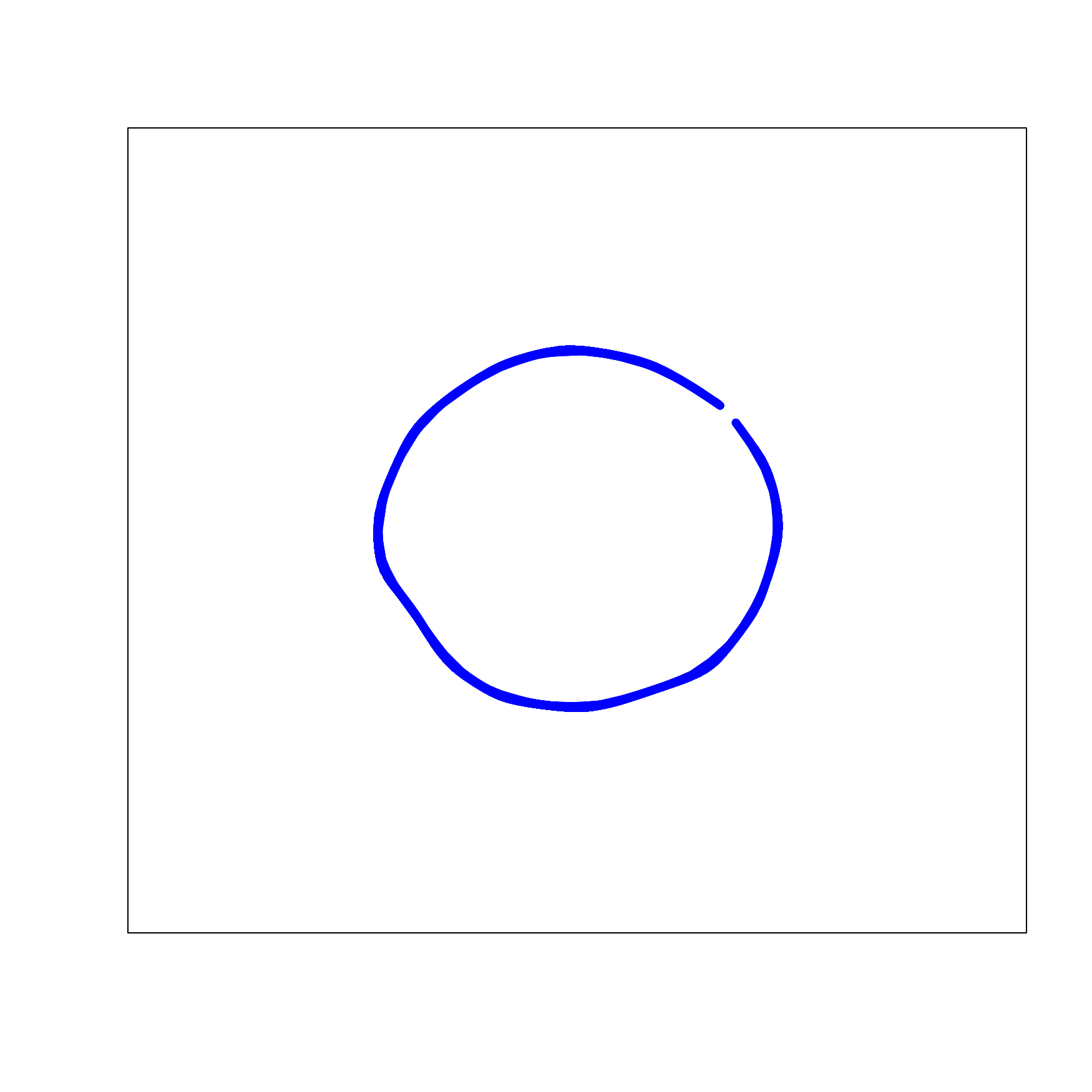}   
\end{tabular}
\end{center}

\begin{center}
\captionof{figure}{\em Two Dimensional example. This plot shows 
a two-dimensional dataset and the $0$- and $1$-dimensional singular
features found with our method.} 
\label{fig::2DMarcoA}
\end{center}

In Figure \ref{fig::2DMarcoB} we plot histograms and 
empirical cdf's of $0$- and $1$-dimensional signatures, together 
with their thresholds $T_d$ $(d=1,2)$ determined by 
the null distribution approach.

\begin{center}
\begin{tabular}{cccc}
\includegraphics[scale=.2]{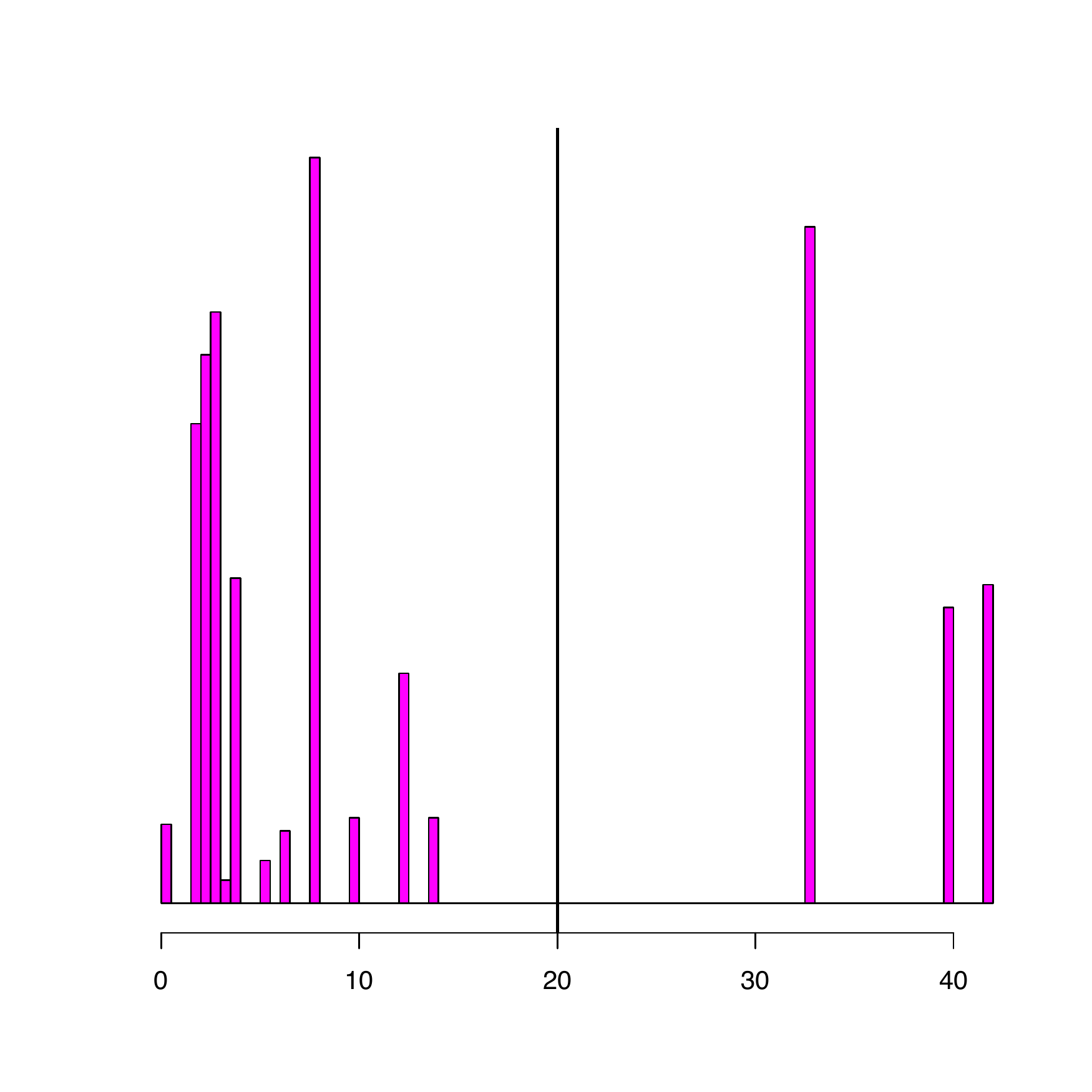}   &
\includegraphics[scale=.2]{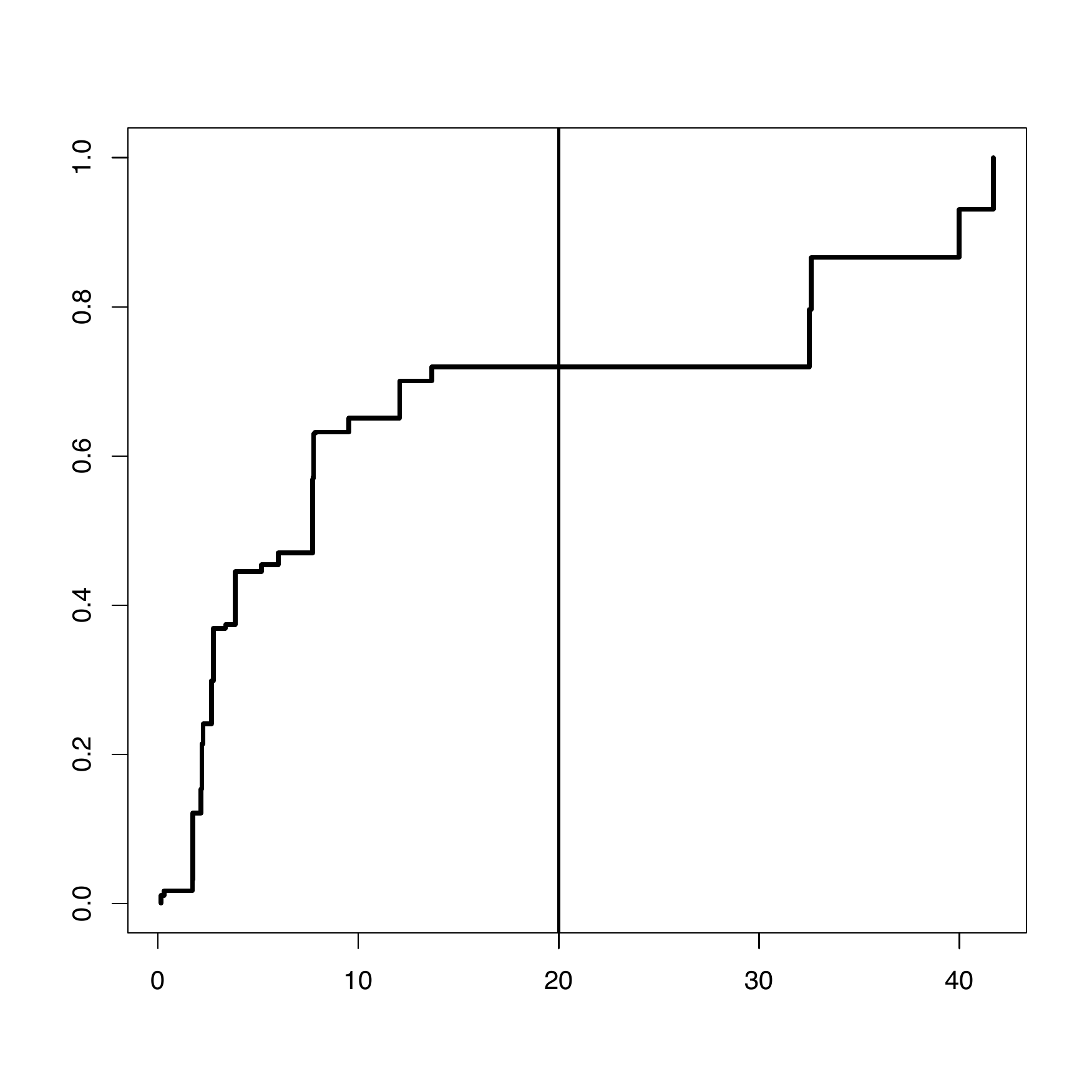}   &
\includegraphics[scale=.2]{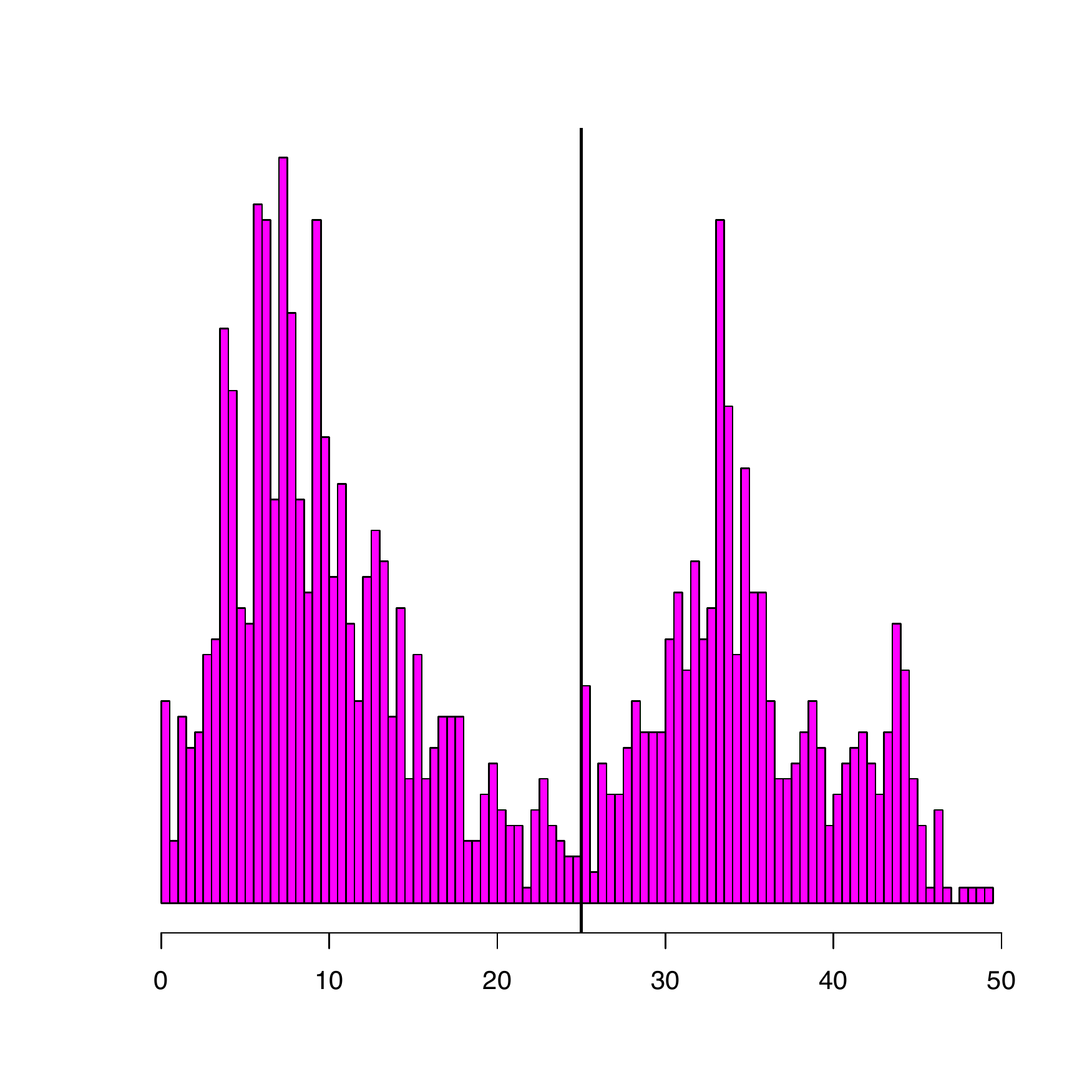}   &
\includegraphics[scale=.2]{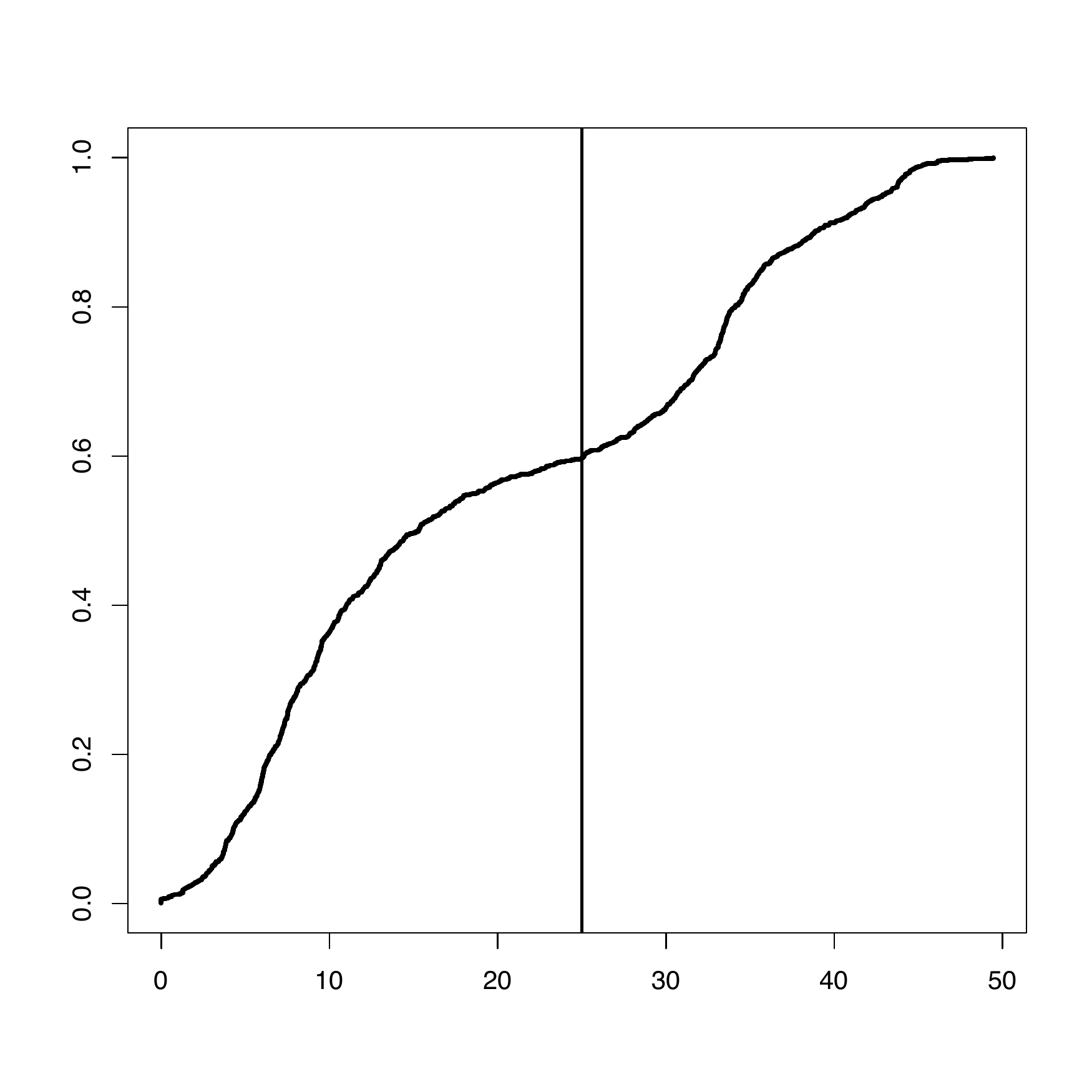}   
\end{tabular}
\end{center}

\begin{center}
\captionof{figure}{\em Two Dimensional example. The four plots
are the signatures' histogram and cdf for $0$-dimensional 
(two left plots) and $1$-dimensional (two right plots) singular 
features. The vertical lines in the four plots are the signature
thresholds $T_d$ for $d=0$ and $d=1$}
\label{fig::2DMarcoB}
\end{center}

\subsection{Three-Dimensional Example.}

This example is similar to the plots shown in Figure
\ref{fig::main-example}, but finding the hidden one-dimensional structure is 
more challenging.
The five singular features hidden in the dataset were found using
the threshold with the null distribution approach.

\vspace{-1cm}
\begin{center}
\begin{tabular}{cc}
\includegraphics[width=0.5\textwidth]{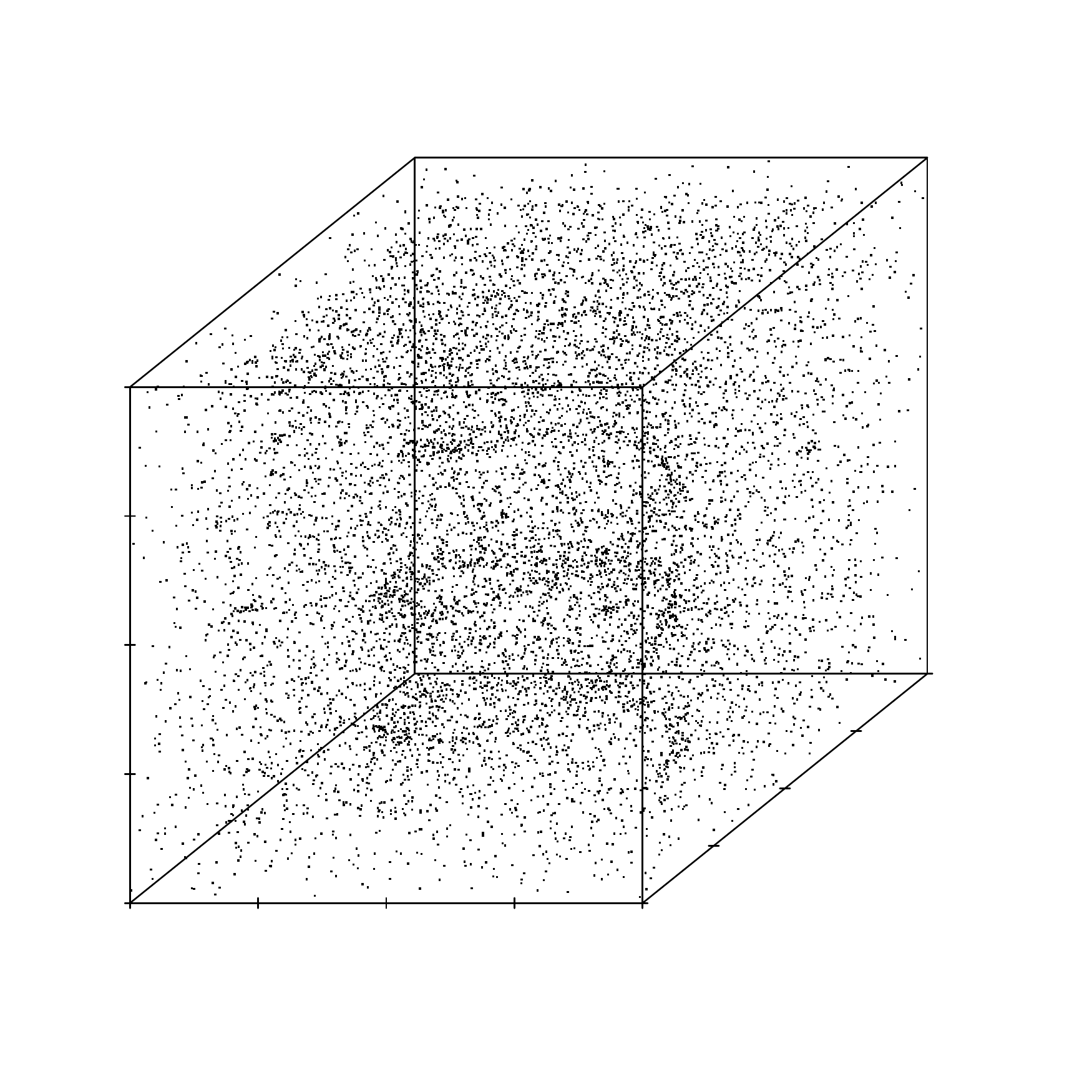} 
\hspace{-3.5em} & \hspace{-3.5em} 
\includegraphics[width=0.5\textwidth]{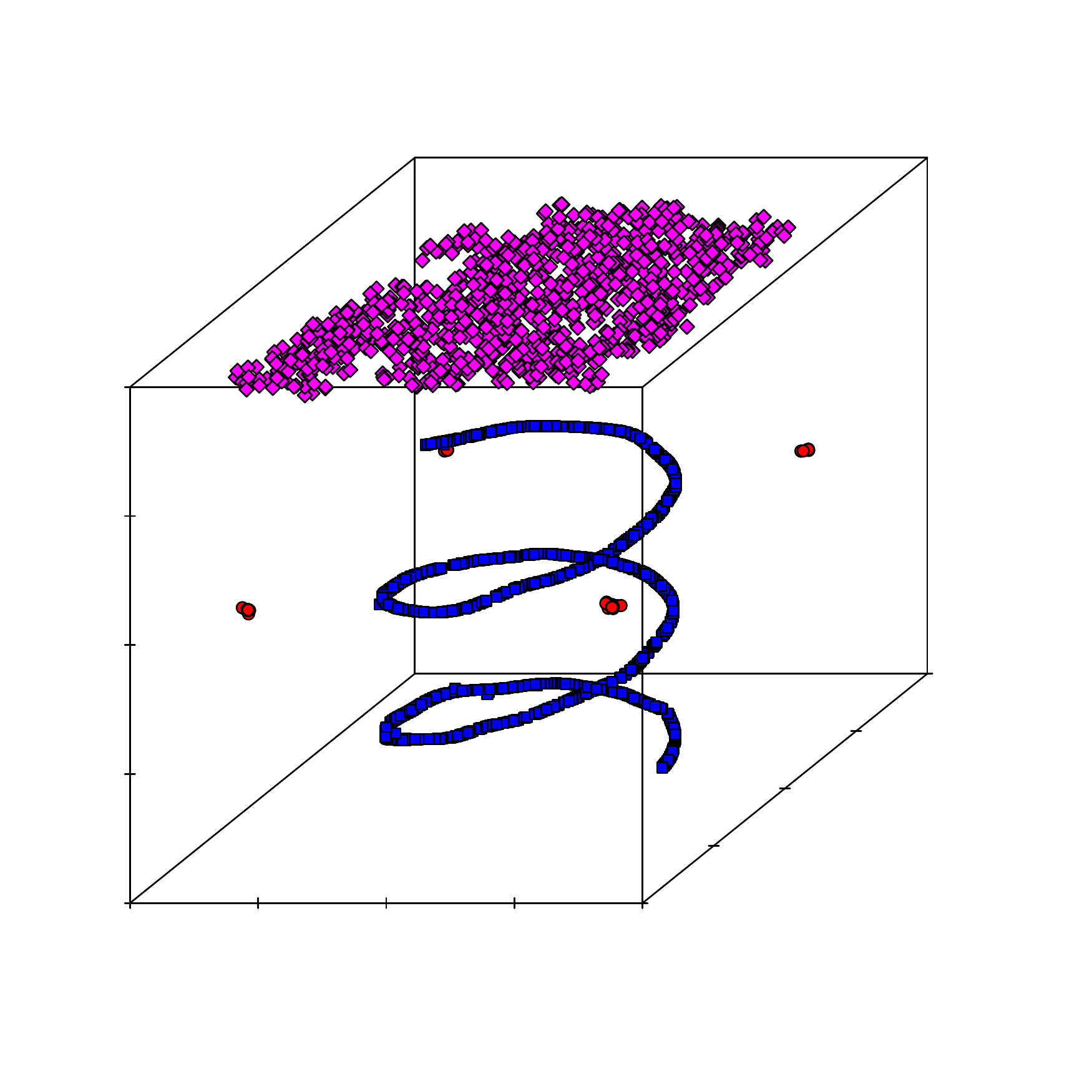} 
\end{tabular}
\vspace{-.6cm}
\captionof{figure}{\em Left: Dataset 
Right: our algorithm extracts four 0-dimensional  
structures (modes), one 1-dimensional structure 
(a corkscrew) and one 2-dimensional structure (plate)
from the point cloud.}
\label{fig::3DIsa}
\end{center}

\subsection{Spectral Clustering Example}

In Figure \ref{fig::spectral}
we analyze the data from the example in \ref{sub::Ex1} using 
spectral clustering.
There are many different versions of spectral clustering.
We use the default method in the R package {\tt kernlab}.
The left plot shows a dataset with no noise.
The colors correspond to the clusters.
In this case, the clustering is perfect.

\vspace{-.6cm}
\begin{center}
\begin{tabular}{cc}
\includegraphics[scale=.3]{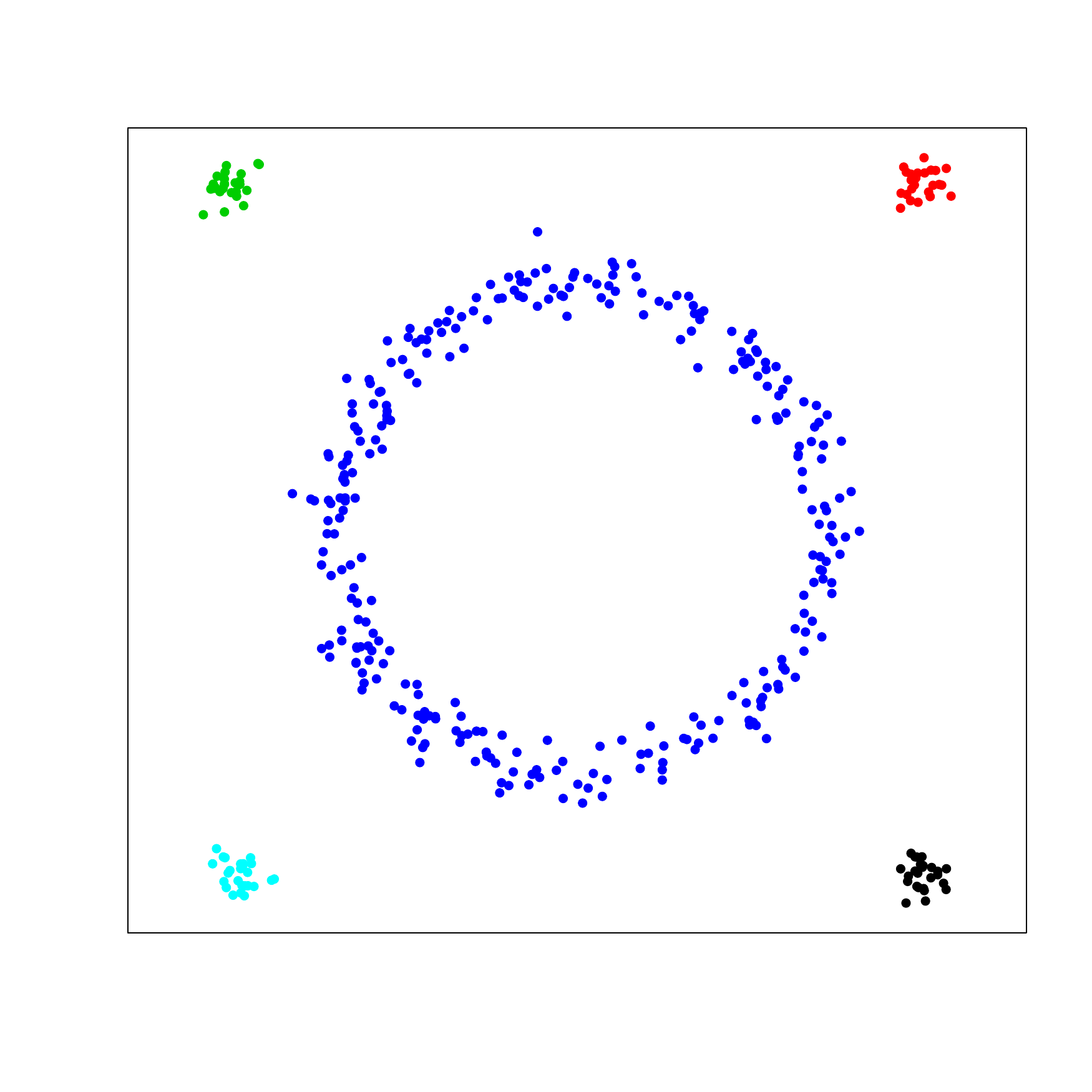} & 
\includegraphics[scale=.3]{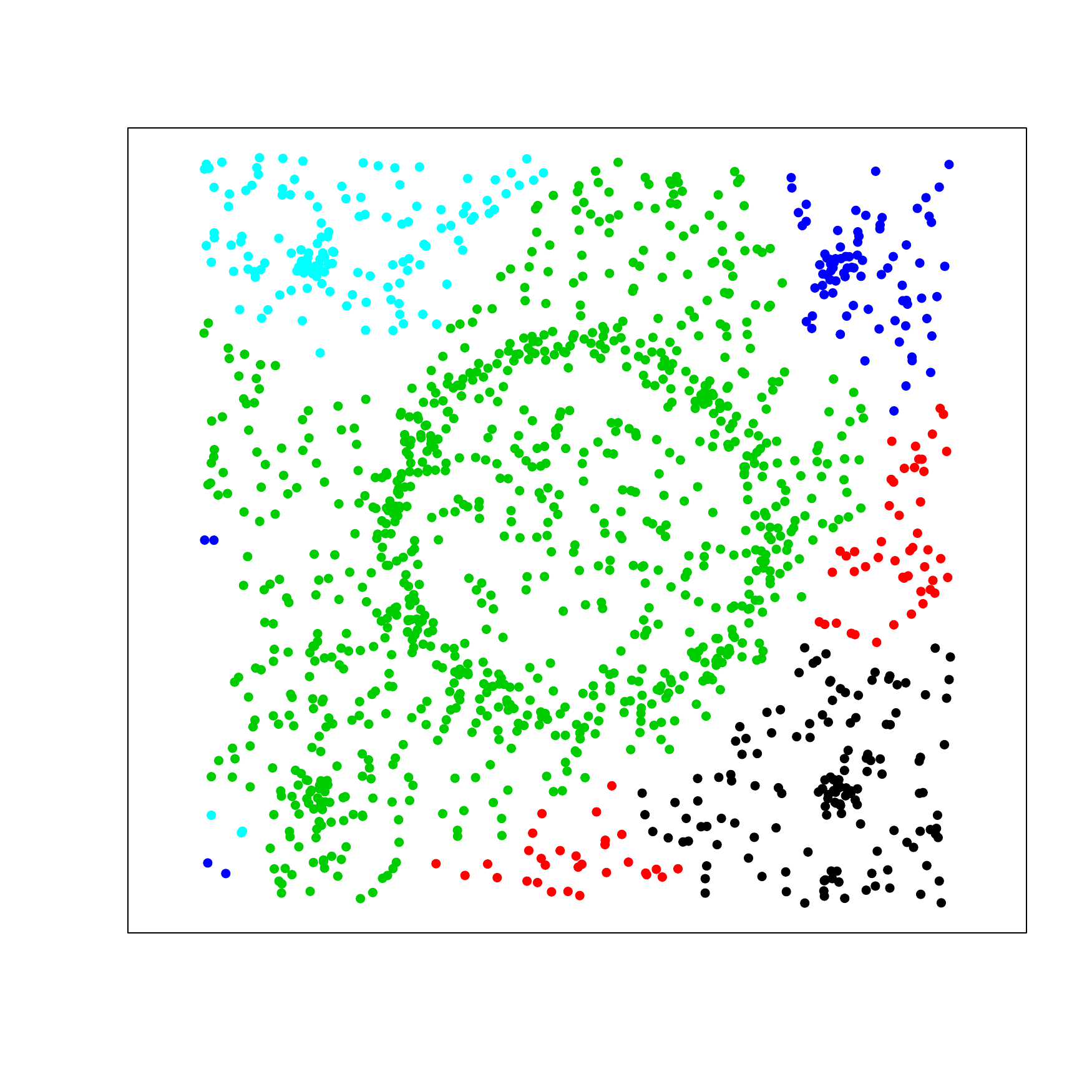}
\end{tabular}

\vspace{-1cm}
\captionof{figure}{\em Spectral clustering applied to the data 
from Example 1. The left plot shows a dataset with no noise.
The colors correspond to the clusters.
In this case, the clustering is perfect.
The right plot shows what happens when we add noise.
In this case spectral clustering fails.}
\label{fig::spectral}
\end{center}

The right plot shows what happens when we add noise.
In this case, spectral clustering fails.
In both cases, the output is two dimensional; spectral clustering does not
output zero or one dimensional features.
Also, we had to specify the correct number of clusters by hand.
Thus we see several advantages of singular feature finding;
it handles noise, it separates clusters by dimension,
and it does not require that we specify the number of clusters.

\subsection{Intersecting Curves}

Figure \ref{fig::intersecting}
shows an interesting example from 
Arias-Castro, Lerman and Zhang (2103).
The data, shown in the top left plot,
fall on two intersecting curves.
The distribution is in fact singular.

Arias-Castro, Lerman and Zhang (2103)
develop a method based on local PCA for 
analyzing such data.
Here we are interested in what happens if we apply singular feature finding.
We focus on $d=1$ as this is obviously the case of interest here.
The top middle plot shows that the method

%\vspace{-.2cm}
\begin{center}
\begin{tabular}{ccc}
\includegraphics[scale=.25]{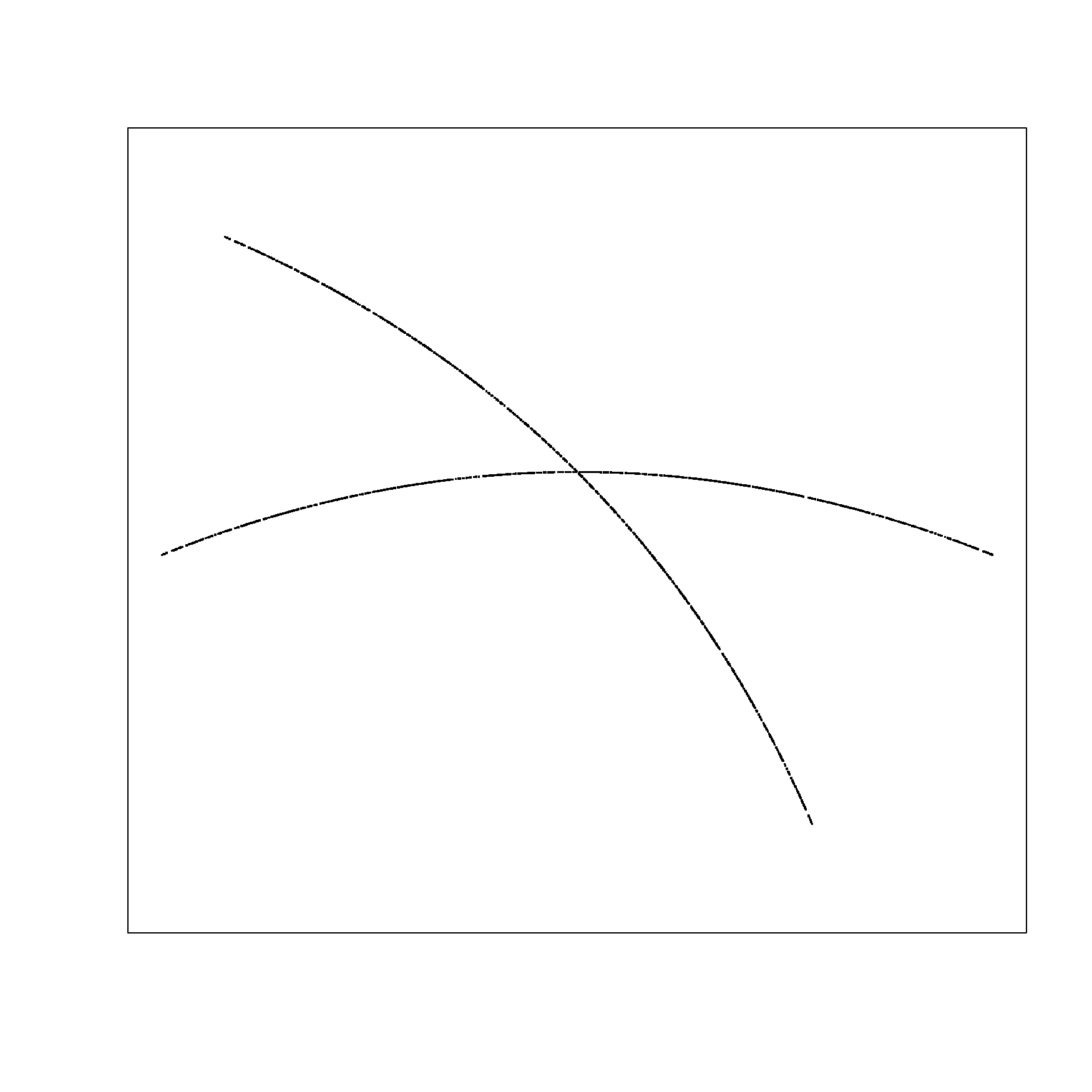} & 
\includegraphics[scale=.25]{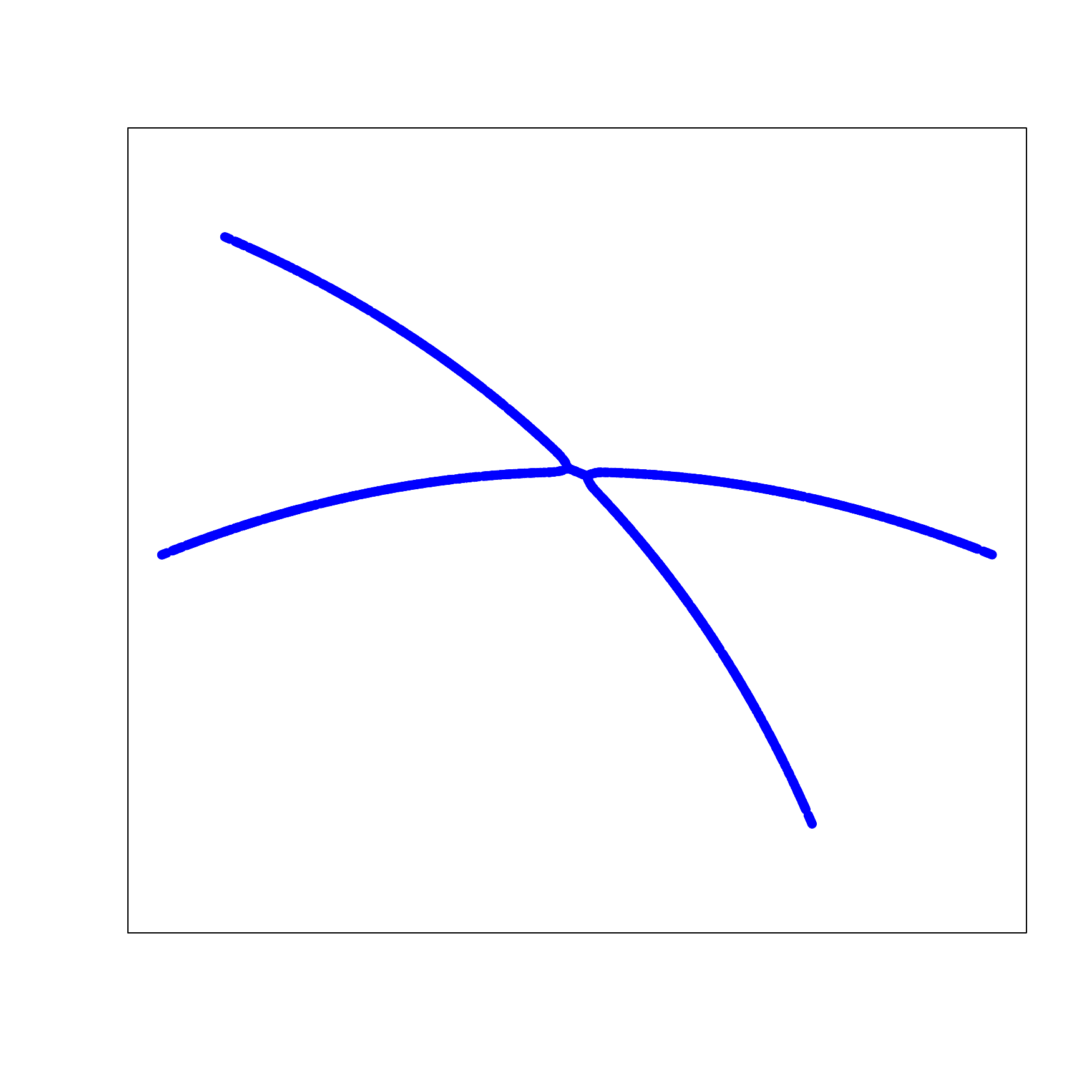} &
\includegraphics[scale=.25]{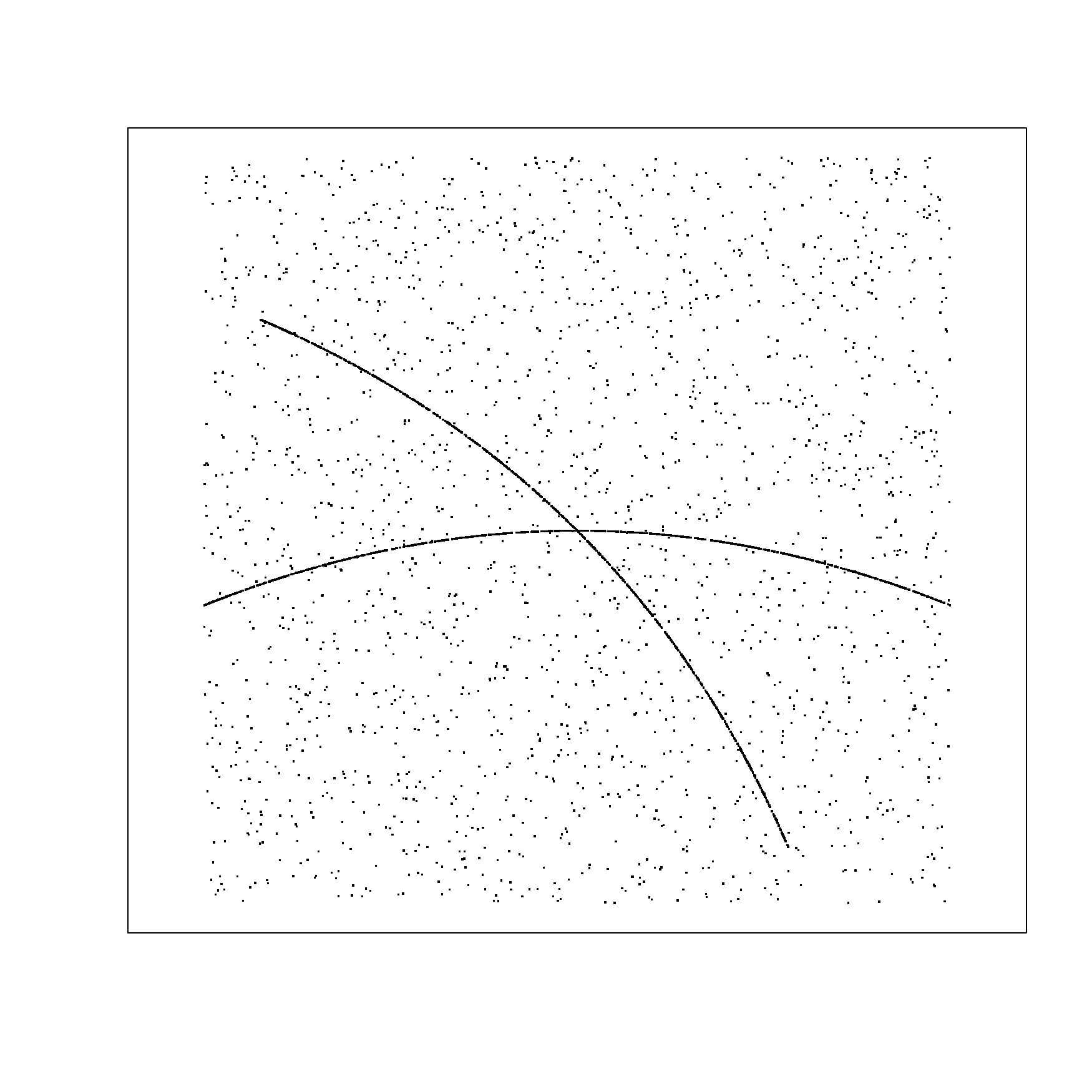} \\ 
\includegraphics[scale=.25]{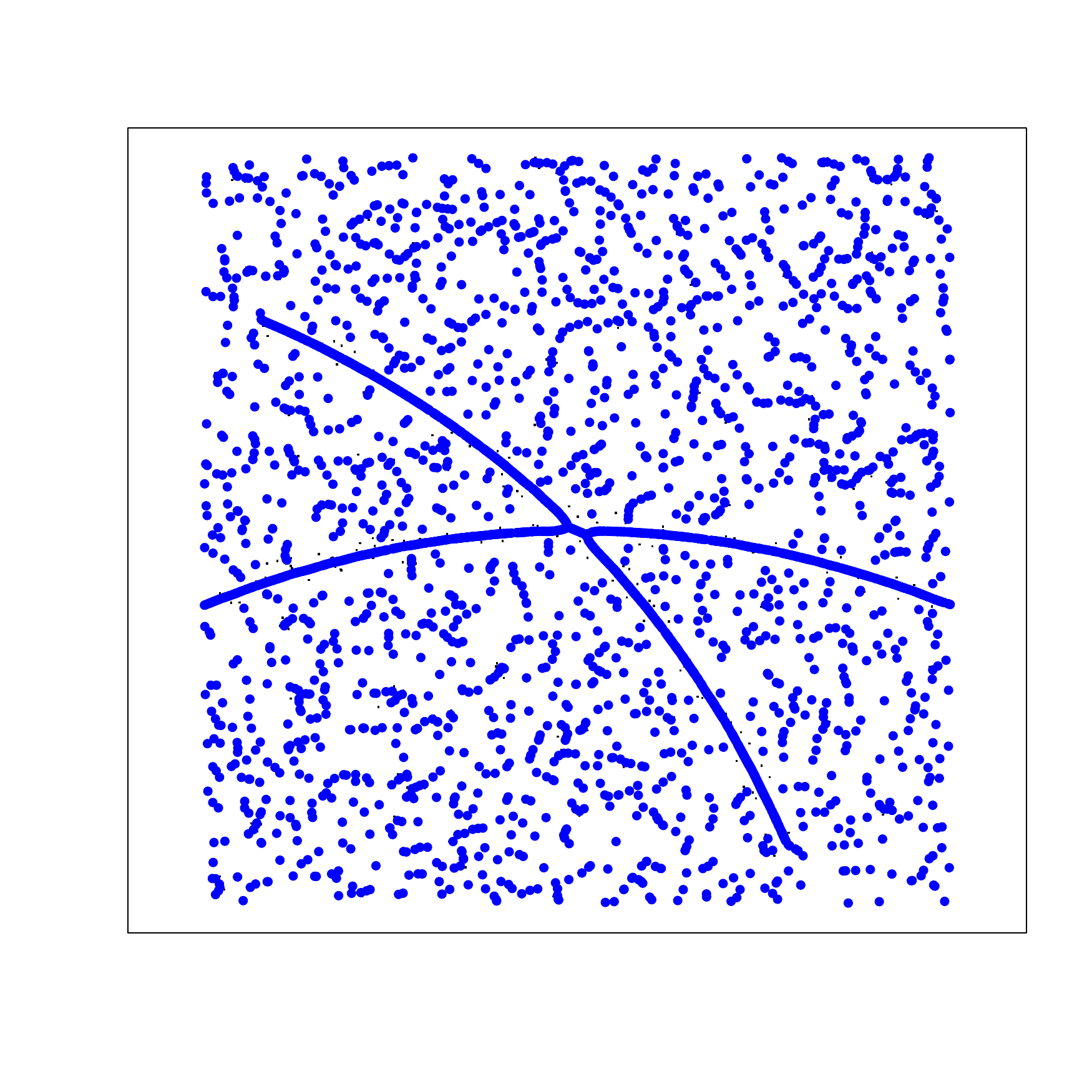}&
\includegraphics[scale=.25]{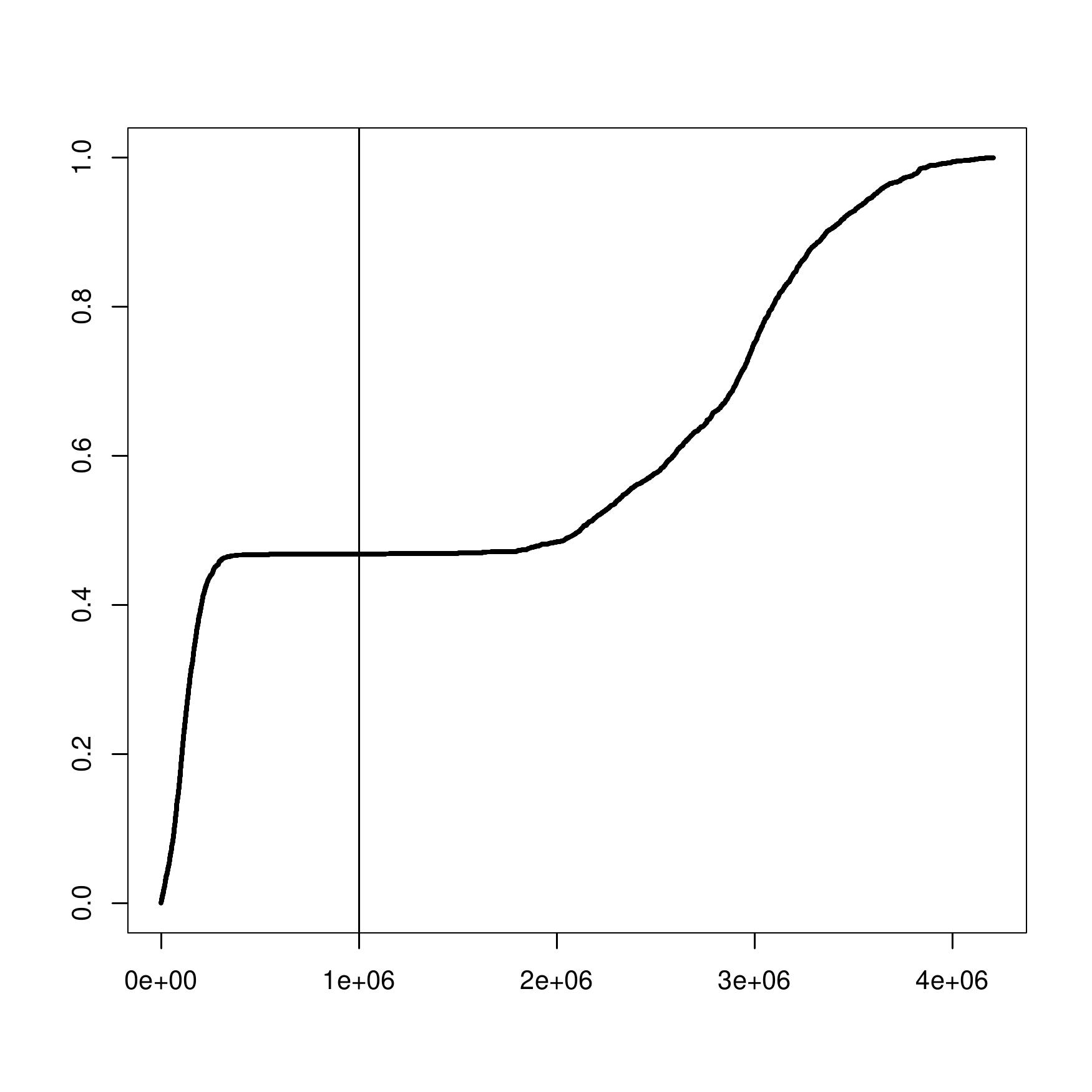} &
\includegraphics[scale=.25]{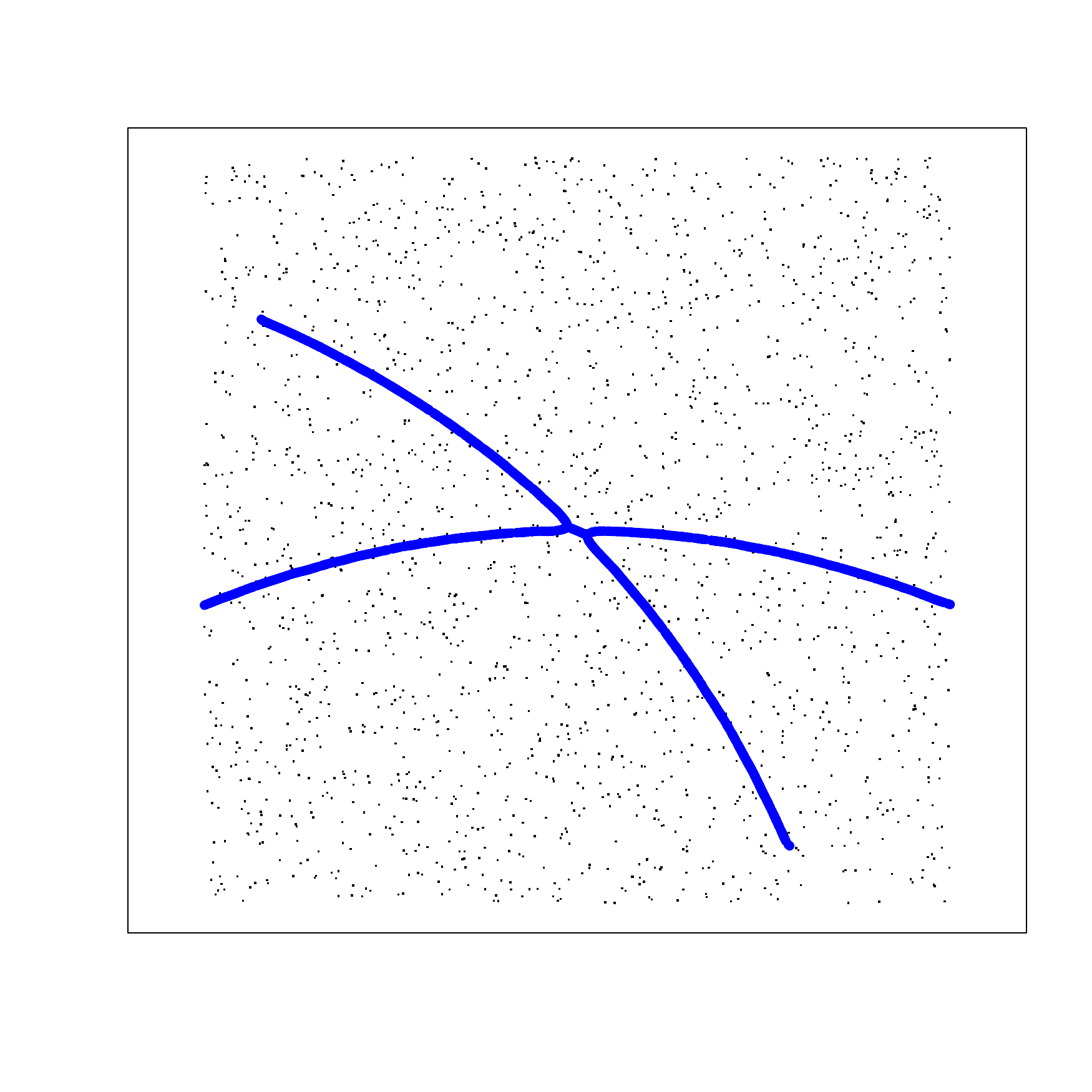}
\end{tabular}
%\vspace{-.2cm}
\captionof{figure}{\em The Intersecting Curves Examples.
TOP. Left: noiseless data. Middle: the resulting singular feature.
Right: noisy data. BOTTOM. Left: singular feature without 
signature thresholding. Note the many spurious points (dark dots),
identified as dimension zero features. Middle: cdf of the signatures. 
Right: singular feature after thresholding. Here the noise
points (light dots) are included in the plot but they are not part of $\hat R_1^\dagger$.}
\label{fig::intersecting}
\end{center}

returns a single feature which follows the data quite closely. 
It is not obvious from the plot 
but there is a small bias at the intersection.
Choosing the bandwidth slightly smaller than the Silverman rule reduces this bias.
In fact, the local PCA method of the aforementioned authors avoids this bias
as it was designed for datasets of this form.
Nevertheless, singular feature finding works quite well.
Next we add noise as shown in the top right plot.
If we blindly apply singular feature finding we get the main cluster 
plus many modes as shown in the bottom left plot.
The bottom middle plot shows the cdf of the signatures.
There is a clear flat spot in the cdf and if we threshold at 
some value in this zone we get the feature shown in the bottom 
right plot. Thus, singular feature finding with signature thresholding
work very well. Note that the noise points (light dots) are 
included in the last plot but they are not part of $\hat R_1^\dagger$.

\subsection{Mt. St. Helens Earthquake Dataset}

Figure \ref{fig::quakes} shows the Mt. St. Helens 
earthquake data, taken from  \cite{scott1992multivariate},
\cite{scott2015multivariate}, and 
\cite{duong2008feature}. The data are available at \\
{\tt http://www.stat.rice.edu/~scottdw/ALL.DATASETS/earthquake}.\\
The data consist of measurements of the epicentres of $510$ earthquakes 
which took place beneath the Mt. St. Helens 
volcano before its 1982 eruption. As in Duong et al.,
we take the first three variables longitude (degrees), 
latitude (degrees) and depth (km), from the full set of 
$5$ variables. Following Scott, the depth variable $z$ is 
transformed to $-\log(-z)$, where negative depths indicate 
distances beneath the Earth's surface.

Figure \ref{fig::quakes} shows the $3$-dimensional
data set together with the detected singular features.
We found three singular features of dimension $0$ (maxima) 
and two singular features of dimension $1$ (filaments). 
The filaments shown in the right plot suggest the existence 
of a depth fault-line in the data. We did not find any singular 
feature of dimension $2$.

\begin{flushleft}
\vspace{-.3cm}
\begin{tabular}{cc}
\includegraphics[scale=.5]{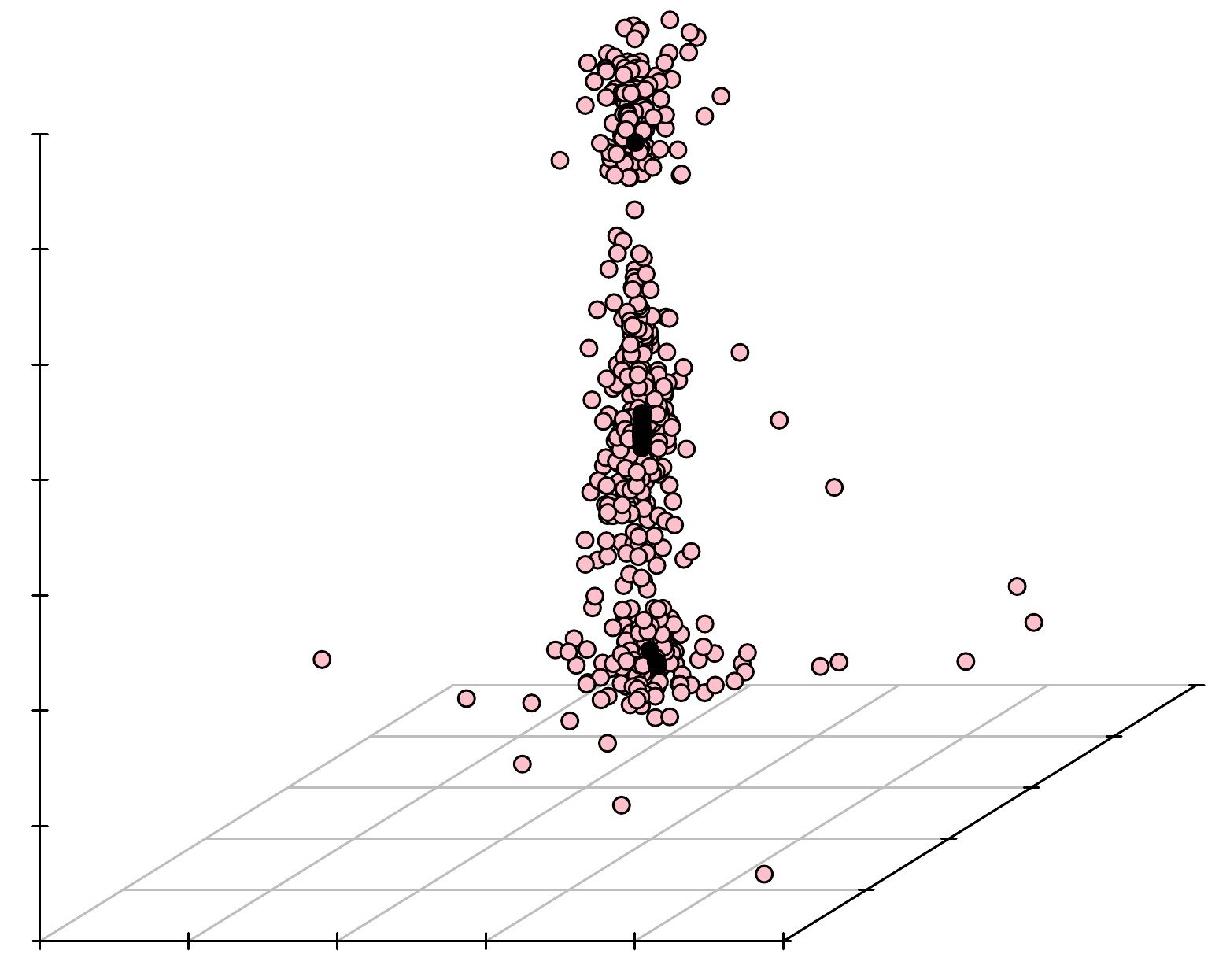} &
\includegraphics[scale=.5]{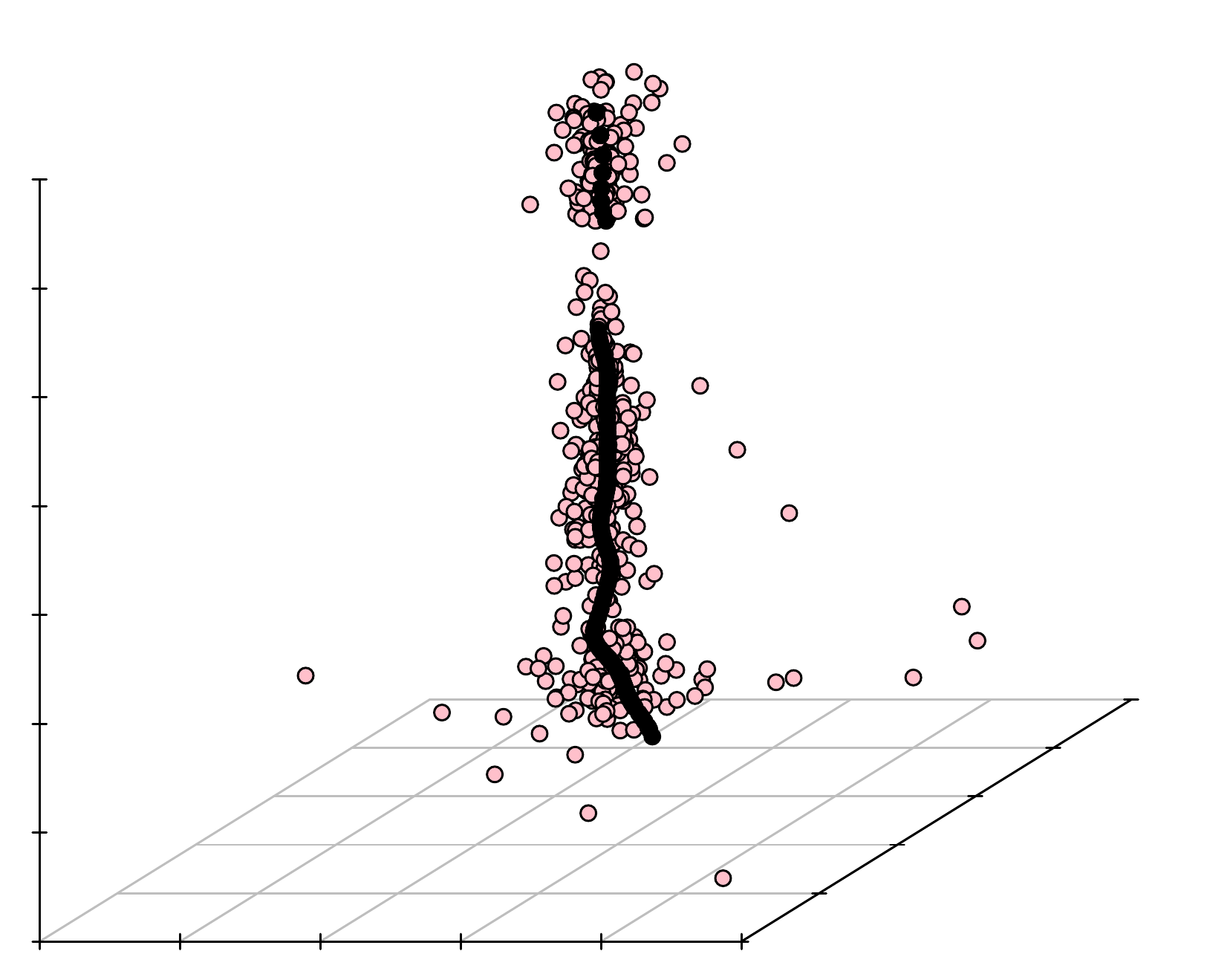} 
\end{tabular}
\captionof{figure}{\em Mt. St. Helen hearthquakes. 
Longitude is shown on the $x$-axis, latitude on the 
$y$-axis, and depth on the $z$-axis. Left: black 
dots show singular features of dimension $d=0$. 
Right: black dots show singular features of 
dimension $d=1$. Singular features of 
dimension $d=2$ were not found.
}
\label{fig::quakes}
\end{flushleft}

\subsection{Voronoi Foam}

Voronoi foam models
are simulations models invented in 
\cite{icke1991galaxy, van1994fragmenting, van2009cosmic}.
These models are meant to mimic astronomical galaxy surveys
by having a mixture of modes, filaments, walls and noise.
The data are created as follows.
We select $N$ random points in the cube $[-1,1]^3$.
We then form the Voronoi partition defined by these points.
The partition is a set of polygons which intersect at points, 
lines and walls. Random points are added at these points, 
lines and walls.

\vspace{-.5cm}
\begin{center}
\begin{tabular}{cc}
\includegraphics[scale=.6]{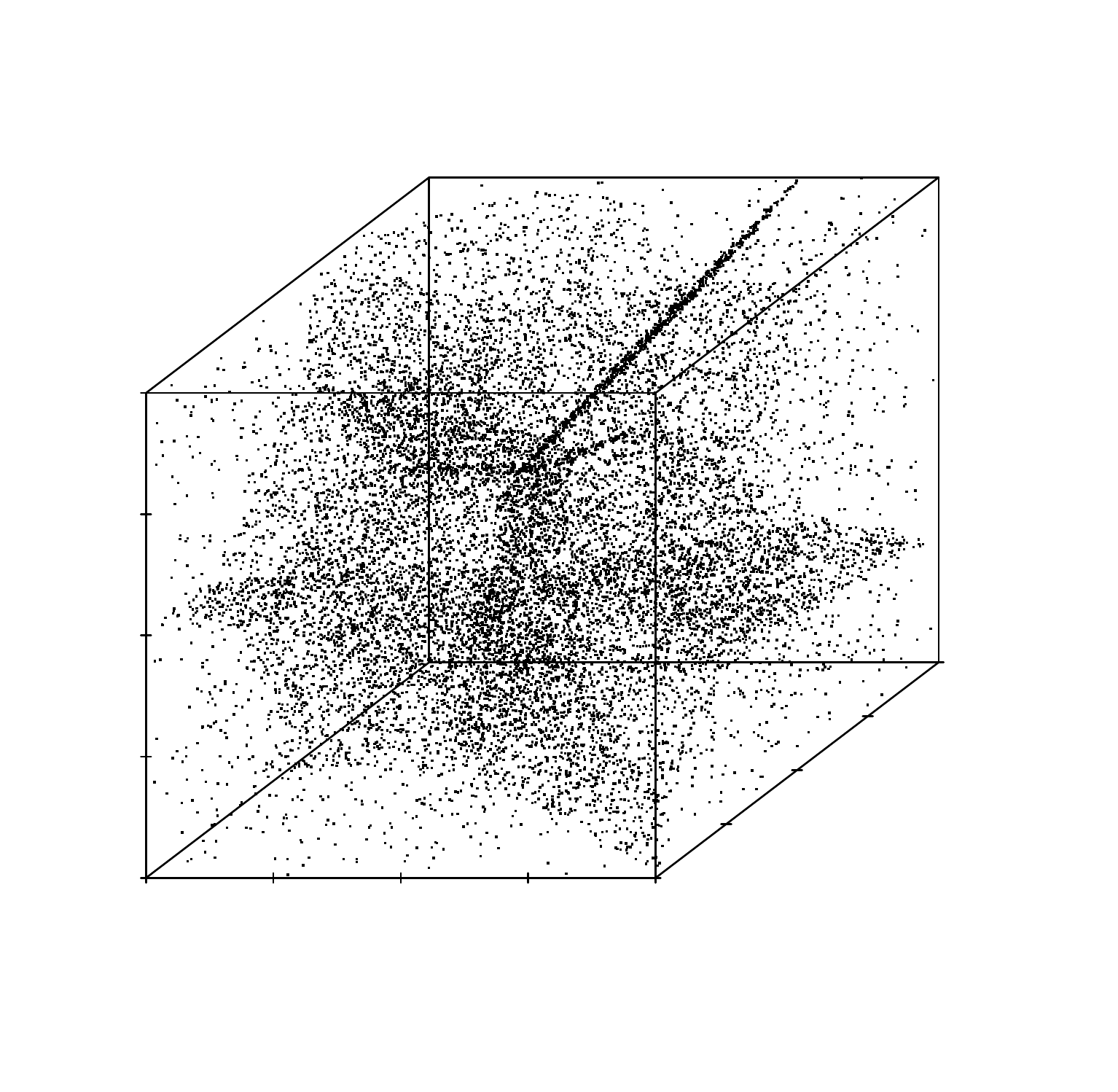} \hspace{-2em} & \hspace{-2em} 
\includegraphics[scale=.6]{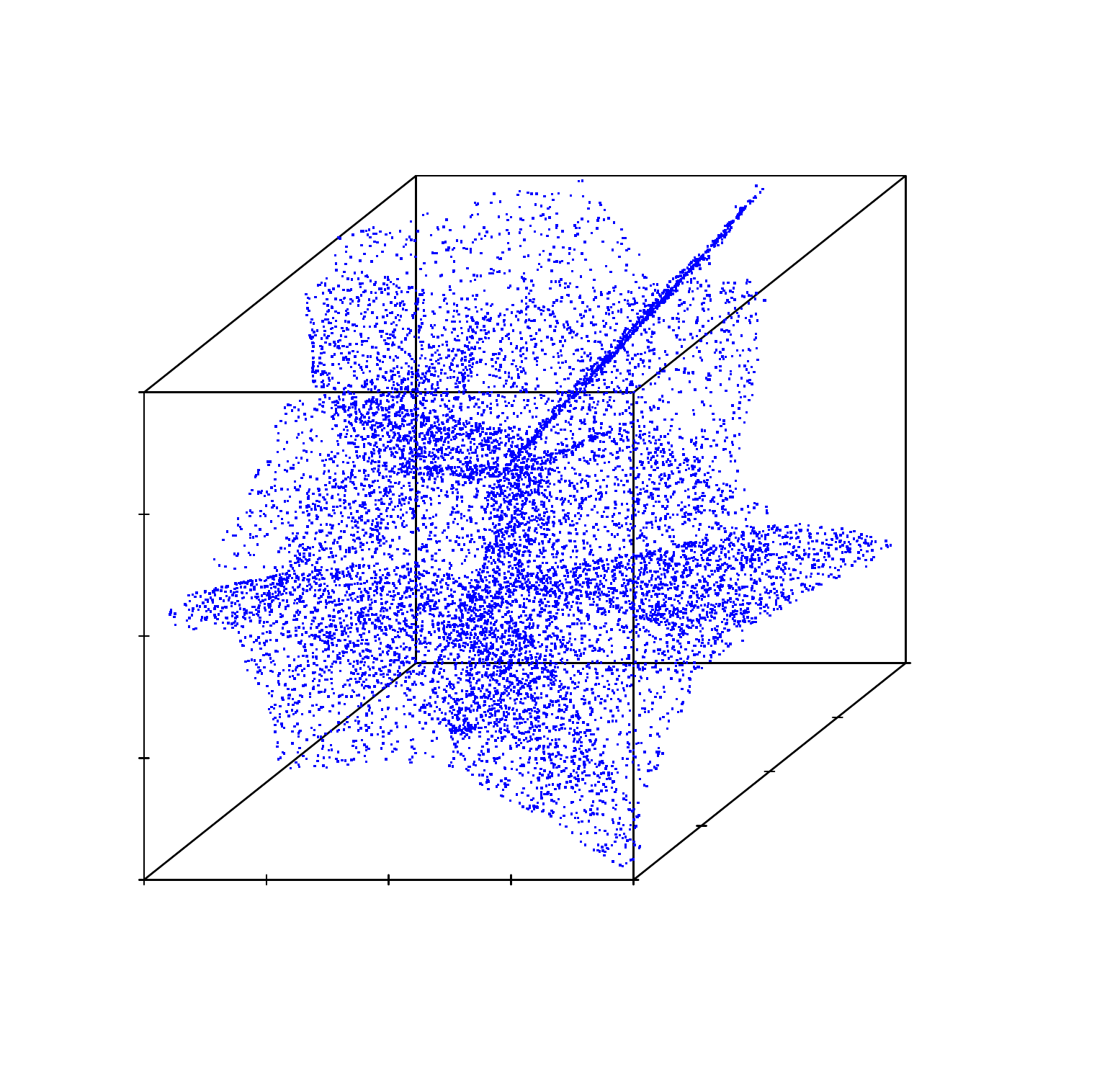} 
\end{tabular}
\vspace{-.4in}
\captionof{figure}{Left: data from a Voronoi foam model generating walls.
Right: High signatures ridge points ($d=2$). 
Althought it might be difficult to see in the plots, note that in 
the right plot the noise is gone and the walls remains.}
\label{fig::foam}
\end{center}

\medskip
By varying the proportion of points on these features on can control 
the amount of modes, filaments and walls.
Finally, clutter points are added which are drawn from a uniform on 
$[-1,1]^3$.
We should point out that this is the simplest Voronoi foam model;
more sophisticated version that also include dynamics (moving galaxies)
can be found in 
\cite{icke1991galaxy, van1994fragmenting, van2009cosmic}.
Figure \ref{fig::foam} shows a dataset generates this way.
In this example, we will focus on wall finding ($d=2$).
To the best of our knowledge, there are no existing wall finding algorithms
that actually output structures of dimensions $d=2$.
The right plot shows the result after filtering by the signatures.
The blue points are the ridge points with a large signature.
We see that the algorithm does a nice job of finding the dominant structures.
It may be difficult to see in the plot,
but in the right plot, the noise is gone and the walls remain.

\section{Asymptotic Theory}
\label{sec::theory}

Here we study the asymptotic properties of the procedure.
We build on results from
\cite{us::2014.ANN}.
To state the results,
we need a few definitions and regularity conditions.
Let
$H'(x) = d{\rm vec}(H(x))/d x^T$ and
let
$\lambda_1(x) \geq \cdots \lambda_D(x)$
be the ordered eigenvalues of $H(x)$.
Let
$V(x)$ be the $(D-d)\times D$ matrix with columns that are the
$D-d$ eigenvectors
corresponding to the $D-d$ smallest eigenvalue.
Let
$L(x) = V(x) V(x)^T$ be the projection matrix corresponding to this subspace.
Recall that $R_d^\dagger(t) = \{x\in R_d:\ S_d(x) \geq t\}$.
Let $B(x,r)$ denote a ball of radius $r$ centered at $x$.
Let
$$
A \oplus r = \bigcup_{x\in A} B(x,r).
$$
The Hausdorff distance ${\rm Haus}(A,B)$ between
two sets $A$ and $B$ is defined by
$$
{\rm Haus}(A,B) = \inf\Bigl\{\epsilon:\ 
A \subset B \oplus \epsilon\ {\rm and}\  
B \subset A \oplus \epsilon \Bigr\}.
$$

We assume there exist
$\beta,\delta,\epsilon , C >0$ such that the following assumptions hold:

(A1) $p$ is supported
on a compact set ${\cal X}\subset \mathbb{R}^D$
and is bounded and continuous and has five, bounded, continuous, square integrable derivatives.
The first, second and third derivatives of $p$ vanish on the boundary of ${\cal X}$.

(A2) The kernel $K$ satisfies the regularity conditions 
in \cite{gine2002rates}, 
\cite{Arias-Castro2013} and
\cite{chacon2011asymptotics}.

(A3) For all $x\in R_d\oplus \delta$,
$\lambda_{d+1}(x) < -\beta$ and
$\lambda_d(x) -\lambda_{d+1}(x) >\beta$.

(A4) For each $x\in R_d\oplus \delta$,
$$
|| (I-L(x))\, g(x)||\ ||{\rm vec}(H'(x))||_{\infty} < \frac{\beta^2}{2 D^{3/2}}.
$$

(A5)
The eigenvalues 
$\lambda_1(x),\ldots, \lambda_D(x)$
are distinct.

(A6) 
Let $T_d>0$ be a fixed constant.
For all
$-\epsilon < \gamma < \epsilon$ we have
${\rm Haus}\Bigl(R_d^\dagger(T_d + \gamma),R_d^\dagger(t)\Bigr)\leq C \gamma$.

(A7)
If $x\in R_d$ but $x\notin R_d\oplus \delta$,
then $S_d(x) < T_d - \epsilon$.

Assumptions (A1)-(A2)
ensure the consistency of the estimates of the density and its derivatives.
The regularity conditions on the kernel in (A2) are standard and we have not listed them to save space.
Assumptions (A3) and (A4)
are from \cite{us::2014.ANN}.
These conditions imply that the ridge is well defined.
These assumptions are needed so that $R_d$ can be estimated.
Assumption (A5) is probably stronger than needed.
It implies that the eigenvalues of the Hessian 
are continuously differentiable functions.
It may be possible to relax this condition
but the analysis would become much more complicated.
(A6) says that the set
$R_d^\dagger(t)$ changes continuously in $t$.
Condition (A7) says that points far from the ridge have a small signature.

Let us define
\begin{equation}
\psi_n = h^2 + \sqrt{\frac{\log n}{n h^{D+6}}},\ \ \ \ \ 
r_n    = h^2 + \sqrt{\frac{\log n}{n h^{D+4}}}.
\end{equation}
The following lemma follows from
Lemma 3 of 
\cite{Arias-Castro2013}
and Theorem 4 of
\cite{chacon2011asymptotics}.

\begin{lemma}
Suppose that
$h\to 0$ and
$n h^{D+6}/\log n \to \infty$.
Assume (A1) and (A2)
There exists $0 < b < 1$ such that:
\begin{align*}
\sup_{(\log n/n)^{1/D} \leq h \leq b^{1/D}}
\sup_{x\in {\cal X}} |\hat p_h^{(\ell)}(x) - p^{(\ell)}(x)| & \leq c_1 h^2 + c_2 \sqrt{\frac{\log n}{n h^{D+2\ell}}}
\end{align*}
for $\ell = 0,1,2,3$.
\end{lemma}

\begin{lemma}
Assume (A1)-(A7).
Suppose that $h\to 0$ and
$n h^{2/(D+6)}/\log n \to \infty$.
Then
$$
\sup_x|\hat S_d(x) - S_d(x)| = O_P(r_n).
$$
\end{lemma}

\begin{proof}
From Lemma 1,
$$
\sup_x \max_{j,k} |\hat H_{jk}(x) - H_{jk}(x)| = O_P(r_n).
$$
Under the given assumptions,
$S_d(x)$ is a continuous, Lipschitz function
of $H$ and the result follows.
\end{proof}

\begin{theorem}
Assume (A1)-(A7).
Suppose that $h\to 0$ and
$n h^{2/(D+6)}/\log n \to \infty$.
Then, 
$$
{\rm Haus}(\hat R_d^\dagger,R_d^\dagger) = O_P(\psi_n + r_n).
$$
\end{theorem}

\begin{proof}
From conditions (A1)-(A4),
we can conclude from Theorem 5 of
\cite{us::2014.ANN}
that
${\rm Haus}(\hat R_d,R_d) = O_P(\psi_n)$.
Let $x\in R_d^\dagger$.
Then, on an event with probability tending to one,
there exists $y\in \hat R_d$ such that
$||x-y||\leq c\psi_n$
for some $c>0$.
Since
$x\in R_d^\dagger$,
we have that $S_d(x) \geq T_d$.
From the previous lemma, and the fact that $S_d$ is Lipschitz,
$$
\hat S_d(y) \geq S_d(y) - O_P(r_n) \geq S_d(x) - O_P(r_n)-O(||x-y||) \geq T_d - 
O_P(r_n)-O(\psi_n) \geq T_d - O_P(r_n+\psi_n).
$$
Thus,
$y\in \hat R_d^\dagger (T_d - O_P(r_n+\psi_n))$.
This shows that
$$
R_d^\dagger(T_d)\oplus O_P(\psi_n) \subset \hat R_d^\dagger(T_d - O_P(\psi_n+r_n)).
$$
By a similar argument,
$$
\hat R_d^\dagger(T_d +O_P(\psi_n+r_n)) \subset R_d^\dagger(T_d)\oplus O_P(\psi_n).
$$
From (A6), we conclude that
$$
{\rm Haus}(\hat R_d^\dagger,R_d^\dagger) = O_P(r_n+\psi_n).
$$
\end{proof}

Next, we find the asymptotic distribution and standard error of
$\hat S_d(x)$.
This is useful for forming tests or confidence intervals for $S_d(x)$.

\begin{theorem}
Assume the conditions Theorem 3.
At each point for which 
$\min_j |\lambda_j(x)| >0$ for $x\in R_d^\dagger$,
$$
\sqrt{n h^{(D+4)/2}}(\hat S_d(x) - \mathbb{E}[\hat S_d(x)])\rightsquigarrow N(0,\Gamma)
$$
where 
$$
\Gamma = J_s J_\lambda \Sigma(x) J_\lambda^T J_s^T,
$$
$\Sigma(x)$ is given in (\ref{eq::Sigma}),
$J_\lambda$ is given in (\ref{eq::Jl})
and
$J_s$ is given in
(\ref{eq::sig1}), 
(\ref{eq::sig2}), 
(\ref{eq::sig3}) and
(\ref{eq::sig4}).
If $h = o(n^{-1/(D+4)})$
then
$$
\sqrt{n h^{(D+4)/2}}(\hat S_d(x) -  S_d(x))\rightsquigarrow N(0,\Gamma).
$$
\end{theorem}

\begin{proof}
Theorem 3 of \cite{duong2008feature}
show that
$$
\sqrt{n} h^{(D+4)/2} ({\rm vech}(\hat H(x)) - \mathbb{E}[{\rm vech}(\hat H(x))])\rightsquigarrow N(0,\Sigma(x))
$$
where
\begin{equation}\label{eq::Sigma}
\Sigma(x) = R({\rm vech} \nabla^{(2)} K) p(x)
\end{equation}
where
$R(g) = \int g(x) g(x)^T dx$
and
${\rm vech}$ is the half-vectorizing operator that stacks the upper triangular
portion of a matrix into a vector.
The condition $\min_j |\lambda_j(x)| >0$ together with (A5) implies that $S_d(x)$ is
continuously differentiable of $H$.
Hence, the result follows from the delta method.
It remains to find $\Gamma$.

By the delta method,
the asymptotic covariance of $S_d(x)$ is
$T = J_s \ {\rm Cov}(\hat\lambda) \ J_s^iT$
where
${\rm Cov}(\hat\lambda)$ is the asymptotic
covariance of $\hat\lambda$ and
$J_s$ the $D \times D$ Jacobian of the signature function.
Now
${\rm Cov}(\hat\lambda) J_\lambda \ \Sigma \ J_\lambda$
where
$J_\lambda$ is the $D \times D^2$ Jacobian of the eigenvalue.
So
$\Gamma = J_s J_\lambda {\rm Cov}(\hat\lambda) J_\lambda^T J_s^T$.

%%%%Let $\hat p$ be a Kernel density estimator based on $D$-dimensional observations
%%%%$X_1, \cdots , X_n$
%%%%with a scalar bandwidth $h$ 
%%%%$$
%%%%\hat p (x) = \frac{1}{n} \sum_{i=1}^n f_i(x)
%%%%$$
%%%%where
%%%%$$
%%%%f_i(x) =
%%%%(2\pi)^{-D/2} h^{-D} exp \left\{ -\frac 1{2h^2} (x-X_i)^t  (x-X_i) \right\}.
%%%%$$
%%%%Its gradient and Hessian are respectively
%%%%\begin{align*}
%%%%\nabla \hat p(x) &= 
%%%%\frac 1{n h^2} \sum_{i=1}^n f_i(x) \cdot \left( (x-X_i)^t \right) 
%%%%\\
%%%%\hat H &=
%%%%\frac 1{n h^4} \sum_{i=1}^n f_i(x) \cdot \left(  (x-X_i) (x-X_i)^t  - h^2 \right) .
%%%%\end{align*}
%%%%
%%%%From now on we simplify the notation omitting the dependence from the point $x$.
%%%%We have defined the signature as
%%%%\begin{eqnarray*}
%%%%S_0 &=&  I(\lambda_{1} < 0) \cdot |\lambda_1|\frac{|\lambda_{1}|}{|\lambda_D|} 
%%%%\\
%%%%S_j &=&  I(\lambda_{j+1} < 0) \cdot |\lambda_{j+1}| 
%%%%\frac{|\lambda_{j+i}|}{|\lambda_D|} \prod_{i=1}^j \left( 1- \frac{|\lambda_i \wedge 0|}{|\lambda_D|} \right)
%%%%\qquad j=1, \cdots, D-1
%%%%\end{eqnarray*}
%%%%where $\lambda_1, \cdots , \lambda_D$ are the eigenvalues of the Hessian
%%%%in decreasing order.
%%%%

The signature vector $S= (S_0, \cdots S_{D-1})$ is function of the vector $\lambda$,
which is function of the Hessian.
We can write 
\begin{eqnarray}  \label{eqn:signif.func}
S = s(\lambda)  
\qquad 
&\mbox{with}& 
\quad
s: \Re^D \mapsto \Re^D \\
\label{eqn:eigenval.func}
\lambda = \ell ({\sf vec}(H)
\qquad 
&\mbox{with}& 
\quad
\ell : \Re^{D^2} \mapsto \Re^D.
\end{eqnarray}

From Chapter 8.8 of 
\cite{magnus1988matrix},
\begin{equation}\label{eq::Jl}
J_\lambda = \left(
u_1 \otimes u_1  | \quad \cdots \quad  | u_D \otimes u_D 
 \right)^T
\end{equation}
where $u_j$ is the normalized eigenvector associated with $\lambda_j$.

Next we compute $J_s$.
By direct computation,
the terms below the diagonal of the Jacobian are the 
partial derivatives of $S_j$ with respect to $\lambda_1, \cdots, \lambda_j$:
\begin{equation}\label{eq::sig1}
\frac{\partial S_j}{\partial \lambda_k} = 
S_j \cdot \frac{I(\lambda_k < 0)}{\lambda_k - \lambda_D} 
\end{equation}
if
$\lambda_{j+1} <0$ and $\lambda_k<0$
and is 0 if
$\lambda_{j+1} \geq 0$ or $\lambda_k \geq 0$.
The terms on the diagonal, excluding the last one, 
are the partial derivatives of $S_j$ with respect to $\lambda_{j+1}$ for $j =0, \cdots, D-2$:
\begin{equation}\label{eq::sig2}
\frac{\partial S_j}{\partial \lambda_{j+1}} =
S_j \cdot \frac{2}{\lambda_{j+1}}
\end{equation}
if $\lambda_{j+1} < 0$ and 0 otherwise.
Finally we consider the terms in the last column.
For $j < D-1$,
\begin{equation}\label{eq::sig3}
\frac{\partial S_{j}}{\partial \lambda_D} =
S_j \cdot \left(
\frac{-1}{\lambda_D} + \frac{1}{\lambda_D} \sum_{i=1}^j \frac{\lambda_i}{\lambda_D - \lambda_i} I(\lambda_i<0).
\right)
\end{equation}
Also,
\begin{equation}\label{eq::sig4}
\frac{\partial S_{D-1}}{\partial \lambda_D} =
S_{D-1} \cdot 
\left(\frac{1}{\lambda_D} + \frac{1}{\lambda_D} 
\sum_{i=1}^{D-1} \frac{\lambda_i}{\lambda_D - \lambda_i} I(\lambda_i<0)
\right).
\end{equation}

The last statement follows since,
when $h = o(n^{-1/(4+d)})$,
the squared bias is smaller than the variance.
\end{proof}

\section{Discussion}
\label{sec::discussion}

We have presented a method for finding features of various dimensions.
The method combines density ridge estimation with eigensignatures.
There are several important problems that deserve further research.

First, we do not have a universal, data-driven method for choosing
the tuning parameters.
Although we have presented some heuristics for choosing the tuning
parameters that seem to work well,
it would be nice to have a truly data-driven method.
Of course, the same is true for all structure finding methods that we are aware of.
For example, spectral clustering methods involve numerous tuning parameters
and so far there does not seem to be refined, principled methods for choosing the parameters.
Indeed, even in simple clustering methods like $k$-means, there is no consensus
on what is a good method for choosing $k$.
This is a endemic problem in unsupervised methods.

Second, it will be interesting to investigate the use of our methods
in high dimensions.
In principle, the method can be used in any dimension.
Of course, we expect it will be challenging to get good results in
very high dimensions.
Indeed, it is well known that the minimax rate for estimating densities
degrades quickly with dimension.
However, this does not make the problem hopeless.
In high dimensions, we could lower our expectations and try to find only very prominent features.
Specifically, let $p_h$ be the mean of the density
estimator with bandwidth $h$.
We expect strong features to be present in $p_h$ even if we do not let $h$ tend to 0.
And if we keep $h$ bounded away from 0, then we can estimate $p_h$
at rate $O(1/n)$ independent of dimension.
This suggests that finding singular features may be feasible in high dimensions
although, at this point, this is merely a conjecture.
But, perhaps the biggest challenge here, is finding a way to visualize the results.
It would be intriguing, for example, to find a $d=6$ dimensional 
feature in a $D=20$ dimensional dataset.
But currently we have no way of visualizing a six-dimensional feature.
Once again, this problem is not specific to our method; it is a very general problem
in statistics and machine learning.
Two promising methods for visualizing
higher dimensional features are
persistence diagrams (\cite{chazal2008towards})
which is a two-dimensional plot summarizing the the topological features of a set
and parallel coordinate plots
(\cite{wegman1990hyperdimensional})
which can be used to visualize data in any dimension.

Third, it would be useful to have a formal significance test
to see if a structure is real.
The asymptotic theory in Section 6 is a first step in this direction
but so far we have no formal hypothesis test.

Fourth,
it is not uncommon, when using clustering methods,
to examine their stability properties.
For example, one can make several splits of the data
and compare the clusterings.
A referee has suggested that a similar analysis would be useful here.
We agree that this would be useful
and the stability properties of singular feature finding should be
further investigated.

Finally, the theory supporting the method
is built on the theory in
\cite{chen2015asymptotic}
and \cite{us::2014.ANN}.
But the theory in those papers is somewhat restricted.
Certain regularity assumptions are assumed there that do not seem to be needed.
For example, much of the theory assumes that the ridges do not intersect.
And yet, we have seen in examples that the method works well in practice
even if there are such intersections.
Thus it would be useful to extend the theory to handle these situations.

\bibliography{paper}

\begin{thebibliography}{52}
\providecommand{\natexlab}[1]{#1}
\providecommand{\url}[1]{\texttt{#1}}
\expandafter\ifx\csname urlstyle\endcsname\relax
  \providecommand{\doi}[1]{doi: #1}\else
  \providecommand{\doi}{doi: \begingroup \urlstyle{rm}\Url}\fi

\bibitem[Amiri et~al.(2015)Amiri, Clarke, Clarke, and Koepke]{amiri}
Saeid Amiri, Bertrand Clarke, Jennifer Clarke, and Hoyt~A. Koepke.
\newblock A general hybrid clustering technique.
\newblock \emph{arXiv:1503.01183}, 2015.

\bibitem[Arias-Castro et~al.(2013)Arias-Castro, Mason, and
  Pelletier]{Arias-Castro2013}
Ery Arias-Castro, David Mason, and Bruno Pelletier.
\newblock On the estimation of the gradient lines of a density and the
  consistency of the mean-shift algoithm.
\newblock Technical report, IRMAR, 2013.

\bibitem[Bubenik and Kim(2007)]{bubenik2007statistical}
Peter Bubenik and Peter~T.\ Kim.
\newblock A statistical approach to persistent homology.
\newblock \emph{Homology, Homotopy and Applications}, 9\penalty0 (2):\penalty0
  337--362, 2007.

\bibitem[Cadre(2006)]{cadre2006kernel}
B.~Cadre.
\newblock Kernel estimation of density level sets.
\newblock \emph{Journal of multivariate analysis}, 97\penalty0 (4):\penalty0
  999--1023, 2006.

\bibitem[Carlsson(2009)]{carlsson2009topology}
Gunnar Carlsson.
\newblock Topology and data.
\newblock \emph{Bulletin of the American Mathematical Society}, 46\penalty0
  (2):\penalty0 255, 2009.

\bibitem[Carlsson and Zomorodian(2009)]{carlsson2009theory}
Gunnar Carlsson and Afra Zomorodian.
\newblock The theory of multidimensional persistence.
\newblock \emph{Discrete \& Computational Geometry}, 42\penalty0 (1):\penalty0
  71--93, 2009.

\bibitem[Cautun et~al.(2013)Cautun, van~de Weygaert, and Jones]{NEXUS}
Marius Cautun, Rien van~de Weygaert, and Berard~JT Jones.
\newblock Nexus: Tracing the cosmic web connection.
\newblock \emph{Monthly Notices of the Royal Astronomical Society},
  429.2:\penalty0 1286--1308, 2013.

\bibitem[Chac{\'o}n(2012)]{chacon}
Chac{\'o}n.
\newblock Clusters and water flows: a novel approach to modal clustering
  through morse theory.
\newblock \emph{arXiv preprint arXiv:1212.1384}, 2012.

\bibitem[Chac{\'o}n et~al.(2011)Chac{\'o}n, Duong, and
  Wand]{chacon2011asymptotics}
J.E. Chac{\'o}n, T.~Duong, and MP~Wand.
\newblock Asymptotics for general multivariate kernel density derivative
  estimators.
\newblock \emph{Statistica Sinica}, 21:\penalty0 807--840, 2011.

\bibitem[Chazal et~al.(2011)Chazal, Guibas, Oudot, and
  Skraba]{chazal2011persistence}
F.~Chazal, L.J. Guibas, S.Y. Oudot, and P.~Skraba.
\newblock Persistence-based clustering in {R}iemannian manifolds.
\newblock In \emph{Proceedings of the 27th annual ACM symposium on
  Computational geometry}, pages 97--106. ACM, 2011.

\bibitem[Chazal and Oudot(2008)]{chazal2008towards}
Fr{\'e}d{\'e}ric Chazal and Steve~Y.\ Oudot.
\newblock Towards persistence-based reconstruction in euclidean spaces.
\newblock In \emph{Proceedings of the twenty-fourth annual symposium on
  Computational geometry}, pages 232--241. ACM, 2008.

\bibitem[Chen et~al.(2014)Chen, Genovese, and Wasserman]{Enhanced}
Yen-Chi Chen, Christopher Genovese, and Larry Wasserman.
\newblock Enhanced mode clustering.
\newblock \emph{arXiv:1406.1780}, 2014.

\bibitem[Chen et~al.(2015{\natexlab{a}})Chen, Genovese, and
  Wasserman]{morse-smale}
Yen-Chi Chen, Christopher Genovese, and Larry Wasserman.
\newblock Statistical inference using the {M}orse-{S}male complex.
\newblock \emph{arXiv preprint arXiv:1506:08826}, 2015{\natexlab{a}}.

\bibitem[Chen et~al.(2015{\natexlab{b}})Chen, Genovese, and
  Wasserman]{chen2015asymptotic}
Yen-Chi Chen, Christopher~R Genovese, and Larry Wasserman.
\newblock Asymptotic theory for density ridges.
\newblock \emph{Annals of Statistics}, To appear, 2015{\natexlab{b}}.

\bibitem[Chen et~al.(2015{\natexlab{c}})Chen, Ho, Freeman, Genovese, and
  Wasserman]{cosmicweb}
Yen-Chi Chen, Shirley Ho, Peter Freeman, Chris Genovese, and Larry Wasserman.
\newblock Cosmic web reconstruction through density ridges: Method and
  algorithm.
\newblock \emph{arXiv:1501.05303}, 2015{\natexlab{c}}.

\bibitem[Cohen-Steiner et~al.(2005)Cohen-Steiner, Edelsbrunner, and
  Harer]{cohenSteiner2005stability}
David Cohen-Steiner, Herbert Edelsbrunner, and John Harer.
\newblock Stability of persistence diagrams.
\newblock In \emph{Proc. of the $21^{st}$ Annu. Symp. Comput. Geom.}, pages
  263--271, New York, June 2005. Assoc. of Computing Machinery.
\newblock ISBN 1-58113-991-8.
\newblock \doi{http://doi.acm.org/10.1145/1064092.1064133}.
\newblock Symposium held in Pisa, Italy.

\bibitem[Comaniciu and Meer(2002)]{Comaniciu}
D.~Comaniciu and P.~Meer.
\newblock Mean shift: a robust approach toward feature space analysis.
\newblock \emph{Pattern Analysis and Machine Intelligence, IEEE Transactions
  on}, 24\penalty0 (5):\penalty0 603 --619, may 2002.
\newblock ISSN 0162-8828.
\newblock \doi{10.1109/34.1000236}.

\bibitem[Descoteaux et~al.(2004)Descoteaux, Collins, and
  Siddiqi]{descoteaux2004multi}
Maxime Descoteaux, Louis Collins, and Kaleem Siddiqi.
\newblock A multi-scale geometric flow for segmenting vasculature in {MRI}.
\newblock In \emph{Computer Vision and Mathematical Methods in Medical and
  Biomedical Image Analysis}, pages 169--180. Springer, 2004.

\bibitem[D{\"u}mbgen and Walther(2008)]{dumbgen2008multiscale}
L.~D{\"u}mbgen and G.~Walther.
\newblock Multiscale inference about a density.
\newblock \emph{The Annals of Statistics}, 36\penalty0 (4):\penalty0
  1758--1785, 2008.

\bibitem[Duong et~al.(2008)Duong, Cowling, Koch, and Wand]{duong2008feature}
Tarn Duong, Arianna Cowling, Inge Koch, and MP~Wand.
\newblock Feature significance for multivariate kernel density estimation.
\newblock \emph{Computational Statistics \& Data Analysis}, 52\penalty0
  (9):\penalty0 4225--4242, 2008.

\bibitem[Eberly(1996)]{eberly1996ridges}
David Eberly.
\newblock \emph{Ridges in image and data analysis}, volume~7.
\newblock Springer Science \& Business Media, 1996.

\bibitem[Edelsbrunner and Harer(2008)]{edelsbrunner2008persistent}
Herbert Edelsbrunner and John Harer.
\newblock Persistent homology-a survey.
\newblock \emph{Contemporary mathematics}, 453:\penalty0 257--282, 2008.

\bibitem[Edelsbrunner and Harer(2010)]{edelsbrunner2010computational}
Herbert Edelsbrunner and John Harer.
\newblock \emph{Computational Topology: An Introduction}.
\newblock Amer Mathematical Society, 2010.

\bibitem[Edelsbrunner et~al.(2002)Edelsbrunner, Letscher, and
  Zomorodian]{edelsbrunner2002topological}
Herbert Edelsbrunner, David Letscher, and Afra Zomorodian.
\newblock Topological persistence and simplification.
\newblock \emph{Disc. and Compu. Geom.}, 28\penalty0 (4):\penalty0 511--533,
  July 2002.

\bibitem[Frangi et~al.(1998)Frangi, Niessen, Vincken, and
  Viergever]{frangi1998multiscale}
Alejandro~F Frangi, Wiro~J Niessen, Koen~L Vincken, and Max~A Viergever.
\newblock Multiscale vessel enhancement filtering.
\newblock In \emph{Medical Image Computing and Computer-Assisted
  Interventation-MICCAI-98}, pages 130--137. Springer, 1998.

\bibitem[Frosini(1992)]{frosini1992discrete}
Patrizio Frosini.
\newblock Discrete computation of size functions.
\newblock \emph{Journal of Combinatorics, Information and System Sciences},
  17\penalty0 (3-4):\penalty0 232--250, 1992.

\bibitem[Fukunaga and Hostetler(1975)]{Fukunaga}
Keinosuke Fukunaga and Larry~D. Hostetler.
\newblock The estimation of the gradient of a density function, with
  applications in pattern recognition.
\newblock \emph{IEEE Transactions on Information Theory}, 21:\penalty0 32--40,
  1975.

\bibitem[Genovese et~al.(2015)Genovese, Perone-Pacifico, Verdinelli, and
  Wasserman]{us2015modes}
Christopher Genovese, Marco Perone-Pacifico, Isabella Verdinelli, and Larry
  Wasserman.
\newblock Nonparametric inference for density modes.
\newblock \emph{Journal of the Royal Statistical Society B.}, 2015.

\bibitem[Genovese et~al.(2014)Genovese, Perone-Pacifico, Verdinelli, and
  Wasserman]{us::2014.ANN}
Christopher~R. Genovese, Marco Perone-Pacifico, Isabella Verdinelli, and Larry
  Wasserman.
\newblock Nonparametric ridge estimtors.
\newblock \emph{The Annals of Statistics}, 42\penalty0 (4):\penalty0
  1511--1545, 2014.

\bibitem[Ghrist(2008)]{ghrist2008barcodes}
Robert Ghrist.
\newblock Barcodes: the persistent topology of data.
\newblock \emph{Bulletin-American Mathematical Society}, 45\penalty0
  (1):\penalty0 61, 2008.

\bibitem[Gin{\'e} and Guillou(2002)]{gine2002rates}
E.~Gin{\'e} and A.~Guillou.
\newblock Rates of strong uniform consistency for multivariate kernel density
  estimators.
\newblock In \emph{Annales de l'Institut Henri Poincare (B) Probability and
  Statistics}, volume~38, pages 907--921. Elsevier, 2002.

\bibitem[Godtliebsen et~al.(2002)Godtliebsen, Marron, and
  Chaudhuri]{godtliebsen2002significance}
F~Godtliebsen, JS~Marron, and Probal Chaudhuri.
\newblock Significance in scale space for bivariate density estimation.
\newblock \emph{Journal of Computational and Graphical Statistics}, 11\penalty0
  (1):\penalty0 1--21, 2002.

\bibitem[Icke and van~de Weygaert(1991)]{icke1991galaxy}
Vincent Icke and Rien van~de Weygaert.
\newblock The galaxy distribution as a voronoi foam.
\newblock \emph{Quarterly Journal of the Royal Astronomical Society},
  32:\penalty0 85--112, 1991.

\bibitem[Irwin(1980)]{irwin1980smooth}
M.C. Irwin.
\newblock \emph{Smooth dynamical systems}, volume~94.
\newblock Academic Press, 1980.

\bibitem[Karypis et~al.(1999)Karypis, Han, and Kumar]{karypis1999chameleon}
George Karypis, Eui-Hong Han, and Vipin Kumar.
\newblock Chameleon: Hierarchical clustering using dynamic modeling.
\newblock \emph{Computer}, 32\penalty0 (8):\penalty0 68--75, 1999.

\bibitem[Klemela(2009)]{klemela2009smoothing}
J.~Klemela.
\newblock \emph{Smoothing of multivariate data: density estimation and
  visualization}, volume 737.
\newblock Wiley, 2009.

\bibitem[Li et~al.(2007)Li, Ray, and Lindsay]{li2007nonparametric}
J.~Li, S.~Ray, and B.G. Lindsay.
\newblock A nonparametric statistical approach to clustering via mode
  identification.
\newblock \emph{Journal of Machine Learning Research}, 8\penalty0 (8):\penalty0
  1687--1723, 2007.

\bibitem[Magnus and Neudecker(1988)]{magnus1988matrix}
X.~Magnus and H.~Neudecker.
\newblock Matrix differential calculus.
\newblock \emph{New York}, 1988.

\bibitem[Ozertem and Erdogmus(2011)]{Principal}
Ozertem and Erdogmus.
\newblock Locally defined principal curves and surfaces.
\newblock \emph{Journal of Machine Learning Research}, 12:\penalty0 1249--1286,
  2011.

\bibitem[Polonik(1995)]{polonik1995measuring}
W.~Polonik.
\newblock Measuring mass concentrations and estimating density contour
  clusters-an excess mass approach.
\newblock \emph{The Annals of Statistics}, pages 855--881, 1995.

\bibitem[Qiao and Polonik(2014)]{qiao2014theoretical}
Wanli Qiao and Wolfgang Polonik.
\newblock Theoretical analysis of nonparametric filament estimation, 2014.

\bibitem[Roweis and Saul(2000)]{roweis2000nonlinear}
Sam~T Roweis and Lawrence~K Saul.
\newblock Nonlinear dimensionality reduction by locally linear embedding.
\newblock \emph{Science}, 290\penalty0 (5500):\penalty0 2323--2326, 2000.

\bibitem[Scott(1992)]{scott1992multivariate}
Scott.
\newblock \emph{Multivariate density estimation. Theory, practice and
  visualization, first edition}.
\newblock New York: John Wiley \& Sons., 1992.

\bibitem[Scott(2015)]{scott2015multivariate}
Scott.
\newblock \emph{Multivariate density estimation: theory, practice, and
  visualization. Second edition}.
\newblock John Wiley \& Sons, 2015.

\bibitem[Snedden et~al.(2014)Snedden, Phillips, Mathews, Coughlin, Suh, and
  Bhattacharya]{snedden2014new}
Ali Snedden, Lara~Arielle Phillips, Grant~J Mathews, Jared Coughlin, In-Saeng
  Suh, and Aparna Bhattacharya.
\newblock A new multi-scale structure finding algorithm to identify
  cosmological structure.
\newblock \emph{arXiv preprint arXiv:1409.7711}, 2014.

\bibitem[Tenenbaum et~al.(2000)Tenenbaum, De~Silva, and
  Langford]{tenenbaum2000global}
Joshua~B Tenenbaum, Vin De~Silva, and John~C Langford.
\newblock A global geometric framework for nonlinear dimensionality reduction.
\newblock \emph{Science}, 290\penalty0 (5500):\penalty0 2319--2323, 2000.

\bibitem[Turner et~al.(2012)Turner, Mileyko, Mukherjee, and
  Harer]{turner2012fr}
Katharine Turner, Yuriy Mileyko, Sayan Mukherjee, and John Harer.
\newblock Frechet means for distributions of persistence diagrams.
\newblock \emph{arXiv preprint arXiv:1206.2790}, 2012.

\bibitem[van~de Weygaert(1994)]{van1994fragmenting}
Rien van~de Weygaert.
\newblock Fragmenting the universe. 3: The constructions and statistics of 3-d
  voronoi tessellations.
\newblock \emph{Astronomy and Astrophysics}, 283:\penalty0 361--406, 1994.

\bibitem[van~de Weygaert and Schaap(2009)]{van2009cosmic}
Rien van~de Weygaert and Willem Schaap.
\newblock The cosmic web: geometric analysis.
\newblock In \emph{Data Analysis in Cosmology}, pages 291--413. Springer, 2009.

\bibitem[Walther(1997)]{walther1997granulometric}
G.~Walther.
\newblock Granulometric smoothing.
\newblock \emph{The Annals of Statistics}, pages 2273--2299, 1997.

\bibitem[Wegman(1990)]{wegman1990hyperdimensional}
Edward~J Wegman.
\newblock Hyperdimensional data analysis using parallel coordinates.
\newblock \emph{Journal of the American Statistical Association}, 85\penalty0
  (411):\penalty0 664--675, 1990.

\bibitem[Zhong and Ghosh(2003)]{zhong2003unified}
Shi Zhong and Joydeep Ghosh.
\newblock A unified framework for model-based clustering.
\newblock \emph{The Journal of Machine Learning Research}, 4:\penalty0
  1001--1037, 2003.

\end{thebibliography}

\end{document}